\documentclass[11pt, titlepage, reqno]{amsart}

\usepackage[utf8]{inputenc}
\usepackage[T1]{fontenc}
\usepackage[british]{babel}
\usepackage{times}
\usepackage{amsaddr, amsbsy, amscd, amsfonts, amsmath, mathrsfs, amssymb, amsthm}
\usepackage{dsfont, mathrsfs, mathtools}
\usepackage{hyperref}
\usepackage{enumitem}
\usepackage{xparse}
\usepackage[bottom, hang, flushmargin]{footmisc}
\usepackage{framed, xcolor}

\usepackage[a4paper, hmargin=2cm, vmargin={1.5cm,2.5cm}, nomarginpar, hcentering, includeheadfoot]{geometry}
\usepackage{eqnarray}
\usepackage{scalerel, stackengine}
\usepackage{enumitem}
\usepackage{fancyhdr}
\usepackage{indentfirst}

\usepackage[toc,page]{appendix}
\theoremstyle{definition}

\theoremstyle{plain}
\newtheorem{note}{Remark}
\newtheorem{theo}{Theorem}[section]
\newtheorem{lemma}[theo]{Lemma}
\newtheorem{prop}[theo]{Proposition}
\newtheorem{cor}[theo]{Corollary}

\numberwithin{equation}{section}
\numberwithin{defn}{section}
\numberwithin{note}{section}

\stackMath

\let\oldker\ker
\let\oldforall\forall
\let\oldexists\exists
\renewcommand\widehat[1]{%
\savestack{\tmpbox}{\stretchto{%
    \scaleto{%
        \scalerel*[\widthof{\ensuremath{#1}}]{\kern.1pt\mathchar"0362\kern.1pt}%
        {\rule{0ex}{\textheight}}
    }{\textheight}%
}{2.4ex}}%
\stackon[-6.9pt]{#1}{\tmpbox}%
}

\NewDocumentCommand \fromToParser {m m m}{%
    \IfValueTF{#3}{%
        \fromTo{#1}{#2}[#3]%
    }{%
        \fromTo{#1}{#2}%
    }%
}
\NewDocumentCommand \fromTo {m m o}{%
    #1\! \IfValueTF{#3}{%
            \to[#3]
        }{%
            \to[\!\quad\!]
        } \!#2%
}
\newcommand{\integrand}[2]{\!#2\; #1}
\NewDocumentCommand \integralLimits {m m}{%
    _{#1}\IfValueT{#2}{^{\,#2}\!}%
}
\newcommand{\defineSymbol}[2]{#1 \vcentcolon = #2}
\newcommand{\enifedSymbol}[2]{#1 = \vcentcolon #2}

\RenewDocumentCommand \newline {o}{%
\hfill\\[\IfValueTF{#1}{#1}{0} pt]
}
\renewcommand{\-}{\mspace{-1.5mu}}
\newcommand{\+}{\mspace{1.5mu}}
\newcommand{\nquad}{\mspace{-18mu}}
\newcommand{\nqquad}{\mspace{-36mu}}
\newcommand{\n}{\noindent}
\newcommand{\vs}{\vspace{0.5cm}}

\providecommand{\comment}[1]{}

\newcommand{\ie}{\emph{i.e.}}
\newcommand{\eg}{\emph{e.g.}}

\renewcommand{\Im}{\mathrm{Im}}
\renewcommand{\Re}{\mathrm{Re}}
\newcommand{\pInfty}{ {\scriptstyle +}\+\infty}
\newcommand{\mInfty}{ {\scriptstyle -}\+\infty}
\newcommand{\Id}{\mathds{1}}
\renewcommand{\forall}{\oldforall\,}
\RenewDocumentCommand \exists {s}{%
    \IfBooleanTF{#1}{\oldexists!}{\oldexists} \;%
}
\ProvideDocumentCommand \define {s >{\SplitArgument{1}{;}} m}{%
    \IfBooleanTF{#1}{\enifedSymbol #2}{\defineSymbol #2}%
}
\RenewDocumentCommand \to {o}{%
    \IfValueTF{#1}{%
        \xrightarrow[\,#1\,]{\;}%
    }{%
        \,\rightarrow\,%
    }%
}
\NewDocumentCommand \maps {m >{\SplitArgument{2}{;}} m}{%
    #1\!: \fromToParser #2%
}
\ProvideDocumentCommand \conjugate {s m}{%
    \IfBooleanTF{#1}{%
        \overline{#2}%
    }{%
        \bar{#2}%
    }%
}
\providecommand{\abs}[1]{\lvert#1\rvert}
\renewcommand*{\vec}[1]{\boldsymbol{#1}}

\NewDocumentCommand \ball {s m m}{%
    \IfBooleanTF{#1}{%
        \mathcal{B}^{\+ c}_{#2}(#3)%
    }{%
        \mathcal{B}_{#2}(#3)%
    }%
}
\NewDocumentCommand \Char { m o }{
    \Id_{#1}
    \IfValueT{#2}{( #2 )}
}
\newcommand*{\oSmall}[1]{o\!\left(#1\right)}
\newcommand*{\oBig}[1]{\mathcal{O}\!\left(#1\right)}
\NewDocumentCommand \integrate { >{\SplitArgument{1}{;}} o >{\SplitArgument{1}{;}} m}{%
    \int \IfValueT{#1}{\integralLimits #1 \!} \integrand #2%
}

\NewDocumentCommand \hilbert{s}{%
    \IfBooleanTF{#1}{\mathfrak{H}}{\mathscr{H}}%
}
\NewDocumentCommand \X{s}{%
    \IfBooleanTF{#1}{\mathcal{X}}{\mathfrak{X}}%
}
\NewDocumentCommand \scalar { m m o }{%
    \langle#1,\,#2\rangle\IfValueT{#3}{_{#3}}%
}
\NewDocumentCommand \norm { m o }{%
    \left\lVert#1\right\rVert \IfValueT{#2}{_{#2}}%
}
\NewDocumentCommand \normConverge {o m m o}{
    \IfValueTF{#4}{
        \fromTo{ \IfValueTF{#1}{\norm{#2 - #3}[#1]\! }{\norm{#2 - #3}} }{\,0}[#4]
    }{
        \fromTo{ \IfValueTF{#1}{\norm{#2 - #3}[#1]\! }{\norm{#2 - #3}} }{\,0}
    }
}
\NewDocumentCommand \weakConverge {m m o}{%
    #1 \!\IfValueTF{#3}{%
            \xrightharpoonup[#3]{\quad}
        }{%
            \xrightharpoonup{\quad}%
        } \- #2
}

\NewDocumentCommand \wlim {m} {
    \mathop{\mathrm{w\!-\!lim}}\limits_{#1}\:
}

\newcommand*{\dom}[1]{\mathscr{D}(#1)}
\newcommand*{\ran}[1]{\mathrm{ran}(#1)}
\renewcommand*{\ker}[1]{\oldker(#1)}
\newcommand*{\adj}[1]{{#1}^{\ast}}
\newcommand*{\linear}[1]{\mathscr{L}\!\left(#1\right)}
\newcommand*{\bounded}[1]{\mathscr{B}\mspace{-1mu}\left(#1\right)}
\NewDocumentCommand \resolvent {m o}{%
    \mathcal{R}_{#1}\IfValueT{#2}{(#2)}%
}
\NewDocumentCommand \spectrum {o m}{%
    \IfValueTF{#1}{%
        \sigma_{\mathrm{#1}}(#2)%
    }{%
        \sigma(#2)%
    }%
}

\NewDocumentCommand \Lp {s m o} { L^{\- #2}\IfBooleanT{#1}{_{\,\mathrm{loc}}}\IfValueT{#3}{(#3)} }
\NewDocumentCommand \LpS {s m o}{
\IfBooleanTF{#1}{\Lp{#2}_{+}}{\Lp{#2}_{\mathrm{sym}}}\IfValueT{#3}{(#3)}
}
\NewDocumentCommand \LpA {s m o}{
\IfBooleanTF{#1}{\Lp{#2}_{-}}{\Lp{#2}_{\mathrm{asym}}}\IfValueT{#3}{(#3)}
}
\NewDocumentCommand \LpSA {m o}{
\Lp{#1}_{\pm}\IfValueT{#2}{(#2)}
}

\NewDocumentCommand \FT { s m o }{
    \IfBooleanTF{#1}{ \IfValueTF{#3}{ (\mathcal{F} \+  #2)\-(#3) }{%
    \mathcal{F} \+  #2 } }{ \hat{#2}\IfValueT{#3}{(#3)} }
}

\newcommand{\C}{\mathbb{C}}
\newcommand{\R}{\mathbb{R}}
\newcommand{\Rplus}{\mathbb{R}_{+}}
\newcommand{\N}{\mathbb{N}}

\NewDocumentCommand \Z {o}{%
    \IfValueTF{#1}{%
        \mathbb{Z}_{\,#1}%
    }{%
        \mathbb{Z}%
    }%
}

\newcommand{\be}{\begin{equation}}
\newcommand{\ee}{\end{equation}}

\title[Hamiltonian for a Bose gas with Contact Interactions]{Hamiltonian for a Bose gas with Contact Interactions}

\author[D. Ferretti]{Daniele Ferretti}
\address{\emph{Gran Sasso Science Institute}, Via Michele Iacobucci, 2 - 67100 (AQ), Italy\\ \textsf{daniele.ferretti@gssi.it}}

\author[A. Teta]{Alessandro Teta}
\address{\emph{Sapienza Universit\`a di Roma}, Piazzale Aldo Moro, 5 - 00185 (RM), Italy\\ \textsf{teta@mat.uniroma1.it}}
\date{}

\thanks{The authors acknowledge the support of the GNFM Gruppo Nazionale per la Fisica Matematica - INdAM}

\begin{document}

\makeatletter

\newcommand{\pushright}[1]{\ifmeasuring@#1\else\omit\hfill$\displaystyle#1$\fi\ignorespaces}
\newcommand{\pushleft}[1]{\ifmeasuring@#1\else\omit$\displaystyle#1$\hfill\fi\ignorespaces}

\renewcommand\section{%
  \@startsection{section}%
    {1}
    {0em}
    {1.5cm \@plus 0.1ex \@minus -0.05ex}
    {0.75cm \@plus 0.2em}
    {\centering \large \scshape}
  }

  \renewcommand\subsection{%
  \@startsection{subsection}%
    {2}
    {0em}
    {0.75cm \@plus 0.1ex \@minus -0.05ex}
    {0.25cm \@plus 0.2em}
    {\bf\scshape}
  }

\makeatother

\begin{abstract}
We study the Hamiltonian for a three-dimensional Bose gas of $N \geq 3$ spinless particles interacting via zero-range (also known as contact) interactions.
Such interactions are encoded by (singular) boundary conditions imposed on the coincidence hyperplanes, \ie,~when the coordinates of two particles coincide.  

\n
It is well known that imposing the same kind of boundary conditions as in the two-body problem with a point interaction leads to a Hamiltonian unbounded from below (and thus unstable).
This is due to the fact that the interaction becomes overly strong and attractive when the coordinates of three or more particles coincide.

\n In order to avoid such instability, we develop a suggestion originally formulated by Minlos and Faddeev in 1962, introducing slightly modified boundary conditions that weaken the strength of the interaction between two particles $i$ and $j$ in two scenarios: (a) a third particle approaches the common position of $i$ and $j$; (b) another distinct pair of particles approach each other.
In all other cases, the usual boundary condition is restored.

\n Using a quadratic form approach, we construct a class of Hamiltonians characterized by such modified boundary conditions, that are self-adjoint and bounded from below.
We also compare our approach with the one developed years ago by Albeverio, H{\o}egh-Krohn and Streit using the theory of Dirichlet forms (J. Math. Phys., 18, 907--917, 1977).
In particular, we show that the $N$-body Hamiltonian defined by Albeverio \emph{et al.} is a special case of our class of Hamiltonians.
Furthermore, we also introduce a Dirichlet form by considering a more general weight function, and we prove that the corresponding $N$-body Hamiltonians essentially coincide with those constructed via our method.\newline[10]
\begin{footnotesize}

\n \emph{MSC 2020: 
    81Q10; 
    81Q15; 
    46N50; 
    81V70. 
}  
\end{footnotesize}
\end{abstract}
\maketitle



\section{Introduction}\label{intro}

In non relativistic Quantum Mechanics it is often useful to model the interactions between the particles by means of so-called \emph{$\delta$-potentials}, also known as \emph{zero-range} or \emph{contact interactions}.
In particular, this happens when studying systems in the low-energy regime, so that the wavelength associated with the particles is much larger than the range of the interaction between the particles.
Under these conditions, it may be reasonable to consider simplified Hamiltonians with contact interactions.
Such Hamiltonians have often been used at a formal level in physics in both the case of one-particle systems and many-body systems.
As examples, we can mention the Kronig-Penney model~(\cite{KP}), the neutron-proton system~(\cite{BP}), the scattering of neutrons~(\cite{F}), the Efimov effect~(\cite{E}) and the derivation of the Lee-Huang-Yang formula for a dilute Bose gas~(\cite{LHY}). 

\n
The variety of interesting physical applications provides good motivation for the mathematical analysis of these Hamiltonians. 
From the mathematical point of view, the problem is to construct the rigorous counterpart of the formal Hamiltonian with contact interactions as a self-adjoint (s.a.) and, possibly, lower-bounded operator in a suitable Hilbert space.

\n
The first step is to provide a precise mathematical definition.
Let us consider an interaction supported by a (sufficiently regular) set $\Sigma \subset \R^s$, $s\geq 1$, with $\mathrm{\mathop{codim}}\,\Sigma\leq 3$ and zero Lebesgue measure, and let us introduce the symmetric, but not s.a. operator $\; \dot{H}_0 = - \Delta$, $\, \dom{\dot{H}_0} = \hilbert \cap H^2_0(\R^s \setminus \Sigma)$, where $\hilbert$ is the proper $L^2$-space associated to the quantum system.
By definition, the Schr\"odinger operator $-\Delta + \delta_{\Sigma}$ is a non-trivial s.a. extension of $\, \dot{H}_0, \, \dom{\dot{H}_0}$.

\n
The second step is the mathematical problem of constructing such s.a. extensions.
It turns out that, roughly speaking, each extension is the Laplacian in $\R^s$ with specific, possibly singular, boundary conditions on $\Sigma$.

\n
For the one-body case in $\R^d$, $d \geq 1$, let us assume that the set $\Sigma$ is a finite set of points (recall that the case of a single point corresponds to the two-body problem read in the relative coordinate).
Then the deficiency subspaces of $\,\dot{H}_0, \, \dom{\dot{H}_0}$ are finite-dimensional, and therefore all the s.a. extensions can be explicitly constructed and a complete spectral analysis can be performed~(\cite{Albeverio}).
Each s.a. extension is the Laplacian in $\R^d$ with suitable boundary conditions satisfied at the fixed points where the interaction is located.
 
\n
In the $N$-body case in $\R^d$, with $N\!\geq\- 3$ and  $d \geq \-1$, consider the set $\Sigma$ as the union of the coincidence hyperplanes $\pi_{ij}\vcentcolon= \left\{(\vec{x}_1,\ldots,\vec{x}_N)\in\R^{dN} \,\big| \; \vec{x}_i\-=\vec{x}_j\right\}$, where  $i,j\-\in\-\{1, \ldots, N\}$, $i \neq j$,  and $\vec{x}_i$ denotes the  coordinate of the $i$-th particle in $\R^d$.
In this case, the deficiency subspaces of $\dot{H}_0, \, \dom{\dot{H}_0}$ are infinite-dimensional and it is not a priori clear how to choose a specific boundary condition on each hyperplane in order to have a physically reasonable s.a. extension.

\n
For $d\!=\!1$, the problem is relatively simple because the codimension of each  hyperplane $\pi_{ij} $ is equal to one, and therefore one can construct the quadratic form of the extension using perturbation theory.
The resulting Hamiltonian is characterized by two-body boundary conditions on each hyperplane, which are a natural generalization of the boundary condition in the two-body case in dimension one.
Furthermore, the Hamiltonian can be obtained as the norm-resolvent limit of approximating Hamiltonians with smooth, two-body, rescaled potentials (see \eg,~\cite{BCFT1}, \cite{GHL} for recent contributions).\newline
For $d\in\-\{2,3\}$, the codimension of $\pi_{ij}$ is larger than $1$, and therefore perturbation theory does not work.

\n
For $d\!=\!2$ one can proceed defining the relevant extension as the Laplacian in $\R^{2N}$ with two-body boundary conditions on each hyperplane 
which are the direct generalization of the boundary condition in the two-body problem.
It turns out that such operator is s.a. and bounded from below~(\cite{DFT}, \cite{DR}).
It has also been recently proven that the Hamiltonian is the norm-resolvent limit of approximating Hamiltonians with smooth, two-body, rescaled potentials and with a suitable renormalization of the coupling constant~(\cite{GH}).

\n
It should be stressed once again that in dimension one and two a physical interesting $N$-body Hamiltonian with contact interactions can be obtained by following the analogy with the two-body problem.
This essentially means that the boundary condition on the hyperplane $\pi_{ij}$ does not depend on the positions of the remaining particles $\vec{x}_k, \+ k \notin \-\{i,j\}$.
In other words, the strength of the interaction between two particles is constant, regardless of the position of the other particles.
In this sense, one can define a genuine two-body contact, or zero-range interaction, for the $N$-body problem.

\n
For $d\!=\!3$ the situation is more subtle, and following the analogy with the two-body problem, one arrives at a result which is unsatisfactory from the physical point of view.
The aim of this paper is precisely to propose the mathematical construction of a different physically reasonable Hamiltonian with contact interactions for a Bose gas in dimension three\footnote{We just recall that the situation is rather different for systems made of different species of fermions, see \eg,~\cite{CDFMT, CDFMT1}, \cite{MS0, MS, MS1}}. 
Let us explain the difficulties and the main idea of our construction. 

\n
We consider a system of $N\-\geq 3$, identical spinless bosons of mass $\frac{1}{2}$ interacting with each other via zero-range interactions.
The Hilbert space of the system is therefore given by
\begin{equation}
    \define{\hilbert_N;\LpS{2}[\R^{3N}]}.
\end{equation}
We denote by $\mathcal{P}_N$ the following set
\begin{equation}
    \mathcal{P}_N\vcentcolon=\left\{\{i,j\}\,\big|\;i,j\in\{1,\ldots,N\},\+i\neq j\right\},\qquad |\mathcal{P}_N|=\frac{N(N\- - \- 1)}{2}
\end{equation}
so that, at least formally, the Hamiltonian is
\begin{equation}\label{formalH}
    \tilde{\mathcal{H}}=\mathcal{H}_0+\nu\nquad\sum_{\{i,\+j\}\,\in\:\mathcal{P}_N}^N\nquad\delta(\vec{x}_i-\vec{x}_j),
\end{equation}
where $\nu$ is a coupling constant and $\mathcal{H}_0$, with domain $\hilbert_N\cap H^2(\R^{3N})$, is the free Hamiltonian, \ie
\begin{align}\label{freeH}
    &\define{\mathcal{H}_0;-\sum_{i=1}^N\Delta_{\vec{x}_i}}\+.
\end{align}
However, one can argue that $\tilde{\mathcal{H}}$ is not a well-defined linear operator in $\hilbert_N$.
Hence, the rigorous version of~\eqref{formalH}, according to the above general definition, is provided by any non-trivial s.a. extension $\mathcal H$ of the operator
\begin{equation}\label{symmetricHamiltonianToBeExtended}
    \dot{\mathcal{H}}_0\vcentcolon= \mathcal{H}_0\upharpoonright\dom{\dot{\mathcal{H}}_0}\+, \qquad \dom{\dot{\mathcal{H}}_0}\vcentcolon=\hilbert_N\cap H^2_0(\R^{3N}\setminus \pi)
\end{equation}
where $\pi$ is the union of the coincidence hyperplanes
\begin{equation}\label{incidenceHyperplanes}
    \define{\pi;\bigcup_{\sigma\+\in\+\mathcal{P}_N}\pi_\sigma}\,, \qquad \pi_{ij}= \left\{(\vec{x}_1,\ldots,\vec{x}_N)\in\R^{3N} \,\big| \; \vec{x}_i\-=\vec{x}_j\right\}\-.
\end{equation}
Notice that~\eqref{symmetricHamiltonianToBeExtended} is symmetric and closed according to the graph norm of $\mathcal{H}_0\+$.
Moreover, $\mathcal H$ must obviously satisfy the condition 
\begin{equation}\label{h=h0}
    \mathcal{H}\psi = \mathcal{H}_0\+\psi,\qquad\forall\psi \in \hilbert_N\cap H^2_0(\R^{3N}\-\smallsetminus\pi).
\end{equation}
Technically, $\mathcal{H}$ is said to be a singular perturbation of $\mathcal{H}_0\+$ supported on $\pi$, and $\dom{\mathcal{H}}$ is a proper subset of $ \hilbert_N\-\cap H^2(\R^{3N}\setminus \pi)$, by construction.

\n
The problem is now the construction of a s.a. extension $\mathcal{H}$.
If one proceeds in analogy with the two-body problem, as it can be done in dimension one and two, a class of extensions of $\dot{\mathcal{H}}_0\+$  is obtained by imposing to a vector $\psi\-\in\-\hilbert_N\-\cap H^2(\R^{3N}\setminus \pi)$ the following singular, two-body boundary condition for each hyperplane $\pi_{\sigma}$, with $\sigma=\{i,j\}\in\mathcal{P}_N$ 
\begin{equation}\label{stmBC}
    \begin{split}
    \psi(\vec{x}_1,\ldots,\vec{x}_N)=&\;\frac{\xi\!\left(\frac{\vec{x}_i+\+\vec{x}_j}{2},\vec{x}_1,\ldots\check{\vec{x}}_\sigma\ldots,\vec{x}_N\right)\!}{\abs{\vec{x}_i-\vec{x}_j}}\,+\\
    &+\-\alpha_0\,\xi\!\left(\tfrac{\vec{x}_i+\+\vec{x}_j}{2},\vec{x}_1,\ldots\check{\vec{x}}_\sigma\ldots,\vec{x}_N\right)\-+\oSmall{1}\-,\quad \text{for }\,\abs{\vec{x}_i\--\vec{x}_j}\to[\!\quad\!] 0\+,
    \end{split}
\end{equation}
in some topology for some $\xi\-\in\-\Lp{2}[\R^3]\-\otimes\-\LpS{2}[\R^{3(N-2)}]$, where $\alpha_0$ is a real parameter and $\check{\vec{x}}_{\sigma} $ 
denotes the omission of the variables $\vec{x}_i$ and $\vec{x}_j\+$.
We stress that~\eqref{stmBC} is a two-body boundary condition along $\pi_{\sigma}$ that does not depend on the positions of the other particles $\vec{x}_k,\+k \notin\-\{i,j\}$ (such positions play only the role of parameters).
In this sense, it is a natural generalization of the boundary condition required in the two-body problem~(\cite{Albeverio}).
The extension of $\dot{\mathcal{H}}_0\+$ characterized by~\eqref{stmBC} is known as \emph{Ter-Martirosyan Skornyakov extension}.
We point out that the parameter $\alpha_0$ is related to the two-body scattering length $\mathfrak{a}$ via the relation
\begin{equation}\label{2BodyScatteringLenght}
    \mathfrak{a} = -\frac{1}{\+\alpha_0\+}\+.
\end{equation}
In particular, by~\eqref{2BodyScatteringLenght}, one can see that the strength of the two-body interaction vanishes as $\abs{\alpha_0} \!\to[\!\quad\!]\- \pInfty$.

\n
As a matter of fact, the Ter-Martirosyan Skornyakov extension is symmetric but not s.a. already in the case $N\!=\-3$.
The deficiency indices are $(1,1)$, then its s.a. extensions are parametrized by a real constant that characterizes the behavior of the wave function close to the triple coincidence point, \ie~where the positions of the three particles coincide.
It also turns out that all these s.a. extensions are unbounded from below~(\cite{MF}, see also \cite{MF2}).
This instability is due to the fact that the interaction becomes too strong when the three particles are close to each other and this determines a collapse (or fall to the center) phenomenon, known in physics literature as the Thomas effect.

\comment{
\n It is important to stress that the result in~\cite{MF} should be interpreted as a sort of no-go theorem.
Indeed, it states that in dimension three for $N\!=\-3$, and a fortiori for $N\!>\-3$, one cannot define a Hamiltonian with contact interactions via the two-body boundary condition~\eqref{stmBC}, where the strength of the interaction between two particles is constant, and therefore the interaction is a genuine two-body interaction.
Conversely, one is forced to add a further three-body boundary condition, corresponding to a sort of three-body force acting between the particles when they are close to each other.
Furthermore, the fact that the s.a. Hamiltonian constructed in~\cite{MF} is not bounded from below, and then it is unsatisfactory from the physical point of view, simply means that the Ter-Martirosyan Skornyakov extension characterized by~\eqref{stmBC} is not a good starting point to construct the Hamiltonian. 
}

\n
This fact is already made clear in~\cite{MF}, where a suggestion to modify the boundary condition~\eqref{stmBC} in order to obtain a lower-bounded Hamiltonian is also given.
An analogous approach is proposed in~\cite{AHKW}.
Unfortunately, details of the construction are omitted in both cases.
Recently, the suggestion has been reconsidered~(\cite{FT}, \cite[section 9]{Miche}, \cite{BCFT}, \cite[section 6]{GM}, \cite{FeT2}) and a lower-bounded Hamiltonian in the three-body problem has been explicitly constructed (see also~\cite{FeT} for the case of a Bose gas interacting with an impurity).
The idea is to replace the constant $\alpha_0$ in~\eqref{stmBC} with $\alpha_0$ plus a real function defined on $\pi_{\sigma}$.
This function must diverge (with a singularity of Coulomb type and a coupling constant sufficiently large) when a particle approaches the common position of the other two particles, whereas it reduces to zero as soon as the particle is sufficiently far away.
The function can be chosen with an arbitrarily small compact support and the Coulomb singularity is the minimal singularity required in order to prevent the collapse when the positions of the three particles coincide.
In this sense, a minimal modification of the boundary condition~\eqref{stmBC} is introduced to obtain an energetically stable Hamiltonian for the three-body problem. 

\n
Notice that, in this way, one replaces the two-body boundary condition~\eqref{stmBC}, which is independent of the position of the third particle, with a new two-body boundary condition (slightly) dependent on the position of the third particle. 
Roughly speaking, this corresponds to the introduction of a sort of three-body interaction such that the (effective) scattering length decreases to zero when the positions of two particles coincide and the third particle is close to the common position of the first two, while it reduces to a constant when the third particle is far away. 

\n
Following the same line of thought, in this paper we consider the general $N$-body problem for a Bose gas.
For $N>3$, in addition to the three-body interaction required in the three-body problem, we also introduce a further four-body interaction.
More precisely, let us define the function 
\begin{equation}\label{alfa}
    \begin{split}
        &\nqquad\maps{\alpha}{\R^3\otimes\R^{3(N-2)};\R},\\[-2.5pt]
    (\vec{z},\vec{y}_1,\ldots,\mspace{2.25mu}\vec{y}_{N-2})\longmapsto &\,\alpha_0+\gamma\-\sum_{k=1}^{N-2}\frac{\theta(\abs{\vec{y}_k\--\vec{z}})}{\abs{\vec{y}_k\--\vec{z}}}+\frac{\gamma}{2}\!\sum_{1\+\leq\+ k\,<\,\ell\+\leq\+ N-2}\!\!\!\frac{\theta(\abs{\vec{y}_k\--\vec{y}_\ell})}{\abs{\vec{y}_k\--\vec{y}_\ell}}
    \end{split}
\end{equation}
with $\gamma>0$ and $\maps{\theta}{\Rplus;\R}$ an essentially bounded function satisfying
\begin{equation}
    \label{positiveBoundedCondition}
    1-\tfrac{r}{b}\leq\theta(r)\leq 1+\tfrac{r}{b},\qquad\text{for almost every } \+r\in\Rplus\+ \text{ and some }\, b>0.
\end{equation}
We observe that the function $\theta$, by assumption~\eqref{positiveBoundedCondition}, is positive almost everywhere in a neighborhood of the origin, and $\lim\limits_{r\to 0^+}\-\theta(r)=1$.
Notice that simple choices for this function are $\theta(r)= e^{-r/b}$ or $\theta(r)= \Char{b}(r)$, where $\Char{b}$ is the characteristic function of the ball of radius $b$ centered at the origin.
One can also choose a smooth $\theta$ with an arbitrarily small compact support. 

\n
Making use of~\eqref{alfa}, we now define the following boundary condition on $\pi_{\sigma}$
\begin{equation}\label{mfBC}
\begin{split}
    \psi(\vec{x}_1,\ldots,\vec{x}_N)=&\;\frac{\xi\!\left(\frac{\vec{x}_i+\+\vec{x}_j}{2},\vec{x}_1,\ldots\check{\vec{x}}_\sigma\ldots,\vec{x}_N\right)\!}{\abs{\vec{x}_i-\vec{x}_j}}\,+\\
    &+\-(\Gamma_{\!\mathrm{reg}}^{\+\sigma} \+\xi)\!\left(\tfrac{\vec{x}_i+\+\vec{x}_j}{2},\vec{x}_1,\ldots\check{\vec{x}}_\sigma\ldots,\vec{x}_N\right)\-+\oSmall{1}\-,\quad \text{for }\,\abs{\vec{x}_i\--\vec{x}_j}\to[\!\quad\!]0,
\end{split}
\end{equation}
where $\Gamma^{\+\sigma}_{\!\mathrm{reg}}$ acts as follows
\begin{equation}\label{regGammaHyp}
    \Gamma_{\!\mathrm{reg}}^{\+\sigma}: \xi\longmapsto \alpha(\vec{x},\vec{x}_1,\ldots\check{\vec{x}}_\sigma\ldots,\vec{x}_N)\+\xi(\vec{x},\vec{x}_1,\ldots\check{\vec{x}}_\sigma\ldots,\vec{x}_N).
\end{equation}
Notice that the function $\alpha(\vec{x},\vec{x}_1,\ldots\check{\vec{x}}_\sigma\ldots,\vec{x}_N)$ is symmetric under the exchange of any couple $\vec{x}_k\-\longleftrightarrow\-\vec{x}_\ell$ with $\ell\-\neq k\-\in\-\{1,\ldots,N\}\-\smallsetminus\-\sigma$.

\n
From the heuristic point of view, we have introduced an effective two-body scattering length associated with the pair $\sigma$ depending on the positions of the other particles.
Such dependence plays the role of a repulsive force that weakens the contact interaction between the particles of the couple $\sigma$.
In particular, two kinds of repulsions are described by~\eqref{alfa}: the first term represents a three-body force that makes the usual two-body point interaction weaker and weaker as a third particle approaches the common position of the first two interacting particles, while the second term represents a four-body repulsion meant to compensate for the singular ultraviolet behavior associated with the situation in which two other different particles compose an interacting couple $\nu$, without elements in common with $\sigma$.

\n
From the mathematical perspective, the ambiguity in defining a contact interaction arises at the intersections between the coincidence hyperplanes: indeed, both kinds of the singularities discussed occur along $\pi_\sigma\cap\pi_\nu\+$, whose codimension is $6$

$$\pi_\sigma\cap\pi_\nu=\begin{cases}
        \{(\vec{x}_1,\ldots,\vec{x}_N)\in\R^{3N}\:|\;\vec{x}_i\-=\vec{x}_j,\+ \vec{x}_k\-=\vec{x}_\ell\},\qquad &\sigma=\{i,\,j\},\+ \nu=\{k,\ell\}, \,\nu\cap\sigma=\emptyset;\\
        \{(\vec{x}_1,\ldots,\vec{x}_N)\in\R^{3N}\:|\;\vec{x}_i\-=\vec{x}_j\-=\vec{x}_k\},\qquad &\sigma=\{i,\,j\},\+ \nu=\{k,\ell\}, \,\ell\-\in\-\sigma.
    \end{cases}
$$
Notice that in the first situation above, $\pi_\sigma\cap\pi_\nu$ identifies two independent pairs of interacting particles, whereas the last one is associated with the collision of three particles.
A heuristic argument motivating the need to introduce a four-body repulsive interaction (in addition to the three-body one) for $N\-\geq \-4$ is provided by Remark~\ref{4BodyNeed}.

\vs

Once again, we stress that we can choose an arbitrarily small compact support of the function $\theta$, so that the usual two-body point interaction is restored as soon as the other bosons are sufficiently far from the singularities.

\n
In this paper we shall prove that the Hamiltonian constructed as a s.a. extension of $\dot{\mathcal{H}}_0$, satisfying~\eqref{mfBC} at least in the weak topology, is bounded from below if one assumes $\gamma$ larger than a proper critical value $\gamma_c$ (see~\eqref{criticalGammaN} in the next section).
The result is obtained by using the theory of quadratic forms.
Specifically, we first construct the quadratic form naturally associated with the formal Hamiltonian characterized by the boundary condition~\eqref{mfBC}.
Then we show that the quadratic form, defined on a suitable form domain, is closed and bounded from below, and therefore it defines a unique s.a. and bounded from below operator.
Such an operator is the Hamiltonian for our Bose gas with contact interactions.

\n
As a further result, we establish the connection between our approach and the one developed by Albeverio, H{\o}egh-Krohn and Streit in 1977 and based on the theory of Dirichlet forms (\cite{AHK}, see also~\cite{MS0} where the Dirichlet form approach is used to study a fermionic system in the thermodynamic limit).
In~\cite[example 4]{AHK}, the authors define a Hamiltonian for a Bose gas with contact interactions starting from a Dirichlet form, with a proper weight function $\phi$ (see eq.~\eqref{weightFunction}).
We stress that the weight function $\phi$, and therefore the Hamiltonian, is characterized by a single parameter $m$, and this means that the scattering length of the two-body interaction, the (inevitable) three- and four-body repulsion required to avoid the collapse, are all expressed in terms of this unique parameter $m$.
Making use of general results on singular perturbations of s.a. operators, one can show that the Hamiltonian defined in~\cite[example 4]{AHK} is a special case of ours.
More precisely, we can recover such a Hamiltonian if we set 
$$\alpha_0=-m, \quad\text{with} \quad m \geq 0,\qquad \gamma=2 \quad\text{and}\quad  \theta(r)= e^{- m r}\;.$$
Notice that the two-body scattering length must be negative, $\gamma$ must be fixed and the function $\theta$ must be an exponential depending on the scattering length.\newline
However, the theory developed in~\cite{AHK} is more general, and we show that, by modifying the weight function $\phi$ (see eq.~\eqref{genericPhiDF}), it is possible to build a larger class of many-body Hamiltonians with contact interactions, which is almost coincident with the one provided with our method.\newline
We conclude by observing that the construction of the Hamiltonian via Dirichlet forms is simple and elegant but also rather implicit.
In fact, the relation between the choice of the function $\phi$ and the type of boundary conditions imposed along the coincidence hyperplanes is not clear \emph{a priori}.
Such a relation is clarified as a consequence of the connection we establish with the class of Hamiltonians constructed with our method.

\vs

\n
For the reader's convenience, we collect some of the notation adopted in the paper.

\vspace{0.25cm}

\n - Given the Euclidean space $(\R^n,\+\cdot\+)$, $\vec{x}$ is a vector in $\R^n$ and $x=\abs{\vec{x}}$.\newline
- $\mathcal{F}:\psi\longmapsto\FT{\psi}$ is the unitary Fourier transform of $\psi$.\newline
- For any $p\in \-[1,\pInfty]$ and $\Omega$ open set in $\R^n$, $\Lp{p}(\Omega,\mu)$ is the Banach space of $p\+$-integrable functions with respect to the Borel measure $\mu$.
We use $\Lp{p}(\Omega)$ in case $\mu$ is the Lebesgue measure and we denote $\define{\norm{\+\cdot\+}[p]\-; \-\norm{\+\cdot\+}[\Lp{\+p}[\R^n]]}$.\newline
- If $\hilbert$ is a complex Hilbert space, we denote by $\scalar{\cdot\+}{\!\cdot}_{\hilbert}\+$,  $\define{\norm{\+\cdot\+}[\hilbert]\!; \!\sqrt{\scalar{\cdot\+}{\!\cdot}_{\hilbert}} }\;$ the inner product and the induced norm, respectively.\newline
- If $\hilbert=\Lp{2}[\R^n]$, we simply denote by $\scalar{\cdot\+}{\!\cdot}$, $\norm{\+\cdot\+}$ the inner product and the induced norm, respectively.\newline
- $H^s(\R^n)$ is the standard Sobolev space of order $s>0$ in $\R^n$.\newline
- $f|_{\pi}\in H^s(\R^{dn})$ is the trace of $f\in H^{s\++\frac{d}{2}}(\R^{d(n+1)})$ on the hyperplane $\pi$ of codimension $d$.\newline
- Given the Hilbert spaces $X$ and $Y$, $\linear{X,Y}$ and $\bounded{X,Y}$ denote the set of linear operators and the Banach space of the bounded linear operators from $X$ to $Y$, respectively.\newline
Moreover, $\define{\linear{X};\linear{X,X}}$ and $\define{\bounded{X};\bounded{X,X}}$.\newline[-2.5]
- Given an integral operator $\mathcal{A}\-\in\-\linear{\Lp{2}(Y),\Lp{2}(X)}$, we denote by $A\Big(\!\-
        \begin{array}{c}
            x\\[-7.5pt]
            y
        \end{array}\!\!\Big)$ its kernel, with $x\-\in\- X$ and $y\-\in\- Y$.

\section{Formulation of the Main Results}\label{main}
In this section, we introduce some notation and formulate our main results.

\n
Let us denote by $G^\lambda$ the kernel of the operator $\resolvent{\mathcal{H}_0}[-\lambda]\-\in\-\bounded{\hilbert_N,\dom{\mathcal{H}_0}}$, with $\lambda\->\-0$, namely the Green's function associated with the differential operator $\maps{(-\Delta+\lambda)}{H^2(\R^{3N});\Lp{2}(\R^{3N})}$.
It is well known that
\begin{equation}\label{greenLaplace}
    G^\lambda\Big(\!\-
        \begin{array}{c}
            \vec{x}\\[-7.5pt]
            \vec{y}
        \end{array}\!\!\Big)\-=\tfrac{1}{\,(2\pi)^{3N/2}\!\-}\left(\-\tfrac{\!\!\sqrt{\lambda}}{\+\abs{\vec{x}\+-\+\vec{y}}\+}\-\right)^{\!\frac{3N}{2}-1}\!\!\mathrm{K}_{\frac{3N}{2}-1}\!\left(\-\sqrt{\lambda}\,\abs{\vec{x}\--\-\vec{y}}\right)\!,\qquad \vec{x},\vec{y}\in\R^{3N},
\end{equation}
where $\maps{\mathrm{K}_\mu}{\Rplus;\Rplus}$ is the modified Bessel function of the second kind (also known as Macdonald's function) of order $\mu\geq 0$.
Here we recall some properties of the function $\mathrm{K}_\mu$
\begin{subequations}\label{macdonaldProperties}
\begin{gather}
    z^\mu\+\mathrm{K}_\mu(z)\,\text{ is decreasing in }\,z\-\in\-\Rplus\+,\label{decreasingMacdonald}\\[-2.5pt]
    \mathrm{K}_\mu(z)=\frac{2^{\mu-1}\Gamma(\mu)}{z^\mu}+\oSmall{z^{-\mu}}\!,\qquad \text{for }\, z\to[\!\quad\!] 0^+\-, \,\mu>0,\label{originMacdonald}\\
    \mathrm{K}_\mu(z)=\sqrt{\-\tfrac{\pi}{2\+z}}\,e^{-z}\!\left[1+\tfrac{4\mu^2-1}{8\+z}+\oBig{z^{-2}}\right]\!,\qquad \text{for }\,z\to[\!\quad\!] \pInfty.\label{infinityMacdonald}
\end{gather}
\end{subequations}
We also define
\begin{align}\label{potentialDef}
    (\mathcal{G}^\lambda\xi)(\vec{x}_1,\ldots,\vec{x}_N)\vcentcolon&\mspace{-5mu}=8\pi\nquad\sum_{\mathcal{P}_N\+\ni\,\sigma\+=\+\{i,\,j\}}\-\integrate[\R^{3(N-1)}]{\xi(\vec{y},\vec{Y}_{\!\!\sigma});\mspace{-33mu}d\vec{y}d\vec{Y}_{\!\!\sigma}}\,G^\lambda\bigg(\!\!
        \begin{array}{r c l}
            \vec{x}_i\+,&\!\!\vec{x}_j\+, &\!\!\vec{X}_{\-\sigma}\\[-2.5pt]
            \vec{y},&\!\!\vec{y},&\!\!\vec{Y}_{\!\!\sigma}
        \end{array}\!\!\!\bigg)\\
    &\define*{;\sum_{\sigma\+\in\+\mathcal{P}_N}\!(\mathcal{G}_\sigma^\lambda\+\xi)(\vec{x}_1,\ldots,\vec{x}_N)},\nonumber
\end{align}
where we have set $\vec{X}_{\-\sigma}\-\vcentcolon=\-(\vec{x}_1,\ldots\check{\vec{x}}_\sigma\ldots,\vec{x}_N)$.
With a slight abuse of terminology, given $\sigma\-=\-\{i,j\}$, we refer to $\mathcal{G}^\lambda_\sigma\+\xi\-\in\!\LpS{2}[\R^6\-,d\vec{x}_i\+ d\vec{x}_j\mspace{-0.75mu}]\-\otimes\-\LpS{2}[\R^{3(N-2)}\-,d\vec{X}_{\-\sigma}\-]$ as the \emph{potential} generated by the \emph{charge} $\xi(\vec{x},\vec{X}_{\-\sigma}\-)$ distributed along $\pi_\sigma\+$, while $\mathcal{G}^\lambda\xi\-\in\-\hilbert_N$ shall consequently be the total potential.
Clearly, due to bosonic symmetry, each charge associated with the coincidence hyperplane $\pi_\sigma$ is equal to any other charge distributed along another hyperplane $\pi_\nu\+$.\newline
This nomenclature takes inspiration from the fact that $\mathcal{G}^\lambda_\sigma\+\xi$ solves the following distributional equation
\begin{equation}\label{potentialDistributionalDef}
    (\mathcal{H}_0\-+\lambda)(\mathcal{G}^\lambda_\sigma\+\xi)(\vec{x}_1,\ldots,\vec{x}_N)=8\pi\,\delta(\vec{x}_i\--\vec{x}_j)\,\xi\!\left(\tfrac{\vec{x}_i+\+\vec{x}_j}{2}\-,\vec{X}_{\-\sigma}\-\right)\!.
\end{equation}
Moreover, one can check that the Fourier transform of the potential reads
\begin{equation}\label{potentialFourier}
    (\widehat{\mathcal{G}^\lambda_\sigma\+\xi})(\vec{p}_1,\ldots,\vec{p}_N)=\sqrt{\tfrac{8}{\pi}\+}\+\frac{\FT{\xi}(\vec{p}_i\-+\vec{p}_j,\vec{P}_{\!\sigma})}{p_1^2+\ldots+p_N^2+\lambda}.
\end{equation}
We recall that the operator $\mathcal{G}^\lambda$ is injective and its range is contained in $\hilbert_N$\comment{because of the correspondence highlighted by identity~\eqref{abstractPotential}}.
Furthermore, we stress that $\ran{ \mathcal{G}^\lambda} \cap H^1(\R^{3N})=\{0\}$\comment{ because of Remark~\ref{uniqueDecomposition}}.

\n
We are now in position to introduce the quadratic form $\mathcal{Q}$ in $\hilbert_N$ 
\begin{align}
    &\dom{\mathcal{Q}}\-\vcentcolon= \!\left\{\psi \in\hilbert_N \,\big |\;\psi-\mathcal{G}^\lambda\xi=w^\lambda\-\in\- H^1(\R^{3N}),\: 
  \xi \in \dom{\Phi^{\lambda}}, 
    \:\lambda\->\-0\right\}\-,\nonumber\\
    \label{QF}&\mspace{7.5mu}\mathcal{Q}[\psi]\vcentcolon=\|\mathcal{H}_0^{\frac{1}{2}}w^\lambda\|^2+\lambda\|w^\lambda\|^2\--\lambda\- \norm{\psi}^2\-+\Phi^\lambda[\xi] 
\end{align}
where $\Phi^\lambda$ is  the following hermitian ``quadratic form of the charges'' in $\hilbert*_{N-1}\vcentcolon=\Lp{2}[\R^3]\-\otimes\-\LpS{2}[\R^{3(N-2)}]$
\begin{gather}
 \dom{\Phi^{\lambda}} \vcentcolon=\!\left\{ \xi \-\in\- \hilbert*_{N-1} \, \Big |\; \xi \-\in\- H^{\frac{1}{2}}(\R^{3(N-1)}) \right\}, \nonumber\\
  \label{phiDefinition} \mspace{7.5mu} \Phi^\lambda \vcentcolon= 4\pi\+ N(N\!-\!1) \left(\Phi^\lambda_{\mathrm{diag}}\!+\Phi^\lambda_{\mathrm{off};\+0}\-+\Phi^\lambda_{\mathrm{off};\+1}\-+\Phi_{\mathrm{reg}} \right)\!,\qquad\lambda>0
\end{gather}
and
\begin{subequations}\label{phiComponents}
\begin{gather}
    \label{diagPhi}\Phi^\lambda_{\mathrm{diag}}[\xi]\vcentcolon=\tfrac{1}{\!\-\sqrt{2\,}\,}\!\-\integrate[\R^{3(N-1)}]{\sqrt{\tfrac{1}{2}\+p^2\-+\-P^2\-+\lambda\+}\,\abs{\FT{\xi}(\vec{p},\vec{P})}^2;\mspace{-33mu}d\vec{p}\+d\mspace{-0.75mu}\vec{P}},\\
    \label{off0Phi}\Phi^\lambda_{\mathrm{off;\+0}}[\xi]\!\vcentcolon=\!-\tfrac{(N\--2)(N\--3)\!}{2\pi^2}\!\!\integrate[\R^{3N}]{\-\frac{\conjugate*{\-\FT{\xi}(\vec{p}_3\!+\-\vec{p}_4,\vec{p}_1,\vec{p}_2,\vec{P}\-)\!}\;\FT{\xi}(\vec{p}_1\!\-+\-\vec{p}_2,\vec{p}_3,\vec{p}_4,\vec{P}\-)\!}{p_1^2+p_2^2+p_3^2+p_4^2+P^2+\lambda};\mspace{-11.5mu}d\vec{p}_1\-\cdots d\vec{p}_4 d\mspace{-0.75mu}\vec{P}}\-,\\
    \label{off1Phi}\Phi^\lambda_{\mathrm{off};\+1}[\xi]\-\vcentcolon=\! -\tfrac{2\+(N\--2)\!}{\pi^2}\!\! \integrate[\R^{3N}]{\-\frac{\conjugate*{\FT{\xi}(\vec{p}_1\!+\-\vec{p}_2, \vec{p}_3, \vec{P})\-}\,\FT{\xi}(\vec{p}_1\!+\-\vec{p}_3, \vec{p}_2, \vec{P})}{p_1^2+p_2^2+p_3^2+P^2+\lambda};\mspace{-11.5mu}d\vec{p}_1 d\vec{p}_2\+d\vec{p}_3 d\mspace{-0.75mu}\vec{P}}\-,
\end{gather}
\begin{equation}\label{regPhi}
    \begin{split}
    &\Phi_{\mathrm{reg}}[\xi]\vcentcolon=\big(\Phi^{(2)}_{\mathrm{reg}}+\Phi^{(3)}_{\mathrm{reg}}+\Phi^{(4)}_{\mathrm{reg}}\big)[\xi],\\
    \Phi^{(2)}_{\mathrm{reg}}[\xi]\vcentcolon=\,\alpha_0\norm{\xi}[\hilbert*_{N-1}]^2\-,\mspace{27mu}&\mspace{27mu}
    \Phi^{(3)}_{\mathrm{reg}}[\xi]\vcentcolon=\,(N\!-\-2)\+\gamma\!\-\integrate[\R^{3(N-1)}]{\frac{\theta(\vec{x}\--\-\vec{x}')}{\abs{\vec{x}\--\-\vec{x}'}}\:\abs{\xi(\vec{x},\vec{x}'\-,\vec{X})}^2;\mspace{-33mu}d\vec{x}d\vec{x}'\-d\vec{X}},\\
    \Phi^{(4)}_{\mathrm{reg}}[\xi]\vcentcolon=&\,\tfrac{\!(N\--2)(N\--3)\!}{4}\,\gamma\!\-\integrate[\R^{3(N-1)}]{\frac{\theta(\vec{x}'\!-\-\vec{x}'')}{\abs{\vec{x}'\!-\-\vec{x}''}}\:\abs{\xi(\vec{x},\vec{x}'\-,\vec{x}''\!,\vec{X})}^2;\mspace{-33mu}d\vec{x}d\vec{x}'\-d\vec{x}''\!d\vec{X}}.
    \end{split}
\end{equation}
\end{subequations}

\n
Notice that, according to our notation, lowercase vectors represent positions in $\R^3$, whereas uppercase vectors (unless otherwise specified) correspond to collections of the missing three-dimensional integration variables.
\comment{
We also set
\begin{equation}\begin{split}
    \Phi_0[\xi]\vcentcolon=&\:\alpha_0\norm{\xi}[\hilbert*_{N-1}]^2\!+(N\!-\-2)\+\gamma\!\-\integrate[\R^{3(N-1)}]{\beta(\vec{x}-\vec{x}')\:\abs{\xi(\vec{x},\vec{x}'\-,\vec{X})}^2;\mspace{-33mu}d\vec{x}d\vec{x}'\-d\vec{X}}+\\
    &+\,\tfrac{\!(N\--2)(N\--3)\!}{4}\,\gamma\!\-\integrate[\R^{3(N-1)}]{\beta(\vec{x}'\--\vec{x}'')\:\abs{\xi(\vec{x},\vec{x}'\-,\vec{x}''\!,\vec{X})}^2;\mspace{-33mu}d\vec{x}d\vec{x}'\-d\vec{x}''\!d\vec{X}},\qquad \beta: \vec{x}\longmapsto \frac{\theta(\vec{x})-1}{\abs{\vec{x}}}
    \end{split}
\end{equation}
so that $\Phi_{\mathrm{reg}}=\Phi_0+\breve{\Phi}_{\mathrm{reg}}$ with
\begin{equation}
\begin{split}
    \breve{\Phi}_{\mathrm{reg}}[\xi]\vcentcolon=&\:(N\!-2)\+\gamma\!\-\integrate[\R^{3(N-1)}]{\frac{\abs{\xi(\vec{x},\vec{x}'\-,\vec{X})}^2\!\-}{\abs{\vec{x}-\vec{x}'}}\,;\mspace{-33mu}d\vec{x}d\vec{x}'\-d\vec{X}}+\\
    &+\,\tfrac{\!(N\--2)(N\--3)\!}{4}\,\gamma\!\-\integrate[\R^{3(N-1)}]{\frac{\abs{\xi(\vec{x},\vec{x}'\-,\vec{x}''\!,\vec{X})}^2\!\-}{\abs{\vec{x}'\--\vec{x}''}};\mspace{-33mu}d\vec{x}d\vec{x}'\-d\vec{x}''\!d\vec{X}}.
    \end{split}
\end{equation}
This decomposition is meaningful, since the quantity $\Phi_0$ is bounded, owing to~\eqref{positiveBoundedCondition} that is equivalent to $\beta\in\Lp{\infty}[\R^3]$.
In particular, this implies that the quantity $\Phi_0$ cannot play any role in the compensation of the singularities contained in the terms $\Phi_{\mathrm{off};\+0}$ and $\Phi_{\mathrm{off};\+1}$, and the only relevant contribution comes from $\breve{\Phi}_{\mathrm{reg}}$.
}

\n
\begin{note}
    We observe that $\mathcal{Q}$ defines an extension of the quadratic form associated to the free Hamiltonian $\mathcal{H}_0\+$.
    Indeed, let $\psi\in\dom{\mathcal{Q}}\cap H^1(\R^{3N})$, \ie~the form domain of the free Hamiltonian.
    Due to the injectivity of $\mathcal{G}^\lambda$, this is possible if and only if the charge $\xi$ is zero.
    Hence,
    \begin{gather*}
        \dom{\mathcal{Q}}\cap H^1(\R^{3N})=\hilbert_N\cap H^1(\R^{3N}),\\
        \mathcal{Q}[\psi]=\|\mathcal{H}_0^{\frac{1}{2}}\psi\|^2,\qquad \forall \psi\in\hilbert_N\cap H^1(\R^{3N}).
    \end{gather*}
\end{note}
\begin{note}\label{vectorDecomposition}
    In Appendix~\ref{heuristics} we will give a heuristic motivation showing that $\mathcal{Q}$ is the quadratic form associated with the formal Hamiltonian introduced in Section~\ref{intro} and characterized by the boundary condition~\eqref{mfBC}.
    However, we provide some preliminary observations here.
    According to Proposition~\ref{diagonalSingularityCancellation}, we have for each $\sigma=\{i,j\}$ an explicit asymptotic behavior of the potential around $\pi_\sigma$ in the weak topology given by
    $$(\mathcal{G}_\sigma^\lambda\+\xi)(\vec{x}_1,\ldots,\vec{x}_N) = \tfrac{\xi\left(\frac{\vec{x}_i+\+\vec{x}_j}{2}\-,\,\vec{X}_{\-\sigma}\!\right)}{\abs{\vec{x}_i-\+\vec{x}_j}}-(\Gamma^{\+\lambda}_{\!\mathrm{diag}}\+\xi)\!\left(\tfrac{\vec{x}_i+\+\vec{x}_j}{2}\-,\vec{X}_{\-\sigma}\!\right)\-+\oSmall{1}\!, \qquad\text{as }\,\abs{\vec{x}_i\--\vec{x}_j}\longrightarrow 0,$$
    where $\maps{\Gamma^{\+\lambda}_{\!\mathrm{diag}}}{\hilbert*_{N-1}\cap H^1(\R^{3(N-1)});\hilbert*_{N-1}}$ satisfies
    \begin{subequations}\label{GammasDefHyp}
    \begin{equation}\label{diagGammaHyp}
        (\widehat{\Gamma^{\+\lambda}_{\!\mathrm{diag}}\+\xi})(\vec{k},\vec{p}_1,\ldots,\vec{p}_{N-2})=\tfrac{1}{\!\-\sqrt{2\+}\+}\sqrt{\tfrac{1}{2}\+k^2\-+p_1^2+\ldots+p^2_{N-2}+\lambda\,}\,\FT{\xi}(\vec{k},\vec{p}_1,\ldots,\vec{p}_{N-2}).
    \end{equation}
    Moreover, the other potentials $\mathcal{G}^\lambda_\nu\+\xi$ with $\nu\neq \sigma$ are not singular along $\pi_\sigma$, since (see again Proposition~\ref{diagonalSingularityCancellation})
    \begin{equation*}
        (\mathcal{G}^\lambda_\nu\+\xi)(\vec{x}_1,\ldots,\vec{x}_N)=-(\Gamma^{\+\nu\+\smallsetminus\sigma,\+\lambda}_{\!\mathrm{off};\+\abs{\sigma\+\cap\+\nu}}\+\xi)\!\left(\tfrac{\vec{x}_i+\+\vec{x}_j}{2}\-,\vec{X}_{\-\sigma}\!\right)\-+\oSmall{1}\-, \qquad\text{as }\,\abs{\vec{x}_i-\vec{x}_j}\longrightarrow 0,
    \end{equation*}
    where, given $\mu=\{k,\ell\}\!\in\!\mathcal{P}_{N\--2}$ and $\mathfrak{m}\!\in\!\{1,\ldots,N\!-\-2\}$, the operators $\Gamma^{\+\mu,\+\lambda}_{\!\mathrm{off};\+0}$ and $\Gamma^{\+\{\mathfrak{m}\},\+\lambda}_{\!\mathrm{off};\+1}$, both defined in $\hilbert*_{N-1}\-\cap \- H^1(\R^{3(N-1)})$
    (see Proposition~\ref{GammaDomain}), are given by
    \begin{align}
        (\widehat{\Gamma^{\+\mu,\+\lambda}_{\!\mathrm{off};\+0}\+\xi})(\vec{k},\vec{p}_k,\vec{p}_\ell,\vec{p}_1,\ldots\check{\vec{p}}_\mu\ldots,\vec{p}_{N-2})=&-\frac{1}{\,\pi^2\!\-}\integrate[\R^3]{\frac{\FT{\xi}(\vec{p}_k\-+\vec{p}_\ell\+,\vec{k}\--\-\vec{q},\vec{q},\vec{p}_1,\ldots\check{\vec{p}}_\mu\ldots,\vec{p}_{N-2})}{\abs{\vec{k}\--\-\vec{q}}^2+q^2+p_1^2+\ldots+p_{N-2}^2+\lambda};\-d\vec{q}},\label{off0GammaHyp}\\
        \big(\widehat{\Gamma^{\+\{\mathfrak{m}\},\+\lambda}_{\!\mathrm{off};\+1}\+\xi}\big)(\vec{k},\vec{p}_{\mathfrak{m}},\vec{p}_1,\ldots\check{\vec{p}}_{\mathfrak{m}}\ldots,\vec{p}_{N-2})=&-\frac{1}{\,\pi^2\!\-}\-\integrate[\R^3]{\frac{\FT{\xi}(\vec{k}\--\-\vec{q}+\vec{p}_{\mathfrak{m}}\mspace{0.75mu},\vec{q},\vec{p}_1,\ldots\check{\vec{p}}_{\mathfrak{m}}\ldots,\vec{p}_{N-2})}{\abs{\vec{k}\--\-\vec{q}}^2+q^2+p_1^2+\ldots+p_{N-2}^2+\lambda};\-d\vec{q}}.\label{off1GammaHyp}
    \end{align}
    \end{subequations}
    This means that the total potential has the following asymptotic expansion near $\pi_\sigma$ in the weak topology
    \begin{equation}\label{potentialAsymptotics}
        (\mathcal{G}^\lambda\xi)(\vec{x}_1,\ldots,\vec{x}_N) = \tfrac{\xi\left(\frac{\vec{x}_i+\+\vec{x}_j}{2}\-,\,\vec{X}_{\-\sigma}\!\right)}{\abs{\vec{x}_i-\+\vec{x}_j}}-(\Gamma^{ \+\sigma,\+\lambda}\+\xi)\!\left(\tfrac{\vec{x}_i+\+\vec{x}_j}{2}\-,\vec{X}_{\-\sigma}\!\right)\-+\oSmall{1}\-, \qquad\text{as }\,\abs{\vec{x}_i\--\vec{x}_j}\longrightarrow 0,
    \end{equation}
    where $\Gamma^{\+\sigma,\+\lambda}\in\linear{\hilbert*_{N-1}}$ is defined by
    \begin{equation}\label{GammaTMSHyp}
        \Gamma^{\+\sigma, \+\lambda}\!\vcentcolon=\Gamma^{\+\lambda}_{\!\mathrm{diag}}+\sum_{\nu\+\neq \+\sigma}\Gamma^{\+\nu\+\smallsetminus\sigma,\+\lambda}_{\!\mathrm{off};\+\abs{\sigma\+\cap\+\nu}}\+,\qquad\dom{\Gamma^{\+\sigma,\+\lambda}}=H^1(\R^{3(N-1)})\cap \hilbert*_{N-1}.
    \end{equation}
    Thus, in order for boundary condition~\eqref{mfBC} to be satisfied (at least in the weak topology), we can decompose any element $\psi\in \hilbert_N\cap H^2(\R^{3N}\!\smallsetminus \pi)$ in the domain of the Hamiltonian $\mathcal{H}$ we want to construct as
    \begin{equation}\label{psiDec}
        \psi = \mathcal{G}^\lambda\xi + \varphi_\lambda
    \end{equation}
    where $\varphi_\lambda$ is an element of $\hilbert_N\cap H^2(\R^{3N})$, so that $\psi$ inherits from the total potential $\mathcal{G}^\lambda\xi$ all the singularities at the leading order near $\pi$.
    Moreover, since our class of s.a. extensions is characterized by the behavior of $\psi$ at $\oBig{1}$ in the expansion along $\pi$, the charge $\xi$ and the regular component $\varphi_\lambda$ must fulfill the following constraint
    \begin{equation}\label{regularComponentTraced}
        \varphi_\lambda|_{\pi_\sigma} \-= (\Gamma^{\+\sigma,\+\lambda}\!+\Gamma^{\+\sigma}_{\!\mathrm{reg}})\+\xi\+.
    \end{equation}
\end{note}


\n
Finally, for any fixed $N\-\geq\- 3$, we define the threshold parameter\footnote{Notice that $\gamma_c$ is surely not optimal, since already for $N\!=\-3$ we have $\gamma_c \->\- \gamma_c^{\text{optimal}}\!= \tfrac{4}{3} -  \tfrac{\sqrt{3}}{\pi}\,,$ where $\gamma_c^{\text{optimal}}$ has been explicitly computed in \cite[eq.~(13)]{FeT2}.}
\begin{equation}\label{criticalGammaN}
    \gamma_c\vcentcolon=2-\frac{1}{\pi(N\!-\-2)\!\left(\+\frac{1}{\!\-\sqrt{3\+}\+}\-+\frac{N\--3}{8}\-\right)\!\-}\,.
\end{equation}

\n
The main technical point of our paper is the proof of the following proposition.
\begin{prop}\label{closednessTheo}
    Let $\Phi^\lambda$ be the hermitian quadratic form in $\hilbert*_{N-1}$ defined by~\emph{(\ref{phiDefinition},~\ref{phiComponents})} and assume $\gamma\->\-\gamma_c\+$.\newline
    Then $\Phi^\lambda$ is closed for any $\lambda\- >\-0$ and satisfies
    \begin{equation*}
        \Phi^\lambda[\xi]\geq 4\pi N (N\!-\-1)\!\left[\sqrt{\tfrac{\lambda}{2}}\,\sqrt{1\--\-\Lambda_N^2}+\alpha_0-\tfrac{(N+1)(N-2)\+\gamma}{4\,b}\right]\!\norm{\xi}[\hilbert*_{N-1}]^2\!,\quad \forall \lambda\->\-0,\,\xi\-\in \hilbert*_{N-1},
    \end{equation*}
    where
    \begin{equation}\label{simplerLambda}
        \Lambda_N\vcentcolon=\max\!\left\{0,1\--(N\!-\-2)\+\pi\-\left(\tfrac{1}{\sqrt{3}}+\tfrac{N\--3}{8}\right)\!(\gamma\--\-\gamma_c)\-\right\}\-\in[0,1).
    \end{equation}
    Moreover, $\Phi^\lambda$ is coercive for any $\lambda\->\-\lambda_0$, with
    \begin{equation}\label{lambdaNote}
        \define{\lambda_0;\frac{2\max\{0,-\alpha_0+(N\!+\-1)(N\!-\-2)\,\gamma/(4b)\}^2\!\-}{1\--\-\Lambda_N^2}}\,.
    \end{equation}
\end{prop}
\n We observe that, by means of the previous proposition, $\Phi^\lambda$ is semi-bounded in $\hilbert*_{N-1}$.
In particular, it provides the definition of a unique s.a. operator
$\Gamma^\lambda$ associated with $\Phi^\lambda$ for any $\lambda\!>\!0$, satisfying
$$\Gamma^\lambda=4\pi N(N\!-\-1)\big(\Gamma^{\+\sigma,\+\lambda}\!+\Gamma^{\+\sigma}_{\!\mathrm{reg}}\big),\qquad \text{at }\, \sigma=\{N\!-\-1,N\}.$$
One can prove (see Proposition~\ref{GammaDomain}) that its domain does not depend on $\lambda$ and it is given by
$$\dom{\Gamma^\lambda}=\!\left\{ \xi \-\in \hilbert*_{N-1} \, \Big |\; \xi \-\in\- H^1(\R^{3(N-1)}) \right\}\!.$$
Additionally, $\Gamma^\lambda$ has a bounded inverse whenever $\lambda\!>\!\lambda_0\+$, since in this case $\Phi^\lambda$ is bounded from below by a positive constant.
Using Proposition~\ref{closednessTheo}, we prove our main result.
\begin{theo}\label{hamiltonianCharacterizationTheo}
    Let $\gamma_c$ and $\lambda_0$ be  given by~\eqref{criticalGammaN} and~\eqref{lambdaNote}, respectively, and assume $\gamma\!>\-\gamma_c\+$.
    Moreover, consider an essentially bounded function $\theta$ satisfying assumption~\eqref{positiveBoundedCondition}.\newline
    Then, the quadratic form $\mathcal{Q}$ in $\hilbert_N$ defined by~\eqref{QF} is independent of $\lambda$, closed and bounded from below by $-\lambda_0\+$.
    The s.a. and bounded from below operator uniquely defined by  $\mathcal{Q}$ is given by
    \begin{align}
        \dom{\mathcal{H}}&=\!\Big\{\psi\in\-\hilbert_N\,\big|\; \psi\--\mathcal{G}^\lambda\xi=\varphi_\lambda\!\in\-  H^2(\R^{3N}),\; \xi\- \in\- \dom{\Gamma^\lambda},\,\Gamma^\lambda\xi\- =  4\pi N(N\!-1)\+\varphi_\lambda|_{\pi_{\{N\--1,\,N\}}}\-,\,\lambda>\-0\Big\}\-,\nonumber\\
        \mathcal{H}\psi&=\mathcal{H}_0\+\varphi_\lambda\--\lambda\+  \mathcal{G}^\lambda \xi.\label{hamiltonianCharacterization}
    \end{align}
    The operator $\mathcal{H}$ is a s.a. extension of the operator $\dot{\mathcal{H}}_0$ defined by~\eqref{symmetricHamiltonianToBeExtended}, the elements of its domain satisfy the boundary condition~\eqref{mfBC} in the topology induced by the norm of the $\Lp{2}$ space and $\mathcal{H}\geq -\lambda_0\+$.\newline
    Moreover, for $\lambda\->\-\lambda_0$ and $f\-\in\-\hilbert_N$ one has \begin{equation}\label{resolventH}
    (\mathcal{H}+\lambda)^{-1}f=(\mathcal{H}_0+\lambda)^{-1}f+\mathcal{G}^\lambda\xi,
    \end{equation}
    where $\xi$ solves the equation
    \begin{equation}\label{chargeEq}
    \Gamma^\lambda\xi=4\pi N(N\!-\-1)\big((\mathcal{H}_0+\lambda)^{-1}\mspace{-2.25mu}f\big)|_{\pi_{\{N,N-1\}}}.
    \end{equation}
\end{theo}
\n We stress that Theorem~\ref{hamiltonianCharacterizationTheo} provides the following estimate for the infimum of the spectrum of the Hamiltonian
\begin{equation}\label{unoptimalLowerBound}
    \inf\spectrum{\mathcal{H}}\geq\begin{dcases}
        \quad 0, & \text{if }\,\alpha_0\geq \tfrac{(N+1)(N-2)\,\gamma}{4\+b},\\
        -\frac{\left[(N\!+\-1)(N\!-\-2)\,\gamma-4\+b\,\alpha_0\right]^2}{8\+b^2\+ (1-\Lambda_N^2)}, & \text{otherwise.}
    \end{dcases}
\end{equation}

\n
As a further result, we also establish the connection between the result in Theorem~\ref{hamiltonianCharacterizationTheo} and the approach in Dirichlet forms developed in~\cite{AHK}.
Let us recall that in~\cite{AHK} the authors introduce the following Dirichlet form
\begin{equation}\label{AlbeverioDF}
    \define{E_\phi[\psi];\integrate[\R^{3N}]{\phi^2(\vec{\mathrm{x}})\abs{\nabla \psi(\vec{\mathrm{x}})}^2;\mspace{-10mu}d\vec{\mathrm{x}}}},\qquad\psi\in \hilbert_N:\:\nabla \psi \in \Lp{2}[\R^{3N}\-,\,\phi^2(\vec{\mathrm{x}})\+d\vec{\mathrm{x}}],
\end{equation}
where, in the case of their \emph{example 4}
\begin{equation}\label{weightFunction}
    \phi\mspace{0.75mu}(\vec{x}_1,\ldots,\vec{x}_N)=\frac{1}{4\pi}\! \sum_{1\leq\+ i\,<\,j\+\leq N}\!\-\frac{\;e^{-m\,\abs{\vec{x}_i-\+\vec{x}_j}}\!}{\abs{\vec{x}_i\--\vec{x}_j}}\in\Lp*{2}[\R^{3N}],\qquad m\geq 0.
\end{equation}
Then it is shown that the quantity \begin{equation}\label{dirichletQF0}
\define{\mathcal{Q}_D[\psi];E_\phi\!\left[\psi/\phi\right]-2m^2\norm{\psi}^2}
\end{equation}
defines a singular perturbation of $\mathcal{H}_0\+,\hilbert_N\cap H^2(\R^{3N})$ supported on $\pi$.
More precisely, for any non-negative value of $m$, the quadratic form $\mathcal{Q}_D$ is associated to a bounded from below operator, denoted by $-\Delta_m$ such that
\begin{gather*}
    -\Delta_m \psi=\mathcal{H}_0\+\psi,\qquad \forall \psi\-\in\!\hilbert_N\cap H_0^2(\R^{3N}\setminus \pi),\\
    -\Delta_m\-\geq\- -2\+m^2.
\end{gather*}
In this way, one defines a $N$-body Hamiltonian with contact interactions with preassigned lower bound $-2m^2$, which is therefore stable.
However, in~\cite{AHK} the domain of the Hamiltonian is not characterized.
In other words, it is not specified which boundary condition is satisfied along the coincidence hyperplanes by the elements of the domain of the Hamiltonian, and therefore it is not clear which kind of contact interaction is defined.
\comment{
Our goal is to rewrite $\mathcal{Q}_D$ in our formalism so that an explicit comparison with our results can be made.
More precisely, in the next proposition we show that the Hamiltonian $-\Delta_m$ defined in~\cite{AHK} is a special case of our class of Hamiltonians.
\begin{prop}\label{ParticularCase}
Let $\mathcal{H}$ be the Hamiltonian in $\hilbert_N$ characterized by~\eqref{hamiltonianCharacterization} with   
\begin{equation*}
 \alpha_0=-m, \, \text{ with } \,  m \geq 0, \qquad  \gamma=2 \qquad \text{and} \qquad \theta: \,r\: \longmapsto\: e^{-m\+r},\quad r>0.
\end{equation*}
Then, one has
\begin{equation*}
    -\Delta_m\-=\mathcal{H}\,.
\end{equation*}
\end{prop}
}
\n As we mentioned in the introduction, we extend the result of~\cite[example 4]{AHK} by considering a more general weight function $\phi$.
More precisely, let us consider
\begin{equation}\label{genericPhiDF}
    \phi(\vec{x}_1,\ldots,\vec{x}_N)=\nquad\sum_{1\leq\+k\,<\,\ell\+\leq N}\mspace{-12mu} \frac{\theta(\abs{\vec{x}_k\!-\-\vec{x}_\ell})}{\abs{\vec{x}_k\!-\-\vec{x}_\ell}}+\alpha_0-\theta'\-(0),
    \end{equation}
    where $\alpha_0\!\in\-\R, \theta\-\in\! H^2(\Rplus)$ such that $\theta(0)\-=\-1$ (notice that in this case assumption~\eqref{positiveBoundedCondition} is satisfied) and $\frac{\Delta\phi}{\phi}\-\in\-\Lp{\infty}[\R^{3N}]$.
    The Dirichlet form~\eqref{AlbeverioDF} defined with such a weight function is closed, hence it defines a positive s.a. operator $H_\phi$ in $\Lp{2}[\R^{3N}\!,\,\phi^2(\vec{\mathrm{x}})d\vec{\mathrm{x}}]$ (for more details, see Section~\ref{dirichletSection}).
    In particular, the generalized version of the quadratic form $\mathcal{Q}_D$ provided in~\eqref{dirichletQF0}, is now given by $E_\phi[\psi/\phi]\--\-\scalar{\psi}{\!\tfrac{\Delta\phi}{\phi}\mspace{-0.75mu}\psi}$.
    Our goal is to rewrite such a quadratic form in our formalism so that an explicit comparison with our results can be made.
    The outcome of such a comparison is expressed by the following proposition.
\begin{prop}\label{lastResult}
    Let $\phi$ be the function defined in~\eqref{genericPhiDF} with $\alpha_0\in\R$ and $\theta\in H^2(\Rplus)$ with $\theta(0)=1$ such that $\frac{\Delta\phi}{\phi}\in\Lp{\infty}[\R^{3N}]$.
    Then, the s.a. operator $-\Delta_\phi=\phi\big(H_\phi-\tfrac{\Delta\phi}{\phi}\big)\phi^{-1}$ in $\hilbert_N$ coincides with the Hamiltonian $\mathcal{H}$ defined in~\eqref{hamiltonianCharacterization} with $\gamma=2$.    
\end{prop}
\n The above result shows that the approach via Dirichlet forms provides essentially the same class of Hamiltonians that we constructed with our method (apart from some additional regularity required for the function $\theta$ and the restriction to the case $\gamma=2$).
Ultimately, if $\gamma=2$ and $\theta\-\in\! H^2(\Rplus)$ with $\theta(0)\mspace{-0.75mu}=\-1$, one finds the relation
\begin{equation}
    \alpha(\vec{z},\vec{y}_1,\ldots,\vec{y}_{N-2})=\lim_{\abs{\vec{r}}\to 0}\left[\phi\big(\vec{y}_1,\ldots,\vec{y}_{N-2},\vec{z}+\tfrac{\vec{r}}{2},\vec{z}-\tfrac{\vec{r}}{2}\big)-\tfrac{1}{\abs{\vec{r}}}\right]\-,
\end{equation}
with $\alpha$ and $\phi$ given by~\eqref{alfa} and~\eqref{genericPhiDF}, respectively.\newline
As a corollary of the previous proposition, with the specific weight function $\phi$ given by~\eqref{weightFunction}, we obtain in our formalism the example originally discussed by Albeverio et al.
\begin{cor}\label{ParticularCase}
Let $\mathcal{H}$ be the Hamiltonian in $\hilbert_N$ characterized by~\eqref{hamiltonianCharacterization} with   
\begin{equation*}
 \alpha_0=-m, \, \text{ with } \,  m \geq 0, \qquad  \gamma=2 \qquad \text{and} \qquad \theta: \,r\: \longmapsto\: e^{-m\+r},\quad r>0.
\end{equation*}
Then, one has
\begin{equation*}
    -\Delta_m\-=\mathcal{H}\,.
\end{equation*}
\end{cor}

\vs

We conclude this section by providing an outline of the content of what follows.\newline
In Section~\ref{qfSection} we carry out an analysis on the quadratic form of the charges $\Phi^\lambda$, highlighting its main properties.\newline
In Section~\ref{proveMain} we provide the proofs of Proposition~\ref{closednessTheo} and Theorem~\ref{hamiltonianCharacterizationTheo}.\newline
In Section~\ref{dirichletSection} we establish the connection between our result and the one given in~\cite{AHK}.\newline
In Appendix~\ref{heuristics} we provide heuristic motivations for the definition of the quadratic form $\mathcal{Q}$.\newline
In Appendix~\ref{appB} we collect some useful technical results about the potential $\mathcal{G}^\lambda$.

\comment{\begin{prop}\label{transform4Preliminaries}
    Given $\xi\in\mathcal{S}(\R^{3(N-1)})$ one has
    \begin{align*}
        I(\vec{p},\vec{p}_1,\vec{p}_2,\vec{P})\vcentcolon=\!\-\integrate[\R^{3(N-1)}]{\frac{e^{-i\vec{p}\+\cdot\+\vec{x}-i\vec{p}_1\cdot\+\vec{x}_1-i\vec{p}_2\cdot\+\vec{x}_2-i\vec{P}\+\cdot\vec{X}}\!\-}{(2\pi)^{\frac{3(N-1)}{2}}};\mspace{-33mu}d\vec{x}d\vec{x}_1 d\vec{x}_2\+d\vec{X}}\!\-\integrate[\R^{3(N-1)}]{\xi(\vec{y},\vec{y}_1,\vec{y}_2,\vec{Y});\mspace{-33mu}d\vec{y}d\vec{y}_1 d\vec{y}_2\+ d\vec{Y}}\times\\
        \times\, G^\lambda\bigg(\!\!\-
            \begin{array}{r c c c l}
                \vec{x}_1, & \!\!\-\vec{x}_2\+, &\!\!\- \vec{x}, & \!\!\- \vec{x}, & \mspace{-12mu}\vec{X}\\[-2.5pt]
                \vec{y}, & \!\!\-\vec{y}, &\!\!\-\vec{y}_1, & \!\!\-\vec{y}_2\+, & \mspace{-12mu}\vec{Y} \end{array}\!\!\-\bigg)\!=\frac{1}{(2\pi)^3\!\-}\integrate[\R^3]{\frac{\FT{\xi}\!\left(\vec{p}_1\-+\vec{p}_2,\frac{\vec{p}}{2}+\vec{q},\frac{\vec{p}}{2}-\vec{q},\vec{P}\right)\!}{\frac{1}{2}\+p^2\-+2q^2\-+p_1^2+p_2^2+\-P^2\-+\lambda};\-d\vec{q}}.
    \end{align*}
    \begin{proof}
        Because of the regularity of $\xi$ and equation~\eqref{greenIntegrated}, Fubini's theorem applies, therefore taking account of Proposition~\ref{greenTransformed}
        \begin{align*}
            I(\vec{p},\vec{p}_1&,\vec{p}_2,\vec{P})=\\
            &=\!\-\integrate[\R^{3(N-1)}]{\frac{e^{-i\vec{p}\,\cdot\+\frac{\vec{y}_1+\+\vec{y}_2}{2}-i(\vec{p}_1+\+\vec{p}_2)\+\cdot\+\vec{y}-i\vec{P}\cdot\vec{Y}-\sqrt{\frac{1}{2}\+p^2+\+p_1^2+\+p_2^2+P^2+\lambda\+}\:\frac{\abs{\vec{y}_1-\+\vec{y}_2}}{\sqrt{2}}}\nquad}{8\pi\+(2\pi)^{\frac{3(N-1)}{2}}\abs{\vec{y}_1\!-\vec{y}_2}}\;\xi(\vec{y},\vec{y}_1,\vec{y}_2,\vec{Y});\mspace{-33mu}d\vec{y}d\vec{y}_1 d\vec{y}_2\+d\vec{Y}}\\
            &=\!\-\integrate[\R^{3(N-1)}]{\frac{e^{-i\vec{p}\,\cdot\+\vec{s}-i(\vec{p}_1+\+\vec{p}_2)\+\cdot\+\vec{y}-i\vec{P}\cdot\vec{Y}-\sqrt{\frac{1}{2}\+p^2+\+p_1^2+\+p_2^2+P^2+\lambda\+}\:\frac{r}{\sqrt{2}}}\!\-}{8\pi\+(2\pi)^{\frac{3(N-1)}{2}}\,r}\;\xi\!\left(\vec{y},\vec{s}+\tfrac{\vec{r}}{2},\vec{s}-\tfrac{\vec{r}}{2},\vec{Y}\right);\mspace{-33mu}d\vec{y}d\vec{r} d\vec{s}\+d\vec{Y}}\\
            &=\frac{1}{(2\pi)^3\!\-}\integrate[\R^{3(N+1)}]{\frac{e^{-i\vec{q}\+\cdot\+\vec{r}-i\vec{p}\,\cdot\+\vec{s}-i(\vec{p}_1+\+\vec{p}_2)\+\cdot\+\vec{y}-i\vec{P}\cdot\vec{Y}}\!\-}{(2\pi)^{\frac{3(N-1)}{2}}}\;\frac{\xi\!\left(\vec{y},\vec{s}+\tfrac{\vec{r}}{2},\vec{s}-\tfrac{\vec{r}}{2},\vec{Y}\right)}{2\+q^2\-+\frac{1}{2}\+p^2\-+p_1^2+p_2^2+\-P^2\-+\lambda};\mspace{-33mu}d\vec{q}d\vec{y}d\vec{r} d\vec{s}\+d\vec{Y}}
        \end{align*}
        where in the last step we have used the Plancherel's theorem and the well known identity
        \begin{equation}\label{fourierAntiTransformYukawa}
            \integrate[\R^3]{\frac{e^{i\vec{q}\+\cdot\+ \vec{x}\+-\+a\, x}\!\-}{x};\-d\vec{x} }=\frac{4\pi}{a^2\-+q^2\!\-}\;,\qquad \forall a>0,\vec{q}\in\R^3.
        \end{equation}
        The result is obtained by changing the coordinates $\vec{s}+\frac{\vec{r}}{2}\longmapsto \vec{t}$, $\vec{s}-\frac{\vec{r}}{2}\longmapsto \vec{u}$.
        Indeed
        \begin{align*}
            I(\vec{p},\vec{p}_1&,\vec{p}_2,\vec{P})=\\
            &=\frac{1}{(2\pi)^3\!\-}\integrate[\R^{3(N+1)}]{\frac{e^{-i\left(\-\frac{\vec{p}}{2}+\vec{q}\-\right)\,\cdot\+\vec{t}-i\left(\-\frac{\vec{p}}{2}-\vec{q}\-\right)\,\cdot\+\vec{u}-i(\vec{p}_1+\+\vec{p}_2)\+\cdot\+\vec{y}-i\vec{P}\cdot\vec{Y}}\xi\!\left(\vec{y},\vec{t},\vec{u},\vec{Y}\right)}{(2\pi)^{\frac{3(N-1)}{2}}\left(2\+q^2\-+\frac{1}{2}\+p^2\-+p_1^2+p_2^2+\-P^2\-+\lambda\right)};\mspace{-33mu}d\vec{q}d\vec{y}d\vec{t} d\vec{u}\+d\vec{Y}}\\
            &=\frac{1}{(2\pi)^3\!\-}\integrate[\R^3]{\frac{\FT{\xi}\!\left(\vec{p}_1\-+\vec{p}_2,\frac{\vec{p}}{2}+\vec{q},\frac{\vec{p}}{2}-\vec{q},\vec{P}\right)\!}{2\+q^2\-+\frac{1}{2}\+p^2\-+p_1^2+p_2^2+\-P^2\-+\lambda};\-d\vec{q}}.
        \end{align*}
        
    \end{proof}
\end{prop}
}

\section{The Quadratic Form \texorpdfstring{$\,\Phi^{\lambda}$}{of the Charges}}\label{qfSection}

In this section, we investigate some properties of the hermitian quadratic form of the charges $\Phi^\lambda$ defined in $\hilbert*_{N-1}=\Lp{2}[\R^3]\-\otimes\-\LpS{2}[\R^{3(N-2)}]$ by~(\ref{phiDefinition},~\ref{phiComponents}), which are the key ingredients for the proof of our main results.
The first step is to show that the definition is well posed, namely, that all the terms in~\eqref{phiDefinition} are finite for $\xi\-\in\- H^{\frac{1}{2}}(\R^{3(N-1)})$.
The assertion is evident for $\Phi^\lambda_{\mathrm{diag}}\+$.
In the following proposition, we prove that the same is true for the other three terms. 
\begin{prop}\label{upperBound}
    Let $\Phi^\lambda$ be the hermitian quadratic form in $\hilbert*_{N-1}$ defined by~\emph{(\ref{phiDefinition},~\ref{phiComponents})}.\newline
    Then, for all $\xi\-\in\- \hilbert*_{N-1}\cap H^{\frac{1}{2}}(\R^{3(N-1)})$ there exist two positive constants $c_1, c_2$ such that
    $$\abs{\Phi_{\mathrm{reg}}[\xi]}\leq c_1\norm{\xi}[H^{1/2}(\R^{3(N-1)})]^2\qquad  \text{and} \qquad \abs{\Phi_{\mathrm{off};\+ \iota}^\lambda[\xi]}\leq c_2 \norm{\xi}[H^{1/2}(\R^{3(N-1)})]^2,\qquad \iota\in\{0,1\}.$$

\begin{proof}
    Concerning the regularizing term~\eqref{regPhi}, one has
    \begin{equation*}
    \begin{split}
    \abs{\Phi_{\mathrm{reg}}[\xi]}\leq \abs{\alpha_0}\norm{\xi}[\hilbert*_{N-1}]^2\!+(N\!-\-2)\+\gamma\norm{\theta}[\infty]&\left[\integrate[\R^{3(N-1)}]{\frac{\abs{\xi\-\left(\vec{r}'\-+\tfrac{\vec{r}}{3},\vec{r}'\--\tfrac{2\vec{r}}{3},\vec{X}\right)}^2\!\-}{r}; \mspace{-33mu}d\vec{r}d\vec{r}'\-d\vec{X}}\right.+\\[-5pt]
    &\;\;+\frac{N\!-\-3}{4}\!\left.\integrate[\R^{3(N-1)}]{\frac{\abs{\xi\big(\vec{x},\vec{r}''\!+\tfrac{\vec{r}'}{2},\vec{r}''\!-\tfrac{\vec{r}'}{2},\vec{X}\big)\-}^2\!\-}{r'};\mspace{-33mu}d\vec{x}d\vec{r}'\-d\vec{r}''\!d\vec{X}}\:\right]\!.
    \end{split}
    \end{equation*}
    We notice that, if $U$ is any unitary operator encoding a $\mathrm{SL}(2,\R)$ transformation $(\vec{x},\vec{y})\longmapsto(\vec{r},\vec{R})$,
    there holds
    $$ (Uf)(\vec{r},\vec{R}) \in H^{\frac{1}{2}}(\R^6, d\vec{r}d\vec{R}) \iff f(\vec{x},\vec{y}) \in H^{\frac{1}{2}}(\R^6,d\vec{x} d\vec{y}).$$
    In our case this means that $\xi\left(\vec{r}'\-+\tfrac{\vec{r}}{3},\vec{r}'\--\tfrac{2\vec{r}}{3},\vec{X}\right)=\vcentcolon\tilde{\xi}(\vec{r},\vec{r}'\-,\vec{X})$ defines a function in $H^{\frac{1}{2}}(\R^{3(N-1)},d\vec{r}d\vec{r}'\-d\vec{X})$ and similarly $\xi(\vec{x},\vec{r}''\!+\tfrac{\vec{r}'}{2},\vec{r}''\!-\tfrac{\vec{r}'}{2},\vec{X})\in H^{\frac{1}{2}}(\R^{3(N-1)}, d\vec{x}d\vec{r}'\-d\vec{r}''\!d\vec{X})$.
    Hence, exploiting the sharp Hardy-Rellich inequality (see~\cite[eq.\footnote{There is a typo in~\cite[eq.~(1.4)]{Y}, since a power $2$ is missing on the $\Gamma$ function in the numerator of the first argument of the maximum.} (1.1, 1.4)]{Y})
    \begin{equation}\label{Hardy-Rellich}
        \integrate[\R^3]{\frac{\abs{u(\vec{x})}^2\!\-}{x};\- d\vec{x}}\leq\frac{\pi}{2}\- \integrate[\R^3]{ k\, \abs{\FT{u}(\vec{k})}^2;\- d\vec{k}},\qquad u\in H^{\frac{1}{2}}(\R^3),
    \end{equation}
    one obtains the upper bound for $\Phi_{\mathrm{reg}}$.\newline
    Next, consider $\+\Phi^\lambda_{\mathrm{off};\+1}$.
    By the Cauchy-Schwarz inequality one gets
    \begin{align*}
        \abs{\Phi^\lambda_{\mathrm{off};\+1}[\xi]}&\leq \tfrac{2\+(N\--2)\!}{\pi^2}\!\!\integrate[\R^6]{\frac{1}{\+p_2^2+p_3^2\,}\,\sqrt{\integrate[\R^{3(N-2)}]{\abs{\FT{\xi}(\vec{p}+\vec{p}_2,\vec{p}_3,\vec{P})}^2; \mspace{-34.5mu}d\vec{p}\+d\mspace{-0.75mu}\vec{P}} \!\integrate[\R^{3(N-2)}]{\abs{\FT{\xi}(\vec{q}+\vec{p}_3,\vec{p}_2,\vec{P})}^2; \mspace{-34.5mu}d\vec{q}d\mspace{-0.75mu}\vec{P}}};\-d\vec{p}_2\+d\vec{p}_3}\\
        &=\tfrac{2\+(N\--2)\!}{\pi^2}\!\! \integrate[\R^6]{\,\frac{\FT{f}(\vec{p}_3)\+\FT{f}(\vec{p}_2)}{p_2^2+p_3^2};\-d\vec{p}_2\+d\vec{p}_3},
    \end{align*}
    where we have defined $f\-\in\- H^{\frac{1}{2}}(\R^3)$ as
    $$\FT{f}(\vec{p})\vcentcolon= \sqrt{\integrate[\R^{3(N-2)}]{\abs{\FT{\xi}(\vec{q},\vec{p},\vec{P})}^2;\mspace{-34.5mu}d\vec{q}d\mspace{-0.75mu}\vec{P}}\+}\+.$$
    The operator $\maps{Q}{\Lp{2}(\R^3);\Lp{2}(\R^3)}$ acting as
    \begin{equation}
        \define{\left(Q\,\psi\right)\!(\vec{p})\-;\!\- \integrate[\R^3]{\frac{\psi(\vec{k})}{\!\-\sqrt{p\+}\+ (p^2\-+k^2)\sqrt{k\+}\,};\-d\vec{k}}}\+,
    \end{equation}
    is bounded with norm $2\pi^2$ (see \eg, \cite[lemma 2.1]{FT}).
    Hence, observing that $\sqrt{p\+}\+\FT{f}(\vec{p})\-\in\- \Lp{2}[\R^3,d\vec{p}]$, one has
    $$\abs{\Phi^\lambda_{\mathrm{off};\+1}[\xi]}\leq 4\+(N\!-\-2)\,\lVert\sqrt{p\+}\+\FT{f}(\vec{p})\rVert_{\Lp{2}[\R^3\!,\:d\vec{p}]}^2\!\leq 4\+(N\!-\-2)\norm{\xi}[H^{1/2}(\R^{3(N-1)})]^2\-.$$
    In a similar way one obtains the upper bound for $\Phi^\lambda_{\mathrm{off};\+0}$
    \begin{align*}
        \abs{\Phi^\lambda_{\mathrm{off};\+0}[\xi]}&\leq \tfrac{(N\--2)(N\--3)\!}{2\pi^2}\!\! \integrate[\R^6]{\frac{1}{\+p_2^2+p_4^2\,}\,\sqrt{\integrate[\R^{3(N-2)}]{\abs{\FT{\xi}(\vec{p}+\vec{p}_4,\vec{q},\vec{p}_2,\vec{P})}^2; \mspace{-34.5mu}d\vec{p}\+d\vec{q}d\mspace{-0.75mu}\vec{P}}\+};\-d\vec{p}_2\+d\vec{p}_4}\+\times\\[-5pt]
        &\pushright{\times\+\sqrt{\integrate[\R^{3(N-2)}]{\abs{\FT{\xi}(\vec{q}'\-+\vec{p}_2,\vec{p}'\-,\vec{p}_4,\vec{P})}^2; \mspace{-34.5mu}d\vec{p}'\-d\vec{q}'\-d\mspace{-0.75mu}\vec{P}}}}\\
        &= \tfrac{(N\--2)(N\--3)\!}{2\pi^2}\!\! \integrate[\R^6]{\,\frac{\FT{g}(\vec{p}_2)\+\FT{g}(\vec{p}_4)}{p_2^2+p_4^2};\-d\vec{p}_2\+d\vec{p}_4},
    \end{align*}
    where
    $$\FT{g}(\vec{p})\vcentcolon=\sqrt{\integrate[\R^{3(N-2)}]{\abs{\FT{\xi}(\vec{q},\vec{q}'\-,\vec{p},\vec{P})}^2;\mspace{-34.5mu}d\vec{q}d\vec{q}'\-d\mspace{-0.75mu}\vec{P}}\+}\+\in \Lp{2}[\R^3, \sqrt{p^2+1\+}\+d\vec{p}].$$
Therefore,
    $$\abs{\Phi^\lambda_{\mathrm{off};\+0}[\xi]}\leq (N\!-\-2)(N\!-\-3)\norm{\xi}[H^{1/2}(\R^{3(N-1)})]^2\-.$$

\end{proof}
\end{prop}

\n
In the sequel we shall use the following proposition where we  provide the representation of $\Phi^{\lambda}_{\mathrm{diag}}\+$,  $\Phi^\lambda_{\mathrm{off;\+0}}\+$,  $\Phi^\lambda_{\mathrm{off;\+1}}\+$ in position space.
\begin{prop}
    Let $\Phi^\lambda$ be the hermitian quadratic form in $\hilbert*_{N-1}$ defined by~\emph{(\ref{phiDefinition},~\ref{phiComponents})}.\newline
    Then, for all $\xi\in\dom{\Phi^\lambda}$
    \begin{subequations}\label{phisComputed}
    \begin{gather}
        \label{diagPhiPos}
        \Phi^\lambda_{\mathrm{diag}}[\xi]=\-\sqrt{\tfrac{\lambda}{2}}\norm{\xi}[\hilbert*_{N-1}]^2\!+4\pi\!\-\integrate[\R^{3(N-1)}]{;\mspace{-33mu}d\vec{x}d\vec{X}}\!\!\!\integrate[\R^{3(N-1)}]{\abs{\xi(\vec{y},\vec{Y})\--\xi(\vec{x},\vec{X})}^2;\mspace{-33mu}d\vec{y}d\vec{Y}} \,G^\lambda\bigg(\!\!
        \begin{array}{r c l}
            \vec{x}, & \!\!\vec{x}, &\!\!\! \vec{X} \\[-2.5pt]
            \vec{y}, & \!\!\vec{y}, &\!\!\! \vec{Y}
        \end{array}\!\!\!\bigg),\\
       \label{off0PhiPos}
            \begin{split}
            \Phi^\lambda_{\mathrm{off;\+0}}[\xi]=&\,-\+4\pi\+(N\!-\-2)(N\!-\-3)\!\-\integrate[\R^{3(N-1)}]{\conjugate*{\xi(\vec{x},\vec{x}'\-,\vec{x}''\!,\vec{X})\-};\mspace{-33mu}d\vec{x}d\vec{x}'\-d\vec{x}''\!d\vec{X}}\,\times\\[-5pt]
            &\,\times\!\-\integrate[\R^{3(N-1)}]{\xi(\vec{y},\vec{y}'\-,\vec{y}''\!,\vec{Y});\mspace{-33mu}d\vec{y}d\vec{y}'\-d\vec{y}''\!d\vec{Y}}\,G^\lambda\bigg(\!\!\-
            \begin{array}{r c c c l}
                \vec{x}'\-, & \!\!\-\vec{x}''\!, &\!\!\- \vec{x}, & \!\!\- \vec{x}, & \!\!\-\vec{X}\\[-2.5pt]
                \vec{y}, & \!\!\-\vec{y}, &\!\!\-\vec{y}'\-, & \!\!\-\vec{y}''\!, & \!\!\-\vec{Y}\end{array}\!\!\-\bigg),
        \end{split}\\
        \label{off1PhiPos}
            \Phi^\lambda_{\mathrm{off};\+1}[\xi]\-=-\+16\pi\+(N\!-\-2)\!\-\integrate[\R^{3(N-1)}]{\conjugate*{\xi(\vec{x},\vec{x}'\-,\vec{X})\-}\+;\mspace{-33mu}d\vec{x}d\vec{x}'\-d\vec{X}}\!\-\integrate[\R^{3(N-1)}]{\xi(\vec{y},\vec{y}'\-,\vec{Y});\mspace{-33mu}d\vec{y}d\vec{y}'\-d\vec{Y}}\+G^\lambda\bigg(\!\!
            \begin{array}{r c c l}
                \vec{x}, & \!\!\-\vec{x}'\-, &\!\!\- \vec{x}, & \!\!\- \vec{X}\\[-2.5pt]
                \vec{y}, & \!\!\-\vec{y}, &\!\!\-\vec{y}'\-, & \!\!\-\vec{Y} \end{array}\!\!\-\bigg).
    \end{gather}
    \end{subequations}
    \begin{proof}
        In order to prove~\eqref{diagPhiPos} we show that
        \begin{equation*}
        \begin{split}
            \integrate[\R^{3(N-1)}]{;\mspace{-33mu}d\vec{x}d\vec{X}}\!\!\!\integrate[\R^{3(N-1)}]{\abs{\xi(\vec{y},\vec{Y})\--\xi(\vec{x},\vec{X})}^2;\mspace{-33mu}d\vec{y}d\vec{Y}} \,G^\lambda\bigg(\!\!
            \begin{array}{r c l}
                \vec{x}, & \!\!\vec{x}, &\!\!\! \vec{X} \\[-2.5pt]
                \vec{y}, & \!\!\vec{y}, &\!\!\! \vec{Y}
            \end{array}\!\!\!\bigg)\!&=\\
            &\mspace{-72mu}=\tfrac{\sqrt{2}}{8\pi}\!\-\integrate[\R^{3(N-1)}]{\!\left(\-\sqrt{\tfrac{1}{2}\+p^2\-+\-P^2\-+\lambda\+}-\sqrt{\vphantom{\tfrac{1}{1}P^2}\lambda\+}\+\right)\-\abs{\FT{\xi}(\vec{p},\vec{P})}^2;\mspace{-33mu}d\vec{p}\+d\mspace{-0.75mu}\vec{P}}.
        \end{split}
        \end{equation*}
        To this end, the left-hand side of the previous equation can be manipulated as follows
        \begin{align*}
            \integrate[\R^{3(N-1)}]{;\mspace{-33mu}d\vec{z}d\vec{Z}}&\!\!\-\integrate[\R^{3(N-1)}]{\abs{\xi(\vec{y}+\vec{z},\vec{Y}\!\-+\-\vec{Z})-\xi(\vec{y},\vec{Y})}^2;\mspace{-33mu}d\vec{y}d\vec{Y}}\,G^\lambda\bigg(\!\!
            \begin{array}{r c l}
                \vec{0}, & \!\!\vec{0}, &\!\!\! \vec{0} \\[-2.5pt]
                \vec{z}, & \!\!\vec{z}, &\!\!\! \vec{Z}
            \end{array}\!\!\!\bigg)\\
            &=\integrate[\R^{3(N-1)}]{;\mspace{-33mu}d\vec{z}d\vec{Z}}G^\lambda\bigg(\!\!
            \begin{array}{r c l}
                \vec{0}, & \!\!\vec{0}, &\!\!\! \vec{0} \\[-2.5pt]
                \vec{z}, & \!\!\vec{z}, &\!\!\! \vec{Z}
            \end{array}\!\!\!\bigg)\norm{\xi(\vec{y}+\vec{z},\vec{Y}\!\-+\-\vec{Z})-\xi(\vec{y},\vec{Y})}[\Lp{2}[\R^{3(N-1)}\!,\,d\vec{y}\+d\vec{Y}]]^2\\
            &=\integrate[\R^{3(N-1)}]{;\mspace{-33mu}d\vec{z}d\vec{Z}}\!\!\-\integrate[\R^{3(N-1)}]{;\mspace{-33mu}d\vec{p}\+d\mspace{-0.75mu}\vec{P}}\abs{e^{-i\vec{p}\,\cdot\+\vec{z}-i\vec{P}\+\cdot\+\vec{Z}}\!-1}^2\,\abs{\FT{\xi}(\vec{p},\vec{P})}^2\,G^\lambda\bigg(\!\!
            \begin{array}{r c l}
                \vec{0}, & \!\!\vec{0}, &\!\!\! \vec{0} \\[-2.5pt]
                \vec{z}, & \!\!\vec{z}, &\!\!\! \vec{Z}
            \end{array}\!\!\!\bigg),
        \end{align*}
        where in the last step we have used  Plancherel's theorem.
        Thanks to Tonelli's theorem, we can exchange the order of integration, obtaining
        \begin{align*}
            2\!\-\integrate[\R^{3(N-1)}]{\abs{\FT{\xi}(\vec{p},\vec{P})}^2;\mspace{-33mu}d\vec{p}\+d\mspace{-0.75mu}\vec{P}}\!\-&\integrate[\R^{3(N-1)}]{;\mspace{-33mu}d\vec{z}\+d\vec{Z}}[1\!-\-\cos(\vec{p}\-\cdot\!\vec{z}\-+\!\vec{P}\!\cdot\-\vec{Z})]\,G^\lambda\bigg(\!\!
            \begin{array}{r c l}
                \vec{0}, & \!\!\vec{0}, &\!\!\! \vec{0} \\[-2.5pt]
                \vec{z}, & \!\!\vec{z}, &\!\!\! \vec{Z}
            \end{array}\!\!\!\bigg)\\
            &=2\!\-\integrate[\R^{3(N-1)}]{\abs{\FT{\xi}(\vec{p},\vec{P})}^2;\mspace{-33mu}d\vec{p}\+d\mspace{-0.75mu}\vec{P}}\!\-\integrate[\R^{3(N-1)}]{;\mspace{-33mu}d\vec{z}\+d\vec{Z}}[1\!-\-\cos(\vec{p}\-\cdot\!\vec{z}\-+\!\vec{P}\!\cdot\-\vec{Z})]\,\lim_{\vec{r}\to\vec{0}} G^\lambda\bigg(\!\!
            \begin{array}{r c l}
                \tfrac{\vec{r}}{2}, & \!\!\!\--\tfrac{\vec{r}}{2}, &\!\!\! \vec{0} \\[-2.5pt]
                \vec{z}, & \!\!\vec{z}, &\!\!\- \vec{Z}
            \end{array}\!\!\!\bigg).
        \end{align*}
        Because of the monotonicity of $G^\lambda$ and the triangle inequality $\abs{\vec{z}+\tfrac{\vec{r}}{2}}^2\-+\abs{\vec{z}-\tfrac{\vec{r}}{2}}^2\geq 2\+\abs{\vec{z}}^2$, one has
        $$G^\lambda\bigg(\!\!
            \begin{array}{r c l}
                \tfrac{\vec{r}}{2}, & \!\!\!\--\tfrac{\vec{r}}{2}, &\!\!\! \vec{0} \\[-2.5pt]
                \vec{z}, & \!\!\vec{z}, &\!\!\- \vec{Z}
            \end{array}\!\!\!\bigg)\!\leq G^\lambda\bigg(\!\!
        \begin{array}{r c l}
            \vec{0}, & \!\!\vec{0}, &\!\!\! \vec{0} \\[-2.5pt]
            \vec{z}, & \!\!\vec{z}, &\!\!\- \vec{Z}
        \end{array}\!\!\!\bigg),$$
        and therefore the limit can be equivalently computed inside or outside the integral, according to the dominated convergence theorem.
        Hence, due to Proposition~\ref{greenTransformed}
        \begin{align*}
            2\!\-\integrate[\R^{3(N-1)}]{\abs{\FT{\xi}(\vec{p},\vec{P})}^2;\mspace{-33mu}d\vec{p}\+d\mspace{-0.75mu}\vec{P}}\!\lim_{\vec{r}\to\vec{0}}\integrate[\R^{3(N-1)}]{;\mspace{-33mu}d\vec{z}\+d\vec{Z}}[1\!-\-\cos(\vec{p}\-\cdot\!\vec{z}\-+\!\vec{P}\!\cdot\-\vec{Z})]\, G^\lambda\bigg(\!\!
            \begin{array}{r c l}
                \tfrac{\vec{r}}{2}, & \!\!\!\--\tfrac{\vec{r}}{2}, &\!\!\! \vec{0} \\[-2.5pt]
                \vec{z}, & \!\!\vec{z}, &\!\!\- \vec{Z}
            \end{array}\!\!\!\bigg).\\
            =2\!\-\integrate[\R^{3(N-1)}]{\abs{\FT{\xi}(\vec{p},\vec{P})}^2;\mspace{-33mu}d\vec{p}\+d\mspace{-0.75mu}\vec{P}}\lim_{r\to 0} \frac{e^{-\sqrt{\frac{\lambda}{2}}\,r}\!-e^{-\sqrt{\frac{1}{2}\+p^2\,+P^2\,+\,\lambda\+}\:\frac{r}{\!\-\sqrt{2\,}\,}}}{8\pi\,r}.
        \end{align*}
        Evaluating the limit one proves identity~\eqref{diagPhiPos}.\newline
        Next we consider the term $\Phi^\lambda_{\mathrm{off};\+1}\+$.
        Exploiting the unitarity of the Fourier transform one gets
        \begin{align*}
            \Phi^\lambda_{\mathrm{off};\+1}[\xi]&=\! -\tfrac{2\+(N\--2)\!}{\pi^2}\!\! \integrate[\R^{3N}]{\-\frac{\conjugate*{\FT{\xi}(\vec{p}, \vec{p}'\-, \vec{P})\-}\,\FT{\xi}\!\left(\vec{p}'\-+\tfrac{\vec{p}}{2}+\vec{q}, \tfrac{\vec{p}}{2}\--\vec{q}, \vec{P}\right)}{\tfrac{1}{2}p^2+2\+q^2+{p'\+}^2\-+\-P^2+\lambda};\mspace{-11.5mu}d\vec{p}\+d\vec{q} d\vec{p}' d\mspace{-0.75mu}\vec{P}}\\
            &=\! -\tfrac{2\+(N\--2)\!}{\pi^2}\!\! \integrate[\R^{3(N-1)}]{\conjugate*{\xi(\vec{x}, \vec{x}'\-, \vec{X})\-};\mspace{-33mu}d\vec{x} d\vec{x}'\-d\vec{X}}\!\integrate[\R^{3N}]{\frac{e^{i\+\vec{x}\+\cdot\+\vec{p}\++\+i\+\vec{x}'\cdot\+\vec{p}'+\+i\+\vec{X}\cdot\+\vec{P}}\!\-}{(2\pi)^{\frac{3}{2}(N-1)}}\,\frac{\FT{\xi}\!\left(\vec{p}'\-+\tfrac{\vec{p}}{2}+\vec{q}, \tfrac{\vec{p}}{2}\--\vec{q}, \vec{P}\right)}{\tfrac{1}{2}p^2+2\+q^2+{p\+'\+}^2\-+\-P^2+\lambda};\mspace{-11.5mu}d\vec{q}d\vec{p}\+d\vec{p}'\-d\mspace{-0.75mu}\vec{P}}.
        \end{align*}
        Introducing the change of variables
        \begin{equation*}
            \begin{dcases}
                \vec{k}=\tfrac{1}{2}\+\vec{p}+\vec{p}'\-+\vec{q},\\
                \vec{k}'\-=\tfrac{1}{2}\+\vec{p}-\vec{q},\\
                \vec{\kappa}=\tfrac{1}{4}\+\vec{p}-\tfrac{1}{2}\vec{p}'\-+\tfrac{1}{2}\vec{q}
            \end{dcases}\iff \begin{dcases}
                \vec{p}=\tfrac{1}{2}\+\vec{k}+\+\vec{k}'\-+\vec{\kappa},\\
                \vec{p}'\-=\tfrac{1}{2}\+\vec{k}-\vec{\kappa},\\
                \vec{q}=\tfrac{1}{4}\+\vec{k}-\tfrac{1}{2}\+\vec{k}'\-+\tfrac{1}{2}\+\vec{\kappa}
            \end{dcases}
        \end{equation*}
        one can explicitly integrate along the $\vec{\kappa}$-variable (see equation~\eqref{fourierTransformYukawa})
        \begin{align*}
            \Phi^\lambda_{\mathrm{off};\+1}[\xi]=\! -\tfrac{2\+(N\--2)\!}{\pi^2}\!\! \integrate[\R^{3(N-1)}]{\conjugate*{\xi(\vec{x}, \vec{x}'\-, \vec{X})\-};\mspace{-33mu}d\vec{x} d\vec{x}'\-d\vec{X}}\!\integrate[\R^{3N}]{\frac{\FT{\xi}\!\left(\vec{k}, \vec{k}'\-, \vec{P}\right)e^{i\+\vec{k}\+\cdot\left(\!\tfrac{\vec{x}\++\+\vec{x}'\-}{2}\-\right)+\+i\+\vec{k}'\cdot\+\vec{x}\++\+i\+\vec{\kappa}\+\cdot\+(\vec{x}-\vec{x}'\-)\++\+i\+\vec{P}\+\cdot\vec{X}}\!\!\!}{(2\pi)^{\frac{3}{2}(N-1)}\left(2\+\kappa^2\-+\tfrac{1}{2}k^2+{k\+'\+}^2\-+\-P^2+\lambda\right)};\mspace{-11.5mu}d\vec{k}\+d\vec{k}'\-d\vec{P}\+d\vec{\kappa}}\\
            =\! -2\+(N\!-\-2)\!\-\integrate[\R^{3(N-1)}]{\frac{\conjugate*{\xi(\vec{x}, \vec{x}'\-, \vec{X})\-}}{\abs{\vec{x}'\--\vec{x}}};\mspace{-33mu}d\vec{x} d\vec{x}'\-d\vec{X}}\!\integrate[\R^{3(N-1)}]{\frac{\FT{\xi}\!\left(\vec{k}, \vec{k}'\-, \vec{P}\right)e^{i\+\vec{k}\+\cdot\left(\!\tfrac{\vec{x}\++\+\vec{x}'\-}{2}\-\right)+\+i\+\vec{k}'\cdot\+\vec{x}\++\+i\+\vec{P}\+\cdot\vec{X}}\!\!\!}{(2\pi)^{\frac{3}{2}(N-1)}\:e^{\frac{\abs{\vec{x}'\--\vec{x}}}{\sqrt{2}}\+\sqrt{\frac{1}{2}k^2+{k\+'\+}^2+P^2+\lambda\+}\+}\+};\mspace{-33mu}d\vec{k}\+d\vec{k}'\-d\mspace{-0.75mu}\vec{P}}.
        \end{align*}
        Again, because of the unitarity of the Fourier transform
        \begin{equation*}
        \begin{split}
            \Phi^\lambda_{\mathrm{off};\+1}[\xi]=\! -2\+(N\!-\-2)\!\-\integrate[\R^{3(N-1)}]{\frac{\conjugate*{\xi(\vec{x}, \vec{x}'\-, \vec{X})\-}}{\abs{\vec{x}'\--\vec{x}}};\mspace{-33mu}d\vec{x} d\vec{x}'\-d\vec{X}}\!\integrate[\R^{3(N-1)}]{\xi(\vec{y},\vec{y}'\-,\vec{Y});\mspace{-33mu}d\vec{y}d\vec{y}'\-d\vec{Y}}\,\times\\
            \times\!\-\integrate[\R^{3(N-1)}]{\frac{e^{i\+\vec{k}\+\cdot\left(\-\vec{y}\+-\+\tfrac{\vec{x}\++\+\vec{x}'\-}{2}\-\right)+\+i\+\vec{k}'\cdot\+(\vec{y}'-\+\vec{x})\++\+i\+\vec{P}\+\cdot\+(\vec{Y}\--\vec{X})}\!\!\!}{(2\pi)^{3(N-1)}\:e^{\frac{\abs{\vec{x}'\--\vec{x}}}{\sqrt{2}}\+\sqrt{\frac{1}{2}k^2+{k\+'\+}^2+P^2+\lambda\+}\+}\+};\mspace{-33mu}d\vec{k}\+d\vec{k}'\-d\mspace{-0.75mu}\vec{P}}\+.
        \end{split}
        \end{equation*}
        Notice that the last integration is precisely the anti-Fourier transform of the right-hand side of the identity in Proposition~\ref{greenTransformed} with parameters $(\vec{x}_i,\vec{x}_j)=(\vec{x},\vec{x}')$ and $\vec{X}_{\-\sigma}\!=(\vec{x},\vec{X})$, which provides~\eqref{off1PhiPos}.\newline
        We proceed in a similar way to prove~\eqref{off0PhiPos}.
        \begin{align*}
            \Phi^\lambda_{\mathrm{off};\+0}[\xi]\!&=\!-\tfrac{(N\--2)(N\--3)\!}{2\pi^2}\!\!\integrate[\R^{3N}]{\-\frac{\conjugate*{\-\FT{\xi}(\vec{p},\vec{p}'\-,\vec{p}''\!,\vec{P}\-)\!}\;\FT{\xi}\!\left(\vec{p}'\!+\-\vec{p}''\!,\vec{q}+\tfrac{\vec{p}}{2},\tfrac{\vec{p}}{2}-\vec{q},\vec{P}\-\right)\!}{2\+q^2\-+\tfrac{1}{2}p^2\-+{p\+'\+}^2\!+{p\+''\+}^2\!+\-P^2+\lambda};\mspace{-11.5mu}d\vec{p}\+d\vec{q}d\vec{p}'\-d\vec{p}'' d\mspace{-0.75mu}\vec{P}}\\
            &=\!-\tfrac{(N\--2)(N\--3)\!}{2\pi^2}\!\!\integrate[\R^{3(N-1)}]{\conjugate*{\xi(\vec{x},\vec{x}'\-,\vec{x}''\!,\vec{X})\-};\mspace{-33mu}d\vec{x}d\vec{x}'\-d\vec{x}''\!d\vec{X}}\:\times\\
            &\mspace{94.5mu}\times\!\-\integrate[\R^{3N}]{\frac{e^{i\+\vec{x}\+\cdot\+\vec{p}\++\+i\+\vec{x}'\cdot\+\vec{p}'+\+i\+\vec{x}''\-\cdot\+\vec{p}''\-+\+i\+\vec{X}\cdot\+\vec{P}}\!\-}{(2\pi)^{\frac{3}{2}(N-1)}}\,\frac{\FT{\xi}\!\left(\vec{p}'\!+\-\vec{p}''\!,\vec{q}+\tfrac{\vec{p}}{2},\tfrac{\vec{p}}{2}-\vec{q},\vec{P}\-\right)\!}{2\+q^2\-+\tfrac{1}{2}p^2\-+{p\+'\+}^2\!+{p\+''\+}^2\!+\-P^2+\lambda}{};\mspace{-11.5mu}d\vec{q}d\vec{p}\+d\vec{p}'\-d\vec{p}''d\mspace{-0.75mu}\vec{P}}.
        \end{align*}
        In this case we consider the following change of variables
        \begin{equation*}
            \begin{dcases}
                \vec{k}=\vec{p}'\-+\vec{p}''\!,\\
                \vec{k}'\-=\tfrac{1}{2}\+\vec{p}+\vec{q},\\
                \vec{k}''\!=\tfrac{1}{2}\+\vec{p}-\vec{q},\\
                \vec{\kappa}=\tfrac{\vec{p}'-\+\vec{p}''\!}{2}
            \end{dcases}\iff \begin{dcases}
                \vec{p}=\vec{k}'\-+\vec{k}''\!,\\
                \vec{p}'\-=\tfrac{1}{2}\+\vec{k}+\vec{\kappa},\\
                \vec{p}''\!=\tfrac{1}{2}\+\vec{k}-\vec{\kappa},\\
                \vec{q}=\tfrac{\vec{k}'-\+\vec{k}''\!}{2}
            \end{dcases}
        \end{equation*}
        so that also in this situation we can integrate along $\vec{\kappa}$
        \begin{align*}
            \Phi^\lambda_{\mathrm{off};\+0}[\xi]\!&=\!-\tfrac{(N\--2)(N\--3)\!}{2\pi^2}\!\!\integrate[\R^{3(N-1)}]{\conjugate*{\xi(\vec{x},\vec{x}'\-,\vec{x}''\!,\vec{X})\-};\mspace{-33mu}d\vec{x}d\vec{x}'\-d\vec{x}''\!d\vec{X}}\:\times\\
            &\mspace{94.5mu}\times\!\-\integrate[\R^{3N}]{\frac{\FT{\xi}(\vec{k},\vec{k}'\-,\vec{k}''\!,\vec{P}\-)\,e^{i\+\vec{k}\+\cdot\left(\!\tfrac{\vec{x}''\-+\+\vec{x}'\-}{2}\-\right)+\+i\+\vec{k}'\cdot\+\vec{x}\++\+i\+\vec{k}''\-\cdot\+\vec{x}\++\+i\+\vec{\kappa}\+\cdot(\vec{x}'-\+\vec{x}''\-)\++\+i\+\vec{P}\+\cdot\vec{X}}\!\!\!}{(2\pi)^{\frac{3}{2}(N-1)}\-\left(2\+\kappa^2\-+\tfrac{1}{2}k^2\-+{k\+'\+}^2\!+{k\+''\+}^2\!+\-P^2+\lambda\right)}{};\mspace{-11.5mu}d\vec{k}\+d\vec{k}'\-d\vec{k}''\!d\vec{P}\+d\vec{\kappa}}\\
            &=\!-\tfrac{(N\--2)(N\--3)}{2}\!\!\integrate[\R^{3(N-1)}]{\frac{\conjugate*{\xi(\vec{x},\vec{x}'\-,\vec{x}''\!,\vec{X})\-}}{\abs{\vec{x}''\--\vec{x}'}};\mspace{-33mu}d\vec{x}d\vec{x}'\-d\vec{x}''\!d\vec{X}}\:\times\\
            &\mspace{94.5mu}\times\!\-\integrate[\R^{3(N-1)}]{\frac{\FT{\xi}(\vec{k},\vec{k}'\-,\vec{k}''\!,\vec{P}\-)\,e^{i\+\vec{k}\+\cdot\left(\!\tfrac{\vec{x}''\-+\+\vec{x}'\-}{2}\-\right)+\+i\+\vec{k}'\cdot\+\vec{x}\++\+i\+\vec{k}''\-\cdot\+\vec{x}\++\+i\+\vec{P}\+\cdot\vec{X}}\!\!\!}{(2\pi)^{\frac{3}{2}(N-1)}\:e^{\frac{\abs{\vec{x}''\--\+\vec{x}'\-}}{\sqrt{2}}\+\sqrt{\frac{1}{2}k^2+{k\+'\+}^2+{k\+''\+}^2+P^2+\lambda\+}}}{};\mspace{-33mu}d\vec{k}\+d\vec{k}'\-d\vec{k}''\!d\mspace{-0.75mu}\vec{P}}.
        \end{align*}
        Hence, exploiting the unitarity of the Fourier transform, one recognizes once again an anti-Fourier transform of the right-hand side of Proposition~\ref{greenTransformed} with parameters $(\vec{x}_i,\vec{x}_j)=(\vec{x}'\-,\vec{x}'')$ and $\vec{X}_{\-\sigma}\!=(\vec{x},\vec{x},\vec{X})$.

    \end{proof}
\end{prop}
\begin{note}\label{negativeContributionsHighlighted}
    By the above proposition and the explicit integration~\eqref{greenIntegrated}, we can isolate an unbounded negative contribution of the off-diagonal terms, \ie
    \begin{subequations}
    \begin{align}
        \Phi^\lambda_{\mathrm{off;\+0}}[\xi]= &-\tfrac{(N-2)(N-3)}{2}\!\-\integrate[\R^{3(N-1)}]{\frac{\:e^{-\sqrt{\-\frac{\lambda}{2}}\,\abs{\vec{x}'-\+\vec{x}''}}\!\-}{\abs{\vec{x}'\--\vec{x}''}}\,\abs{\xi(\vec{x},\vec{x}'\-,\vec{x}''\!,\vec{X})}^2;\mspace{-33mu}d\vec{x}d\vec{x}'\-d\vec{x}''\!d\vec{X}}\++\label{off0PhiSplitted}\\
        &+2\pi\+(N\!-\-2)(N\!-\-3)\!\-\integrate[\R^{3(N-1)}]{;\mspace{-33mu}d\vec{x}d\vec{x}'\-d\vec{x}''\!d\vec{X}}\!\!\!\integrate[\R^{3(N-1)}]{\abs{\xi(\vec{y},\vec{y}'\!,\vec{y}''\!,\vec{Y})\--\xi(\vec{x},\vec{x}'\!,\vec{x}''\!,\vec{X})}^2;\mspace{-33mu}d\vec{y}d\vec{y}'\-d\vec{y}''\!d\vec{Y}}\+\times\nonumber\\[-5pt]
        &\mspace{408mu}\times\-G^\lambda\bigg(\!\!\-
            \begin{array}{r c c c l}
                \vec{x}'\!, & \!\!\-\vec{x}''\!, &\!\!\- \vec{x}, & \!\!\- \vec{x}, & \!\!\!\vec{X}\\[-2.5pt]
                \vec{y}, & \!\!\-\vec{y}, &\!\!\-\vec{y}'\!, & \!\!\-\vec{y}''\!, & \!\!\!\vec{Y}
            \end{array}\!\!\-\bigg),\nonumber
        \end{align}
        \vspace{-0.625cm}
        \begin{align}
        \Phi^\lambda_{\mathrm{off};\+1}[\xi]\-= &-\+2(N\!-\!2)\!\-\integrate[\R^{3(N-1)}]{\frac{\:e^{-\sqrt{\-\frac{\lambda}{2}}\,\abs{\vec{x}-\+\vec{x}'}}\!\-}{\abs{\vec{x}-\vec{x}'}}\,\abs{\xi(\vec{x},\vec{x}'\-,\vec{X})}^2;\mspace{-33mu}d\vec{x}d\vec{x}'\-d\vec{X}}\++\label{off1PhiSplitted}\\[-5pt]
        &+8\pi\+(N\!-\-2)\!\-\integrate[\R^{3(N-1)}]{;\mspace{-33mu}d\vec{x}d\vec{x}'\-d\vec{X}}\!\!\!\integrate[\R^{3(N-1)}]{\abs{\xi(\vec{y},\vec{y}'\-,\vec{Y})\--\xi(\vec{x},\vec{x}'\-,\vec{X})}^2;\mspace{-33mu}d\vec{y}d\vec{y}'\-d\vec{Y}}\+G^\lambda\bigg(\!\!
            \begin{array}{r c c l}
                \vec{x}, & \!\!\-\vec{x}'\-, &\!\!\- \vec{x}, & \!\!\! \vec{X}\\[-2.5pt]
                \vec{y}, & \!\!\-\vec{y}, &\!\!\-\vec{y}'\-, & \!\!\!\vec{Y}
            \end{array}\!\!\-\bigg).\nonumber
    \end{align}
    \end{subequations}
\end{note}
\begin{note}\label{4BodyNeed}
    Formulas~\eqref{off1PhiSplitted} and~\eqref{off0PhiSplitted} show that the off-diagonal terms $\Phi^{\lambda}_{\mathrm{off};\+1}$ and $\Phi^{\lambda}_{\mathrm{off};\+0}$ of the quadratic form $\Phi^{\lambda}$ contain a negative contribution characterized by the same Coulomb-type singularity.
    Such negative terms must be compensated in order to prove that $\Phi^\lambda$ is closed and coercive for $\lambda$ large enough, and therefore the resulting Hamiltonian is stable.
    This is the technical reason why we introduced the regularizing three-body $\Phi^{(3)}_{\mathrm{reg}}$ and four-body $\Phi^{(4)}_{\mathrm{reg}}$ terms in the quadratic form (see~\eqref{regPhi}), which in turn come from the last two terms in~\eqref{alfa}.
    Note that for $N\!=\-3$ the terms $\Phi^{\lambda}_{\mathrm{off};\+0}$ and $\Phi^{(4)}_{\mathrm{reg}}$ are absent, and it is known that the quadratic form is closed and coercive if $\Phi^{(3)}_{\mathrm{reg}}$ is present. 
    By analogy, it is reasonable to expect that for $N \-\geq\- 4$ the singular negative contribution of $\Phi^{\lambda}_{\mathrm{off};\+0}$ must be compensated by a term like $\Phi^{(4)}_{\mathrm{reg}}\+$.\newline
    In principle, one might hope that the three-body term $\Phi^{(3)}_{\mathrm{reg}}$ would be sufficient to guarantee stability even for $N \geq 4$.
    However, at least heuristically, it seems to us that this is hard to believe, since the two singular terms $\Phi^{\lambda}_{\mathrm{off};\+1}$,  $\Phi^{\lambda}_{\mathrm{off};\+0}$ have the same kind of singularity but, on the other hand, $\Phi^{\lambda}_{\mathrm{off};\+0}$ is a term of order $N^2$ while $\Phi^{(3)}_{\mathrm{reg}}$ grows as $N$.
    So, at least for $N$ large, $\Phi^{(3)}_{\mathrm{reg}}$ does not seem sufficient to guarantee stability.\newline
    It would be interesting to give a rigorous proof of the necessity of $\Phi^{(4)}_{\mathrm{reg}}\+$ for the stability, but at present we are unable to exhibit such a result.\newline
    It is notable that the same regularizing terms also appear in the approach via Dirichlet forms (see~\eqref{Dirichlet4Body}).
    We also stress that no further regularizing term associated with a $N$-body interaction is required if $N\-\geq\- 5$. 
\end{note}
\n We conclude this section by providing a lower bound for $\Phi^\lambda$ in terms of the diagonal contribution $\Phi^\lambda_{\mathrm{diag}}\+$.

\begin{lemma}\label{neededToClosedness}
    Let $\Phi^\lambda$ be the hermitian quadratic form in $\hilbert*_{N-1}$ defined by~\emph{(\ref{phiDefinition},~\ref{phiComponents})} and $\alpha^-_0\vcentcolon=\max\{0,-\alpha_0\}$.
    Then
    \begin{equation}\label{firstPhiLowerBound}
        \Phi^\lambda[\xi]\geq 4\pi N(N\!-\-1)\!\left[1\--\Lambda_N-\sqrt{\tfrac{2}{\lambda}\+}\+\alpha^-_0-\tfrac{(N+1)(N-2)\gamma}{\sqrt{8\+\lambda}\, b}\right]\-\Phi^\lambda_{\mathrm{diag}}[\xi],\quad\forall \lambda\->\-0,\, \xi\-\in\-\dom{\Phi^\lambda},
    \end{equation}
    with $\Lambda_N$ given by~\eqref{simplerLambda}.
    Furthermore, $\Phi^\lambda$ is coercive in $H^{\frac{1}{2}}(\R^{3(N-1)})$ whenever $\lambda\->\-\lambda^\ast_0$ with
    \begin{equation}\label{firstLambdaNote}
        \lambda^\ast_0\vcentcolon=\frac{2\left[\alpha^-_0+(N\!+\-1)(N\!-\-2)\,\gamma/(4\+b)\right]^2}{(1\--\Lambda_N)^2}
    \end{equation}
    provided $\gamma>\gamma_c$ defined by~\eqref{criticalGammaN}.
    \begin{proof}
        Taking account of Remark~\ref{negativeContributionsHighlighted} and the lower bound for $\theta$ provided by assumption~\eqref{positiveBoundedCondition}, we get
        \begin{align*}
            (\Phi^{(3)}_{\mathrm{reg}}\!+\-\Phi^\lambda_{\mathrm{off;\+1}})[\xi]\mspace{-2.25mu}\geq \mspace{-0.75mu}(N\!-\-2)\!\-&\integrate[\R^{3(N-1)}]{\frac{\gamma\--\-2\+e^{-\sqrt{\frac{\lambda}{2}\+}\,\abs{\vec{x}\+-\+\vec{x}'}}\nquad}{\abs{\vec{x}\--\-\vec{x}'}}\:\abs{\xi(\vec{x},\vec{x}'\-,\vec{X}\-)}^2;\mspace{-33mu}d\vec{x}d\vec{x}'\-d\vec{X}}-\tfrac{(N\--2)\+\gamma}{b}\norm{\xi}[\hilbert*_{N-1}]^2\!,\\
            (\Phi^{(4)}_{\mathrm{reg}}\!+\-\Phi^\lambda_{\mathrm{off;\+0}})[\xi]\mspace{-2.25mu}\geq \mspace{-0.75mu}\tfrac{\-(N\--2)(N\--3)\-}{4}\!\!&\integrate[\R^{3(N-1)}]{\frac{\gamma\--\-2\+e^{-\sqrt{\frac{\lambda}{2}\+}\,\abs{\vec{x}'-\+\vec{x}''}}\nquad}{\abs{\vec{x}'\--\vec{x}''}}\:\abs{\xi(\vec{x},\vec{x}'\-,\vec{x}''\!,\vec{X}\-)}^2;\mspace{-33mu}d\vec{x}d\vec{x}'\-d\vec{x}''\!d\vec{X}}\--\tfrac{(N\--2)(N\--3)\+\gamma}{4\+b}\norm{\xi}[\hilbert*_{N-1}]^2\!.
        \end{align*}
        This implies
        \begin{align*}
            (\Phi_{\mathrm{reg}}\!+\Phi^\lambda_{\mathrm{off;\+0}}\-+\Phi^\lambda_{\mathrm{off;\+1}})[\xi]\-
            \geq (N\!-\-2)&\min\{0,\gamma\--\-2\}\mspace{9mu}\mspace{-12mu}\left[\integrate[\R^{3(N-1)}]{\frac{1}{r}\left\lvert\xi\!\left(\tfrac{1}{3}\+\vec{r}+\vec{r}'\-,\vec{r}'\--\tfrac{2}{3}\+\vec{r},\vec{X}\right)\-\right\rvert^2;\mspace{-33mu}d\vec{r}d\vec{r}'\-d\vec{X}}\-+\right.\\
            &\left.+\+\frac{\!(N\!-\-3)\!}{4}\!\-\integrate[\R^{3(N-1)}]{\frac{1}{r\+'}\left\lvert\xi\!\left(\vec{x},\vec{r}''\!+\tfrac{1}{2}\+\vec{r}'\-,\vec{r}''\!-\tfrac{1}{2}\+\vec{r}'\-,\vec{X}\right)\-\right\rvert^2;\mspace{-33mu}d\vec{x}d\vec{r}'\-d\vec{r}''\!d\vec{X}}\right]+\\
            &\pushright{+\!\left[\alpha_0-\frac{\!(N\!+\-1)(N\!-\-2)\+\gamma}{4\,b}\right]\!\norm{\xi}[\hilbert*_{N-1}]^2}.
        \end{align*}
        We recall that in order to compute the Fourier transform of a function $f(\vec{x},\vec{y})\-\in\-\Lp{2}[\R^d\-\otimes\-\R^d\-,\,d\vec{x}d\vec{y}]$ composed with a $2\-\times\-2$ non-singular complex matrix $A$ applied to its arguments, apart from a Jacobian factor, one needs to compose the Fourier transform of $f$ with the inverse of the transpose of $A$, namely
        \begin{equation}
            \widehat{f\circ A}=\abs{\det{A}}^d\,\FT{f}\circ(A^\intercal)^{-1}.
        \end{equation}
        Therefore, exploiting the Hardy-Rellich inequality~\eqref{Hardy-Rellich} one gets
        \begin{align*}
            (\Phi_{\mathrm{reg}}\!+\Phi^\lambda_{\mathrm{off;\+0}}\-&+\Phi^\lambda_{\mathrm{off;\+1}})[\xi]\-\geq \tfrac{(N\--\+2)\,\pi}{2}\min\{0,\gamma\--\-2\}\!\left[\integrate[\R^{3(N-1)}]{\+q\-\left\lvert\FT{\xi}\!\left(\vec{q}+\tfrac{2}{3}\+\vec{q}'\-,\tfrac{1}{3}\+\vec{q}'\--\vec{q},\vec{P}\right)\-\right\rvert^2;\mspace{-33mu}d\vec{q}d\vec{q}'\-d\mspace{-0.75mu}\vec{P}}\!\-+\right.\\
            &\left.+\+\tfrac{(N\--\+3)}{4}\!\-\integrate[\R^{3(N-1)}]{\+q'\-\left\lvert\FT{\xi}\!\left(\vec{q},\vec{q}'\-+\tfrac{\vec{q}''}{2},\tfrac{\vec{q}''}{2}-\vec{q}'\-,\vec{P}\right)\-\right\rvert^2;\mspace{-33mu}d\vec{q}d\vec{q}'d\vec{q}''\-d\mspace{-0.75mu}\vec{P}}\right]\-+\!\left[\alpha_0-\tfrac{\!(N\-+1)(N\--2)\+\gamma}{4\+b}\right]\!\norm{\xi}[\hilbert*_{N-1}]^2\\
            =&\,\tfrac{(N\--\+2)\,\pi}{2}\min\{0,\gamma\--\-2\}\!\left[\integrate[\R^{3(N-1)}]{\+\tfrac{\abs{\vec{p}\+-2\vec{p}'}}{3}\+\abs{\FT{\xi}(\vec{p},\vec{p}'\-,\vec{P})}^2;\mspace{-33mu}d\vec{p}\+d\vec{p}'\-d\mspace{-0.75mu}\vec{P}}+\right.\\
            &\left.+\+\tfrac{(N\--\+3)}{4}\!\-\integrate[\R^{3(N-1)}]{\tfrac{\abs{\vec{p}'-\+\vec{p}''}}{2}\+\abs{\FT{\xi}(\vec{p},\vec{p}'\-,\vec{p}''\!,\vec{P})}^2;\mspace{-33mu}d\vec{p}\+d\vec{p}'\-d\vec{p}''\!d\mspace{-0.75mu}\vec{P}}\right]\-+\!\left[\alpha_0-\tfrac{\!(N\-+1)(N\--2)\+\gamma}{4\+b}\right]\!\norm{\xi}[\hilbert*_{N-1}]^2\-.
        \end{align*}
        Moreover, taking into account the elementary (sharp) inequalities
        \begin{align}
            a+2b\leq \sqrt{6}\+\sqrt{\tfrac{\,a^2\!}{2}+b^2}, && a+b\leq \sqrt{2}\+\sqrt{a^2+b^2}, \qquad \forall a,b>0            
        \end{align}
        one can deduce
        \begin{align*}
            (\Phi_{\mathrm{reg}}\!+\Phi^\lambda_{\mathrm{off;\+0}}\-&+\Phi^\lambda_{\mathrm{off;\+1}})[\xi]\-\geq\,\tfrac{(N\--2)\,\pi}{2}\min\{0,\gamma\--\-2\}\!\left[\sqrt{\tfrac{2}{3}\+}\!\-\integrate[\R^{3(N-1)}]{\sqrt{\tfrac{p^2}{2}+{p\+'\+}^2\+}\+\abs{\FT{\xi}(\vec{p},\vec{p}'\-,\vec{P})}^2;\mspace{-33mu}d\vec{p}\+d\vec{p}'\-d\mspace{-0.75mu}\vec{P}}+\right.\\
            &\left.+\+\tfrac{(N\--3)}{4\sqrt{2}}\!\-\integrate[\R^{3(N-1)}]{\sqrt{{p\+'\+}^2\!+{p\+''\+}^2\+}\+\abs{\FT{\xi}(\vec{p},\vec{p}'\-,\vec{p}''\!,\vec{P})}^2;\mspace{-33mu}d\vec{p}\+d\vec{p}'\-d\vec{p}''\!d\mspace{-0.75mu}\vec{P}}\right]\-+\!\left[\alpha_0-\tfrac{\!(N\-+1)(N\--2)\+\gamma}{4\+b}\right]\!\norm{\xi}[\hilbert*_{N-1}]^2\-.
        \end{align*}
        Hence,
        \begin{equation}\label{actuallyUsefulLowerBound} (\Phi_{\mathrm{reg}}\!+\Phi^\lambda_{\mathrm{off;\+0}}\-+\Phi^\lambda_{\mathrm{off;\+1}})[\xi]\-\geq -\Lambda_N\,\Phi^0_{\mathrm{diag}}[\xi]+\!\left[\alpha_0-\tfrac{\!(N\-+1)(N\--2)\,\gamma}{4\+b}\right]\!\norm{\xi}[\hilbert*_{N-1}]^2
        \end{equation}
        which implies
        \begin{equation*}
            (\Phi_{\mathrm{reg}}\!+\Phi^\lambda_{\mathrm{off;\+0}}\-+\Phi^\lambda_{\mathrm{off;\+1}})[\xi]\geq -\Lambda_N\,\Phi^\lambda_{\mathrm{diag}}[\xi]-\!\left[\alpha^-_0+\tfrac{\!(N\-+1)(N\--2)\,\gamma}{4\+b}\right]\!\sqrt{\tfrac{2}{\lambda}\+}\,\Phi^\lambda_{\mathrm{diag}}[\xi]
        \end{equation*}
        proving the first statement.
        Concerning the coercivity in $H^{\frac{1}{2}}(\R^{3(N-1)})$, the result can be obtained by noticing that $\Phi^\lambda_{\mathrm{diag}}$ defines an equivalent norm to $H^{\frac{1}{2}}(\R^{3(N-1)})$
        so that inequality~\eqref{firstPhiLowerBound} provides the result assuming $\Lambda_N<1$ (that holds if and only if $\gamma\->\-\gamma_c$ is fulfilled) and $\lambda$ large enough.

    \end{proof}
\end{lemma}

\comment{
To this end, it is convenient to introduce the Gagliardo semi-norm of the Sobolev space $H^{\frac{1}{2}}(\R^n)$, defined as
\begin{equation}\label{GagliardoHilbertSobolev}
    [\+u\+]^2_{\frac{1}{2}}\-\vcentcolon=\!\-\integrate[\R^{2n}]{\frac{\abs{u(\vec{x})\--u(\vec{y})}^2\!\!}{\;\abs{\vec{x}-\vec{y}}^{n+1}}; \!\!d\vec{x}d\vec{y}},\qquad u\in H^{\frac{1}{2}}(\R^n),
\end{equation}
so that $\norm{u}[H^{1/2}(\R^n)]^2\!\-=\norm{u}^2\!+[\+u\+]^2_{\frac{1}{2}}$.
In terms of the Fourier transform of $u$ one also has (see \eg,~\cite[proposition 3.4]{DNPV} or~\cite[section 7.12 (4)]{LL})
\begin{equation}\label{GagliardoHilbertSobolevFourier}
   [\+u\+]^2_{\frac{1}{2}}\!=\frac{\;2\+\pi^{\frac{n+1}{2}}\!}{\Gamma\!\left(\frac{n\,+\+1}{2}\-\right)}\-\integrate[\R^n]{\abs{\vec{k}}\abs{\FT{u}(\vec{k})}^2;\-d\vec{k}}.
\end{equation}
}

\section{Proof of the Main Results}\label{proveMain}
We are now ready to prove the results stated in Section~\ref{main}.

\begin{proof}[Proof of Proposition~\ref{closednessTheo}]
First, let us show the lower bound for $\Phi^\lambda$.
Let us rewrite inequality~\eqref{actuallyUsefulLowerBound} in the following way
\begin{equation*}
        \Phi^\lambda[\xi]\!\geq \- \tfrac{4\pi N (N-1)\!}{\sqrt{2}}\!\!\-\integrate[\R^{3(N-1)}]{\!\-\left[\-\sqrt{\tfrac{1}{2}\+p^2\!+\!P^2\!+\!\lambda}-\-\Lambda_N\sqrt{\tfrac{1}{2}\+p^2\!+\!P^2}\-+\-\sqrt{2\+}\+\alpha_0\--\-\tfrac{(N\-+1)(N\--2)\+\gamma}{2\+\sqrt{2\+}\+b}\right]\!\abs{\FT{\xi}(\vec{p},\-\vec{P})}^2;\mspace{-33mu}d\vec{p}\+d\mspace{-0.75mu}\vec{P}}\!.
\end{equation*}
One can check that, as long as $\Lambda_N\-<\-1$, the integrand attains its minimum when
\begin{equation*}
    \sqrt{\tfrac{1}{2}\+p^2\-+\!P^2}=\Lambda_N\+\sqrt{\frac{\lambda}{1\--\-\Lambda_N^2\!\-}\,}\,,
\end{equation*}
therefore
\begin{align*}
        \Phi^\lambda[\xi]\-&\geq \tfrac{4\pi N (N\--1)\!}{\sqrt{2}}\!\-\integrate[\R^{3(N-1)}]{\!\left[\sqrt{\lambda}\,\sqrt{1\--\-\Lambda_N^2}+\sqrt{2\+}\+\alpha_0-\tfrac{(N\-+1)(N\--2)\+\gamma}{2\+\sqrt{2\+}\+b}\right]\!\abs{\FT{\xi}(\vec{p},\vec{P})}^2;\mspace{-33mu}d\vec{p}\+d\mspace{-0.75mu}\vec{P}}\\
        &=\tfrac{4\pi N (N\--1)\!}{\sqrt{2}}\!\left[\sqrt{\lambda}\,\sqrt{1\--\-\Lambda_N^2}+\sqrt{2\+}\+\alpha_0-\tfrac{(N\-+1)(N\--2)\+\gamma}{2\+\sqrt{2\+}\+b}\right]\!\norm{\xi}[\hilbert*_{N-1}]^2\!.
\end{align*}
Hence, we have proved that $\Phi^\lambda$ is coercive for any $\lambda\->\-\lambda_0$ with
\begin{equation}
    \lambda_0\vcentcolon=\frac{2\max\{0,-\alpha_0+(N\!+\-1)(N\!-\-2)\,\gamma/(4\+b)\}^2\!\-}{1\--\-\Lambda_N^2}\,.
\end{equation}
Observe that $\lambda_0<\lambda_0^\ast$, with $\lambda_0^\ast$ given by~\eqref{firstLambdaNote}.\newline
Next, we focus on the closedness of the quadratic form.
In light of Proposition~\ref{neededToClosedness}, we already know that $\Phi^\lambda$ is closed for any $\lambda\->\-\lambda^\ast_0\+$; however, we are going to extend this result for any value of $\lambda\->\-0$.
Since $\Phi^\lambda$ is hermitian, the quantity
\begin{equation}
    \define{\norm{\cdot}[\Phi^\lambda]^2;\Phi^\lambda[\+\cdot\+]+4\pi N(N\!-\-1)\!\left[\sqrt{\tfrac{\lambda}{2}\+}\-\left(1\!-\-\sqrt{1\--\-\Lambda_N^2}\right)\--\alpha_0+\tfrac{(N\-+1)(N\--2)\,\gamma}{4\+b}\right]\!\norm{\cdot}[\hilbert*_{N-1}]^2}
\end{equation}
defines a norm in $\hilbert*_{N-1}$ for any fixed $\lambda\->\-0$.
Therefore, let us show that $H^{\frac{1}{2}}(\R^{3(N-1)})$ is closed under this newly defined norm.
Let $\psi\in\hilbert*_{N-1}$ and $\{\psi_n\}_{n\+\in\+\N}\subset \dom{\Phi^\lambda}$ be, respectively, a vector and a Cauchy sequence for $\norm{\cdot}[\Phi^\lambda]\+$ such that $\norm{\psi_n\--\psi}[\hilbert*_{N-1}]\!\!\longrightarrow 0$.
Therefore, since inequality~\eqref{actuallyUsefulLowerBound} reads
\begin{equation*}
    \Phi^\lambda[\+\cdot\+]\geq 4\pi N(N\!-\-1)\!\left\{(1-\Lambda_N)\+
    \Phi^\lambda_{\mathrm{diag}}[\+\cdot\+]+\-\left[\alpha_0-\tfrac{(N\-+1)(N\--2)\,\gamma}{4\+b}\right]\-\norm{\cdot}[\hilbert*_{N-1}]^2\+\right\}\!,\qquad \lambda>0
\end{equation*}
one has
\begin{equation*}
    (1-\Lambda_N)\+
    \Phi^\lambda_{\mathrm{diag}}[\psi_n\--\psi_m]+\sqrt{\tfrac{\lambda}{2}\+}\-\left(1\--\-\sqrt{1\--\-\Lambda_N^2}\right)\!\norm{\psi_n\--\psi_m}[\hilbert*_{N-1}]^2\-\leq\frac{\:\norm{\psi_n\--\psi_m}[\Phi^\lambda]^2\!\!}{4\pi N(N\!-\-1)}\longrightarrow 0.
\end{equation*}
The left-hand side is composed of positive terms, the second of which vanishes in the limit.
This means that
\begin{equation*}
    \Phi^\lambda_{\mathrm{diag}}[\psi_n\--\psi_m]\longrightarrow 0.
\end{equation*}
In other words, $\{\psi_n\}_{n\+\in\+\N}$ is a Cauchy sequence also in $H^{\frac{1}{2}}(\R^{3(N-1)})$ and, because of the uniqueness of the limit, it converges to $\psi\!\in\!H^{\frac{1}{2}}(\R^{3(N-1)})$.
We have shown that $(\dom{\Phi^\lambda}, \norm{\cdot}[\Phi^\lambda])$ is a Banach space for any $\lambda\->\-0$, therefore $\Phi^\lambda$ is closed regardless of the value of $\lambda\->\-0$.

\end{proof}

\begin{prop}\label{GammaDomain}
    Let $\lambda\->\-0$ and $\Gamma^\lambda\-\in\-\linear{\hilbert*_{N-1}}$ be uniquely defined by the hermitian, closed and lower semi-bounded quadratic form $\Phi^\lambda$ in $\hilbert*_{N-1}$ given by~(\ref{phiDefinition},~\ref{phiComponents}), \ie
    \begin{equation}
        \Gamma^\lambda=4\pi N(N\!-\-1)\big(\Gamma^\lambda_{\!\mathrm{diag}}\! +\Gamma^\lambda_{\!\mathrm{off};\+0}+\Gamma^\lambda_{\!\mathrm{off};\+1}\-+\Gamma_{\!\mathrm{reg}}\big)
    \end{equation}
    with
    \begin{subequations}\label{GammaDefs}
    \begin{gather}
        (\widehat{\Gamma^\lambda_{\!\mathrm{diag}}\+\xi})(\vec{k},\vec{p}_1,\ldots,\vec{p}_{N-2})=\tfrac{1}{\!\-\sqrt{2\+}\+}\sqrt{\tfrac{1}{2}\+k^2\-+p_1^2+\ldots+p_{N-2}^2+\lambda\,}\,\FT{\xi}(\vec{k},\vec{p}_1,\ldots,\vec{p}_{N-2}),\label{diagGamma}\\(\widehat{\Gamma^\lambda_{\!\mathrm{off};\+1}\+\xi})(\vec{k},\vec{p}_1,\ldots,\vec{p}_{N-2})= -\tfrac{2}{\:\pi^2\!}\sum_{i\+=\+1}^{N\--2}\-\integrate[\R^3]{\frac{\FT{\xi}(\vec{k}-\vec{q}+\vec{p}_i,\vec{q},\vec{p}_1,\ldots\check{\vec{p}}_i\ldots,\vec{p}_{N-2})}{\abs{\vec{k}\--\-\vec{q}}^2+q^2+p_1^2\-+\-\ldots+p_{N-2}^2+\lambda};\-d\vec{q}},\label{off1Gamma}\\
        (\widehat{\Gamma^\lambda_{\!\mathrm{off};\+0}\+\xi})(\vec{k},\vec{p}_1,\ldots,\vec{p}_{N-2})= -\tfrac{1}{\:\pi^2\!}\!\sum_{\substack{i,\,j\+=1 \\ i\+<\+j}}^{N\--2}\integrate[\R^3]{\frac{\FT{\xi}(\vec{p}_i+\vec{p}_j,\vec{k}-\vec{q},\vec{q},\vec{p}_1,\ldots\check{\vec{p}}_{ij}\ldots,\vec{p}_{N-2})}{\abs{\vec{k}\--\-\vec{q}}^2+q^2+p_1^2\-+\ldots+\-p_{N-2}^2+\lambda};\-d\vec{q}},\label{off0Gamma}\\[-2.5pt]
        (\Gamma_{\!\mathrm{reg}}\+\xi)(\vec{z},\vec{y}_1,\ldots,\vec{y}_{N-2})=\-\Bigg[\alpha_0\-+\gamma\sum_{i\+=\+1}^{N\--2}\,\frac{\theta(\vec{z}\--\-\vec{y}_i)}{\abs{\vec{z}\--\-\vec{y}_i}}+\tfrac{\gamma}{2}\!\sum_{\substack{i,\,j\+=1 \\ i\+<\+j}}^{N\--2}\frac{\theta(\vec{y}_i\--\-\vec{y}_j)}{\abs{\vec{y}_i\--\-\vec{y}_j}}\Bigg]\xi(\vec{z},\vec{y}_1,\ldots,\vec{y}_{N-2}).\label{regGamma}
    \end{gather}
    \end{subequations}

    \vspace{-0.25cm}
    
    \n Then, one has $\dom{\Gamma^\lambda}=H^1(\R^{3(N-1)})\cap\hilbert*_{N-1}$.
    \begin{proof}
        First of all, the domain of $\Gamma^\lambda_{\mathrm{diag}}$ is clearly $H^1(\R^{3(N-1)})\cap\hilbert*_{N-1}$, hence the result is proven as soon as we show that $\dom{\Gamma^\lambda}\-\supseteq\- H^1(\R^{3(N-1)})\cap\hilbert*_{N-1}$ or, equivalently,
        $$\Gamma^\lambda\in\bounded{H^1(\R^{3(N-1)})\cap\hilbert*_{N-1},\hilbert*_{N-1}}\!.$$
        To this end, assume $\xi\-\in\- \hilbert*_{N-1}\cap H^1(\R^{3(N-1)})$.\newline
        Focusing on $\Gamma_{\!\mathrm{reg}}\+$, it is sufficient to recall $\theta\in\Lp{\infty}[\R^3]$ and the Hardy inequality
        \begin{equation}
            \integrate[\R^3]{\frac{\abs{f(\vec{x})}^2\!\-}{x^2};\-d\vec{x}}\leq 4 \norm{\nabla f}^2\!,\qquad f\in H^1(\R^3).
        \end{equation}
        Regarding $\Gamma^\lambda_{\!\mathrm{off};\+1}$, first we remind the triangular inequality for higher powers which is obtained via Jensen's inequality
        \begin{equation}\label{highPowerTriangular}
            \bigg\lvert\sum_{i=1}^n a_i\bigg\rvert^s\!\!=n^s\bigg\lvert\frac{1}{n}\sum_{i=1}^n a_i\bigg\rvert^s\!\!\leq n^s\,\frac{1}{n}\sum_{i=1}^n \,\abs{a_i}^s=n^{s-1}\sum_{i=1}^n \,\abs{a_i}^s,\qquad \forall \{a_i\}^n_{i=1}\subset\C,\, s\geq 1.
        \end{equation}
        We exploit the previous inequality for $s=2$ and, due to the bosonic symmetry,
        $$\lVert\Gamma^\lambda_{\!\mathrm{off};\+1}\+\xi\rVert_{\hilbert*_{N-1}}^2\!\leq\frac{4\+(N\!-\-2)^2\!\-}{\pi^4}\-\integrate[\R^{3(N-1)}]{;\mspace{-33mu}d\vec{p}\+d\vec{p}'\-d\mspace{-0.75mu}\vec{P}}\-\left\lvert\integrate[\R^3]{\frac{\FT{\xi}(\vec{p}-\vec{q}+\vec{p}'\-,\vec{q},\vec{P})}{\abs{\vec{p}\--\-\vec{q}}^2+q^2+p{\+'\+}^2\-+\-P^2+\lambda};\-d\vec{q}}\right\rvert^2\!.$$
        Furthermore, by replacing $(\vec{p},\vec{p}')\longmapsto\big(\tfrac{1}{2}\+\vec{k}+\vec{k}'\-,\tfrac{1}{2}\+\vec{k} \--\-\vec{k}'\big)$
        \begin{align*}
            \lVert\Gamma^\lambda_{\!\mathrm{off};\+1}\+\xi\rVert_{\hilbert*_{N-1}}^2\!&\leq\frac{4\+(N\!-\-2)^2\!\-}{\pi^4}\-\integrate[\R^{3(N-1)}]{;\mspace{-33mu}d\vec{k}\+d\vec{k}'\-d\mspace{-0.75mu}\vec{P}}\-\left\lvert\integrate[\R^3]{\frac{\FT{\xi}(\vec{k}-\vec{q},\vec{q},\vec{P})}{\abs{\frac{1}{2}\+\vec{k}+\vec{k}'\--\vec{q}}^2\-+q^2\-+\abs{\frac{1}{2}\+\vec{k}-\vec{k}'}^2\-+\-P^2\-+\-\lambda};\-d\vec{q}}\right\rvert^2\\
            &=\frac{4\+(N\!-\-2)^2\!\-}{\pi^4}\-\integrate[\R^{3(N-1)}]{;\mspace{-33mu}d\vec{k}\+d\vec{k}'\-d\mspace{-0.75mu}\vec{P}}\-\left\lvert\integrate[\R^3]{\frac{\FT{\xi}(\vec{k}-\vec{q},\vec{q},\vec{P})}{\frac{1}{2}\+\abs{\vec{k}-\vec{q}}^2\-+q^2\-+2\+\abs{\vec{k}'\--\frac{1}{2}\+\vec{q}}^2\-+\-P^2\-+\-\lambda};\-d\vec{q}}\right\rvert^2\\
            &\leq\frac{4\+(N\!-\-2)^2\!\-}{\pi^4}\-\integrate[\R^{3(N-1)}]{;\mspace{-33mu}d\vec{k}\+d\vec{k}'\-d\mspace{-0.75mu}\vec{P}}\-\left[\integrate[\R^3]{\frac{\abs{\FT{\xi}(\vec{k}-\vec{q},\vec{q},\vec{P})}}{q^2\-+2\+\abs{\vec{k}'\--\frac{1}{2}\+\vec{q}}^2};\-d\vec{q}}\right]^2\\
            &\leq \frac{4\+(N\!-\-2)^2\!\-}{\pi^4}\-\integrate[\R^{3(N-1)}]{;\mspace{-33mu}d\vec{k}\+d\vec{k}'\-d\mspace{-0.75mu}\vec{P}}\-\left[\integrate[\R^3]{\frac{\abs{\FT{\xi}(\vec{k}-\vec{q},\vec{q},\vec{P})}}{k{\+'\+}^2+\frac{1}{2}\+q^2};\-d\vec{q}}\right]^2\!\!,
        \end{align*}
        where in the last step we exploited the elementary lower bound
        $$ 2 \+\abs{\vec{v}-\tfrac{1}{2}\+\vec{w}}^2+w^2\geq v^2+\tfrac{1}{2} w^2\qquad\iff\qquad \abs{\vec{v}-\vec{w}}^2\geq 0.$$
        Next, writing the square of the previous integrand, one gets
        \begin{align*}
            \lVert\Gamma^\lambda_{\!\mathrm{off};\+1}\+\xi\rVert_{\hilbert*_{N-1}}^2\!&\leq\frac{4\+(N\!-\-2)^2\!\-}{\pi^4}\-\integrate[\R^{3(N-1)}]{;\mspace{-33mu}d\vec{k}\+d\vec{k}'\-d\mspace{-0.75mu}\vec{P}}\!\!\integrate[\R^6]{\frac{\abs{\FT{\xi}(\vec{k}-\vec{q},\vec{q},\vec{P})}\,\abs{\FT{\xi}(\vec{k}-\vec{q}'\-,\vec{q}'\-,\vec{P})}}{\left(k{\+'\+}^2+\frac{1}{2}\+q^2\right)\!\left(k{\+'\+}^2+\frac{1}{2}\+q{\+'\+}^2\right)};\-d\vec{q}d\vec{q}'\-}\\
            &\leq\frac{4\+(N\!-\-2)^2\!\-}{\pi^4}\-\integrate[\R^6]{;\-d\vec{q}d\vec{q}'\-}\!\!\integrate[\R^3]{;\-d\vec{k}'\-}\frac{\FT{f}(\vec{q})\FT{f}(\vec{q}')}{(k{\+'\+}^2\-+\frac{1}{2}\+q^2)(k{\+'\+}^2\-+\frac{1}{2}\+q{\+'\+}^2)},
        \end{align*}
        owing to the Cauchy–Schwarz inequality, with
        $$\FT{f}(\vec{p})\vcentcolon=\sqrt{\integrate[\R^{3(N-2)}]{\abs{\FT{\xi}(\vec{k},\vec{p},\vec{P})}^2;\mspace{-34.5mu}d\vec{k}\+d\mspace{-0.75mu}\vec{P}}\,}\+\in\Lp{2}[\R^3,\,(1+p^2)\+d\vec{p}].$$
        The integration along the variable $\vec{k}'$ can be performed explicitly, since
        \begin{equation}\label{explicitIntegral}
            \integrate[0;\pInfty]{\frac{x^2}{(x^2+a^2)(x^2+b^2)};\nquad dx}=\frac{\pi}{2\,(\abs{a}+\abs{b})},\qquad \forall a,b\in\R \,:\; a\neq 0 \lor b \neq 0.
        \end{equation}
        Therefore, one obtains
        $$\lVert\Gamma^\lambda_{\!\mathrm{off};\+1}\+\xi\rVert_{\hilbert*_{N-1}}^2\!\leq\frac{8\sqrt{2}\+(N\!-\-2)^2\!\-}{\pi^2}\-\integrate[\R^6]{\frac{\FT{f}(\vec{q})\FT{f}(\vec{q}')}{q+q{\+'}};\-d\vec{q}d\vec{q}'\-}.$$
        Since the kernel of the previous integral operator is symmetric, the Schur test reads
        $$    \lVert\Gamma^\lambda_{\!\mathrm{off};\+1}\+\xi\rVert_{\hilbert*_{N-1}}^2\!\leq\frac{8\sqrt{2}\+(N\!-\-2)^2\!\-}{\pi^2}\-\integrate[\R^3]{q^2\abs{\FT{f}(\vec{q})}^2;\-d\vec{q}}\;\mathop{\mathrm{ess} \;\mathrm{sup}}_{\vec{p}\,\in\,\R^3}\,\frac{h(\vec{p})}{p}\integrate[\R^3]{\frac{1}{(p{\+'}\-+p)\+p{\+'}}\+\frac{1}{h(\vec{p}')};\-d\vec{p}'}
        $$
        for any (almost everywhere) positive measurable function $\maps{h}{\R^3;\Rplus}$ for which the right-hand side of the previous expression is finite.
        In particular, choosing $h\-:\,\vec{s}\:\longmapsto\,\abs{\vec{s}}^{\frac{3}{2}}$
        \begin{align*}    \lVert\Gamma^\lambda_{\!\mathrm{off};\+1}\+\xi\rVert_{\hilbert*_{N-1}}^2\!&\leq\frac{32\sqrt{2}\+(N\!-\-2)^2\!\-}{\pi}\-\integrate[\R^3]{q^2\abs{\FT{f}(\vec{q})}^2;\-d\vec{q}}\;\sup_{p \,>\, 0}\sqrt{p}\integrate[0;\pInfty]{\frac{1}{(p{\+'}\-+p)\+\sqrt{p{\+'}}};\mspace{-18mu}dp{\+'}}\\
        &=\frac{32\sqrt{2}\+(N\!-\-2)^2\!\-}{\pi}\-\integrate[\R^3]{q^2\abs{\FT{f}(\vec{q})}^2;\-d\vec{q}}\integrate[0;\pInfty]{\frac{1}{(k+\-1)\+\sqrt{k}};\nquad dk}
        =32\sqrt{2}\+(N\!-\-2)^2\!\-\integrate[\R^3]{q^2\abs{\FT{f}(\vec{q})}^2;\-d\vec{q}}\\
        &\leq 32\sqrt{2}\,(N\!-\-2)^2\norm{\xi}[H^1(\R^{3(N-1)})]^2\-.
        \end{align*}
        Analogously, for $\Gamma^\lambda_{\!\mathrm{off};\+0}$ we consider the substitution $(\vec{p}'\-,\vec{p}'')\longmapsto\big(\tfrac{1}{2}\vec{k}+\vec{k}'\-,\tfrac{1}{2}\+\vec{k}\--\-\vec{k}'\big)$
        \begin{align*}
            \lVert\Gamma^\lambda_{\!\mathrm{off};\+0}\+\xi\rVert_{\hilbert*_{N-1}}^2\!&\leq\tfrac{(N\--2)^2(N\--3)^2\!\-}{4\+\pi^4}\!\integrate[\R^{3(N-1)}]{;\mspace{-33mu}d\vec{p}\+d\vec{p}'\-d\vec{p}''\!d\mspace{-0.75mu}\vec{P}}\-\left\lvert\integrate[\R^3]{\frac{\FT{\xi}(\vec{p}'\-+\vec{p}''\!,\vec{p}-\vec{q},\vec{q},\vec{P})}{\abs{\vec{p}\--\-\vec{q}}^2+q^2+p{\+'\+}^2+p{\+''\+}^2\-+\-P^2+\lambda};\-d\vec{q}}\right\rvert^2\\
            &=\tfrac{(N\--2)^2(N\--3)^2\!\-}{4\+\pi^4}\!\integrate[\R^{3(N-1)}]{;\mspace{-33mu}d\vec{p}\+d\vec{k}\+d\vec{k}'\-d\mspace{-0.75mu}\vec{P}}\-\left\lvert\integrate[\R^3]{\frac{\FT{\xi}(\vec{k},\vec{p}-\vec{q},\vec{q},\vec{P})}{\abs{\vec{p}\--\-\vec{q}}^2\-+q^2\-+\frac{1}{2}\+k^2\-+2\+k{\+'\+}^2\-+\-P^2\-+\-\lambda};\-d\vec{q}}\right\rvert^2\\
            &\leq\tfrac{(N\--2)^2(N\--3)^2\!\-}{4\+\pi^4}\!\integrate[\R^{3(N-1)}]{;\mspace{-33mu}d\vec{p}\+d\vec{k}\+d\vec{k}'\-d\mspace{-0.75mu}\vec{P}}\-\left[\integrate[\R^3]{\frac{\abs{\FT{\xi}(\vec{k},\vec{p}-\vec{q},\vec{q},\vec{P})}}{q^2\-+2\+k{\+'\+}^2};\-d\vec{q}}\right]^2.
        \end{align*}
        Setting $g\in H^1(\R^3)$ as
        $$\FT{g}(\vec{q})\vcentcolon=\sqrt{\integrate[\R^{3(N-2)}]{\abs{\FT{\xi}(\vec{k},\vec{p},\vec{q},\vec{P})}^2;\mspace{-34.5mu}d\vec{k}\+d\vec{p}\+d\mspace{-0.75mu}\vec{P}}\,}\+,$$
        the following is due to the Cauchy-Schwarz inequality and~\eqref{explicitIntegral}
        \begin{align*}
            \lVert\Gamma^\lambda_{\!\mathrm{off};\+0}\+\xi\rVert_{\hilbert*_{N-1}}^2\!&\leq\tfrac{(N\--2)^2(N\--3)^2\!\-}{4\+\pi^4}\!\integrate[\R^6]{\FT{g}(\vec{q})\,\FT{g}(\vec{q}');\-d\vec{q}d\vec{q}'\-}\!\integrate[\R^3]{\frac{1}{\left(q^2\-+2k{\+'\+}^2\right)\!\left(q{\+'\+}^2\-+2k{\+'\+}^2\right)};\-d\vec{k}'\-}\\
            &=\tfrac{(N\--2)^2(N\--3)^2\!\-}{2\sqrt{2}\+\pi^2}\!\integrate[\R^6]{\frac{\FT{g}(\vec{q})\,\FT{g}(\vec{q}')}{q+q{\+'}};\-d\vec{q}d\vec{q}'\-}.
        \end{align*}
        Exploiting the same Schur test previously carried out, we conclude 
        $$\lVert\Gamma^\lambda_{\!\mathrm{off};\+0}\+\xi\rVert_{\hilbert*_{N-1}}^2\!\leq \sqrt{2}\,(N\--2)^2(N\--3)^2\norm{\xi}[H^1(\R^{3(N-1)})]^2\!.$$
        
    \end{proof}
\end{prop}

\begin{proof}[Proof of Theorem~\ref{hamiltonianCharacterizationTheo}]
    Given the result of Proposition~\ref{closednessTheo}, the proof can be obtained with standard techniques.
    However, we provide the details for the reader's convenience.\newline
    First of all, by definition~\eqref{QF}, $\mathcal{Q}\geq\--\lambda_0$ since $\Phi^\lambda$ is coercive as soon as $\lambda\->\-\lambda_0$; therefore, the associated Hamiltonian $\mathcal{H}$ has the same lower bound.\newline
    In order to characterize the domain and action of the Hamiltonian, let $\psi=\varphi_\lambda\-+\mathcal{G}^\lambda\xi\-\in\-\dom{\mathcal{H}}$, with $\lambda\->\-0\+$.
    Then, there exists $f\-\in\-\hilbert_N$ such that the sesquilinear form $\mathcal{Q}[\+\cdot,\cdot\+]$ associated with $\mathcal{Q}[\+\cdot\+]$ via the polarization identity satisfies \begin{equation}\label{sesquilinearQRelation}
        \mathcal{Q}[v,\psi]=\scalar{v}{f},\quad \forall v=w^\lambda_v+\mathcal{G}^\lambda\xi_v\in\dom{Q},
    \end{equation}
    where $\define*{f;\mathcal{H}\psi}$.
    By definition, one has
    \begin{equation}\label{sesquilinearQ}
        Q[v,\psi]=\scalar{\mathcal{H}_0^{\frac{1}{2}}w^\lambda_v}{\mathcal{H}_0^{\frac{1}{2}} \varphi_\lambda}+\lambda\+\scalar{w^\lambda_v}{\varphi_\lambda}-\lambda\scalar{v}{\psi}+\Phi^\lambda[\xi_v,\xi].
    \end{equation}
    Let us consider $v\in H^1(\R^{3N})\cap\hilbert_N$, so that $\xi_v\equiv 0$ by injectivity of $\mathcal{G}^\lambda$.
    Then,
    \begin{equation*}
        \scalar{\mathcal{H}_0^{\frac{1}{2}} v}{\mathcal{H}_0^{\frac{1}{2}} \varphi_\lambda}+\lambda\+\scalar{v}{\varphi_\lambda}-\lambda\scalar{v}{\psi}=\scalar{v}{f},\quad\forall v\in H^1(\R^{3N})\cap\hilbert_N\+.
    \end{equation*}
    Hence, $\varphi_\lambda\-\in\- H^2(\R^{3N})\cap\hilbert_N$ and
    \begin{equation}\label{HpsiHfreew}
        (\mathcal{H}_0+\lambda)\varphi_\lambda-\lambda\+ \psi=f,
    \end{equation}
    which is equivalent to
    \begin{equation}
        \mathcal{H}\psi=\mathcal{H}_0\+ \varphi_\lambda-\lambda\+ \mathcal{G}^\lambda\xi.
    \end{equation}
    Next, let $v\-\in\-\dom{Q}$.
    Taking account of~\eqref{HpsiHfreew}, we have
    \begin{equation*}
        \scalar{v}{f+\lambda\+ \psi}=\scalar{w^\lambda_v}{(\mathcal{H}_0+\lambda)\varphi_\lambda}+\scalar{\mathcal{G}^\lambda\xi_v}{(\mathcal{H}_0+\lambda)\varphi_\lambda}.
    \end{equation*}
    On the other hand, recalling~\eqref{sesquilinearQRelation} and~\eqref{sesquilinearQ},
    \begin{gather*}
        \scalar{v}{f+\lambda\+ \psi}=Q[v,\psi]+\lambda\scalar{v}{\psi}=\scalar{w^\lambda_v}{(\mathcal{H}_0+\lambda)\varphi_\lambda}+\Phi^\lambda[\xi_v,\xi],
    \intertext{hence, exploiting~(\ref{potentialDef},~\ref{potentialDistributionalDef})}
        \Phi^\lambda[\xi_v,\xi]=\scalar{\mathcal{G}^\lambda\xi_v}{(\mathcal{H}_0+\lambda)\varphi_\lambda}=4\pi N(N\!-\-1)\scalar{\xi_v}{\varphi_\lambda|_{\pi_{\{N-1,N\}}}}[\hilbert*_{N-1}],\quad\forall\xi_v\in H^{\frac{1}{2}}(\R^{3(N-1)})\cap\hilbert*_{N-1}.
    \end{gather*}
    Therefore, by means of Proposition~\ref{GammaDomain}, we conclude that $\xi\-\in\-\dom{\Gamma^\lambda}=H^1(\R^{3(N-1)})$ and \begin{equation}
        \Gamma^\lambda\xi=4\pi N(N\!-\-1)\varphi_\lambda\big|_{\pi_{\{N-1,N\}}}.
    \end{equation}
    In the following, we show that our Hamiltonian satisfies boundary conditions~\eqref{mfBC} in the $\Lp{2}$-topology.
    In other words, we want to show the following identity
    \begin{equation}\label{strongTopology}
        \lim_{\vec{r}\to\vec{0}}\integrate[\R^{3(N-1)}]{\!\left\lvert (U_\sigma\psi)(\vec{r},\vec{x},\vec{X}_{\-\sigma})-\tfrac{\xi(\vec{x},\vec{X}_{\-\sigma}\-)}{r}-(\Gamma^{\+\sigma}_{\!\mathrm{reg}}\+\xi)(\vec{x},\vec{X}_{\-\sigma})\right\rvert^2;\mspace{-33mu}d\vec{x}d\vec{X}_{\-\sigma}}\!=0,\qquad\forall \psi \in\dom{\mathcal{H}},\,\sigma\in\mathcal{P}_N,
    \end{equation}
    where, given $\sigma=\{i,j\}\-\in\-\mathcal{P}_N$, $U_\sigma$
    is the following unitary operator
    \begin{equation}
    \label{unitaryCoordinate}\begin{split}
        &\maps{U_\sigma}{\hilbert_N; \left\{\psi\in\Lp{2}[\R^6,d\vec{r}d\vec{x}]\-\otimes\-\LpS{2}[\R^{3(N-2)},d\vec{X}_{\-\sigma}]\:\Big|\;\psi(\vec{r},\vec{x},\vec{X}_{\-\sigma})=\psi(-\vec{r},\vec{x},\vec{X}_{\-\sigma})\-\right\}}\\
        &\mspace{175mu}(U_\sigma\+f)(\vec{r},\vec{x},\vec{X}_{\-\sigma})=f(\vec{x}+\tfrac{\vec{r}}{2},\vec{x}-\tfrac{\vec{r}}{2},\vec{X}_{\-\sigma}).
        \end{split}
    \end{equation}
    Exploiting the decomposition of $\psi$ given by~\eqref{psiDec}, we write
    \begin{align*}
        \integrate[\R^{3(N-1)}]{;\mspace{-33mu}d\vec{x}d\vec{X}_{\-\sigma}}&\!\left\lvert (U_\sigma\psi)(\vec{r},\vec{x},\vec{X}_{\-\sigma})-\tfrac{\xi(\vec{x},\vec{X}_{\-\sigma}\-)}{r}-(\Gamma^{\+\sigma}_{\!\mathrm{reg}}\+\xi)(\vec{x},\vec{X}_{\-\sigma})\right\rvert^2\!=\\
        =&\integrate[\R^{3(N-1)}]{\!\left\lvert (U_\sigma\+\varphi_\lambda)(\vec{r},\vec{x},\vec{X}_{\-\sigma})+(U_\sigma\+\mathcal{G}^\lambda\xi)(\vec{r},\vec{x},\vec{X}_{\-\sigma})-\tfrac{\xi(\vec{x},\vec{X}_{\-\sigma}\-)}{r}-(\Gamma^{\+\sigma}_{\!\mathrm{reg}}\+\xi)(\vec{x},\vec{X}_{\-\sigma})\right\rvert^2;\mspace{-33mu}d\vec{x}d\vec{X}_{\-\sigma}}\\
        \leq&\:3\!\integrate[\R^{3(N-1)}]{\!\left\lvert (U_\sigma\+\varphi_\lambda)(\vec{r},\vec{x},\vec{X}_{\-\sigma})-(\varphi_\lambda|_{\pi_\sigma})(\vec{x},\vec{X}_\sigma)\right\rvert^2;\mspace{-33mu}d\vec{x}d\vec{X}_{\-\sigma}}\!+\tag{\textasteriskcentered}\label{firstStrongConvergence}\\[-2.5pt]
        &+3\!\integrate[\R^{3(N-1)}]{\!\bigg\lvert \Big(U_\sigma\!\sum_{\nu\+\neq\+\sigma}\mathcal{G}^\lambda_\nu\+\xi\Big)(\vec{r},\vec{x},\vec{X}_{\-\sigma})+\sum_{\nu\+\neq\+\sigma}(\Gamma^{\+\nu\+\smallsetminus\sigma,\+\lambda}_{\!\mathrm{off};\+\abs{\sigma\+\cap\+\nu}}\+\xi)(\vec{x},\vec{X}_{\-\sigma})\bigg\rvert^2;\mspace{-33mu}d\vec{x}d\vec{X}_{\-\sigma}}\!+\tag{\textasteriskcentered\textasteriskcentered}\label{secondStrongConvergence}\\
        &+3\!\integrate[\R^{3(N-1)}]{\!\left\lvert(U_\sigma\+\mathcal{G}^\lambda_\sigma\+\xi)(\vec{r},\vec{x},\vec{X}_{\-\sigma})-\tfrac{\xi(\vec{x},\vec{X}_{\-\sigma}\-)}{r}+(\Gamma^{\+\lambda}_{\!\mathrm{diag}}\+\xi)(\vec{x},\vec{X}_{\-\sigma}\-)\right\rvert^2;\mspace{-33mu}d\vec{x}d\vec{X}_{\-\sigma}}\!,\tag{\textasteriskcentered\textasteriskcentered\textasteriskcentered}\label{thirdStrongConvergence}
    \end{align*}
    where we exploited~\eqref{GammaTMSHyp}, \eqref{regularComponentTraced} and the triangular inequality~\eqref{highPowerTriangular}.
    Then, we prove that each of the previous integrals vanishes as $\vec{r}$ approaches zero.\newline
    To this end, it is convenient to mention the action of the operator $U_\sigma$ in the Fourier space representation of $\Lp{2}[\R^{3(N-1)}, d\vec{x}d\vec{X}_{\-\sigma}]$ for a given $f\-\in\-\LpS{2}[\R^{3N}]$, namely
    \begin{equation}\label{unitaryCoordinateTransformed}
        (\widehat{U_\sigma f})(\vec{r},\vec{p},\vec{P}_{\!\sigma})
        =\tfrac{1}{(2\pi)^{3/2}\!\-}\!\integrate[\R^3]{\cos\!\left[\vec{r}\-\cdot\!\left(\tfrac{\vec{p}}{2}\--\vec{q}\right)\-\right]\-\FT{f}(\vec{p}-\vec{q},\vec{q},\vec{P}_{\!\sigma});\-d\vec{q}}.
    \end{equation}
    Indeed, we can now make use of Plancherel's theorem to study each of the three terms.\newline
    First of all, we have
    \begin{align*}
        \eqref{firstStrongConvergence}=&\,\frac{3}{\+(2\pi)^3\!\-}\-\integrate[\R^{3(N-1)}]{\!\left\lvert \integrate[\R^3]{\!\left(\cos\!\left[\vec{r}\-\cdot\!\left(\tfrac{\vec{p}}{2}\--\vec{q}\right)\-\right]\--\-1\right)\FT{\varphi}_\lambda(\vec{p}\--\-\vec{q},\vec{q},\vec{P}_{\!\sigma});\-d\vec{q}}\right\rvert^2;\mspace{-33mu}d\vec{p}\+d\mspace{-0.75mu}\vec{P}_{\!\sigma}}\\
        =&\,\frac{3}{\+(2\pi)^3\!\-}\-\integrate[\R^{3(N+1)}]{\!\left(\cos\!\left[\vec{r}\!\cdot\!\left(\tfrac{\vec{p}}{2}\--\vec{q}\right)\-\right]\!-\-1\right)\!\-\left(\cos\!\left[\vec{r}\!\cdot\!\left(\tfrac{\vec{p}}{2}\--\vec{q}'\right)\-\right]\!-\-1\right)\conjugate*{\FT{\varphi}_{\mspace{-0.5mu}\lambda}(\vec{p}\--\-\vec{q},\vec{q},\vec{P}_{\mspace{-3.5mu}\sigma})}\,\FT{\varphi}_{\mspace{-0.5mu}\lambda}(\vec{p}\--\-\vec{q}'\!,\vec{q}'\-,\vec{P}_{\mspace{-3.5mu}\sigma});\mspace{-33mu}d\vec{p}\+d\mspace{-0.75mu}\vec{P}_{\mspace{-3.5mu}\sigma}d\vec{q}d\vec{q}'\-}.
    \end{align*}
    We claim that the function
    $$4\,\abs{\FT{\varphi}_\lambda(\vec{p}\--\-\vec{q},\vec{q},\vec{P}_{\!\sigma})}\,\abs{\FT{\varphi}_\lambda(\vec{p}\--\-\vec{q}'\-,\vec{q}'\-,\vec{P}_{\!\sigma})}$$
    is a uniform integrable majorant, so that the limit $\vec{r}\longrightarrow\vec{0}$ can be computed inside the integral, obtaining the result.
    Indeed, due to the Cauchy-Schwarz inequality
    \begin{align*}
        4\!\integrate[\R^{3(N+1)}]{\abs{\FT{\varphi}_\lambda(\vec{p}\--\-\vec{q},\vec{q},\vec{P}_{\!\sigma})}\,\abs{\FT{\varphi}_\lambda(\vec{p}\--\-\vec{q}'\-,\vec{q}'\-,\vec{P}_{\!\sigma})};\mspace{-33mu}d\vec{q}d\vec{q}'\-d\vec{p}\+d\mspace{-0.75mu}\vec{P}_{\!\sigma}}\leq 4\!\left[\integrate[\R^3]{\-\sqrt{\-\integrate[\R^{3(N-1)}]{\abs{\FT{\varphi}_\lambda(\vec{p},\vec{q},\vec{P}_{\!\sigma})}^2\+;\mspace{-33mu}d\vec{p}\+d\mspace{-0.75mu}\vec{P}_{\!\sigma}}}\,;\-d\vec{q}}\right]^2\!\\[-2.5pt]
        \leq 4\!\integrate[\R^3]{\frac{1}{1+q{\+'\+}^4};\-d\vec{q}'\-}\integrate[\R^{3N}]{(1\-+q^4)\,\abs{\FT{\varphi}_\lambda(\vec{p},\vec{q},\vec{P}_{\!\sigma})}^2;\mspace{-10mu}d\vec{q}d\vec{p}\+d\mspace{-0.75mu}\vec{P}_{\!\sigma}}\leq 4\sqrt{2\+}\+\pi^2\norm{\varphi_\lambda}[H^2(\R^{3N})]^2\!.
    \end{align*}
    Concerning the second term, taking into account~\eqref{potentialFourier}, (\ref{off0GammaHyp},~\ref{off1GammaHyp}) and~\eqref{unitaryCoordinateTransformed} one has

    \vspace{-0.35cm}
    
    \begin{align*}
        \eqref{secondStrongConvergence}=3&\!\integrate[\R^{3(N-1)}]{\!\bigg\lvert\Big(\widehat{U_\sigma\!\!\sum\limits_{\nu\+\neq\+\sigma}\!\mathcal{G}_\nu^\lambda\+\xi}\Big)(\vec{r},\vec{p},\vec{P}_{\!\sigma})+\sum_{\nu\+\neq\+\sigma}(\widehat{\Gamma^{\+\nu,\+\lambda}_{\!\mathrm{off};\+\abs{\sigma\+\cap\+\nu}}\+\xi})(\vec{p},\vec{P}_{\!\sigma})\bigg\rvert^2\!;\mspace{-33mu}d\vec{p}\+d\mspace{-0.75mu}\vec{P}_{\!\sigma}}\\[-2.5pt]
        =3&\!\integrate[\R^{3(N-1)}]{\!\bigg\lvert\-\sum_{\substack{\nu\+=\{k,\+\ell\}\+:\\\sigma\,\cap\,\nu\+=\+\emptyset}}\mspace{-6mu}\frac{1}{\,\pi^2\!\-}\-\integrate[\R^3]{\!\left(\cos\!\left[\vec{r}\!\cdot\!\left(\tfrac{\vec{p}}{2}\--\vec{q}\right)\-\right]\!-\-1\right)\frac{\FT{\xi}(\vec{p}_k\-+\vec{p}_\ell\+,\vec{p}\--\-\vec{q},\vec{q},\vec{P}_{\!\sigma\,\cup\,\nu})}{\abs{\vec{p}\--\-\vec{q}}^2+q^2+P^2_{\!\sigma}+\lambda};\-d\vec{q}}\,+;\mspace{-33mu}d\vec{p}\+d\mspace{-0.75mu}\vec{P}_{\!\sigma}}\\[-7.5pt]
        &\pushright{+\mspace{-9mu}\sum_{\substack{\nu\+:\\ |\sigma\,\cap\,\nu|\+=\+1}}\mspace{-6mu}\frac{1}{\,\pi^2\!\-}\-\integrate[\R^3]{\!\left(\cos\!\left[\vec{r}\!\cdot\!\left(\tfrac{\vec{p}}{2}\--\vec{q}\right)\-\right]\!-\-1\right)\frac{\FT{\xi}(\vec{p}\--\-\vec{q}+\vec{p}_{\nu\+\smallsetminus\sigma}\+,\vec{q},\vec{P}_{\!\sigma\,\cup\:\nu})}{\abs{\vec{p}\--\-\vec{q}}^2+q^2+P^2_{\!\sigma}+\lambda};\-d\vec{q}}\bigg\rvert^2\!.}
    \end{align*}

    \vspace{-0.2cm}
    
    \n The dominated convergence theorem can be easily applied also in this case, as long as $\xi$ is in the domain of $\Gamma^\lambda$.\newline
    Lastly,

    \vspace{-0.8cm}
    
    \begin{align*}
        \eqref{thirdStrongConvergence}&=3\!\integrate[\R^{3(N-1)}]{\!\left\lvert(\widehat{U_\sigma\+\mathcal{G}^\lambda_\sigma\+\xi})(\vec{r},\vec{p},\vec{P}_{\!\sigma})-\tfrac{\FT{\xi}(\vec{p},\vec{P}_{\!\sigma}\-)}{r}+(\widehat{\Gamma^\lambda_{\!\mathrm{diag}}\+\xi})(\vec{p},\vec{P}_{\!\sigma}\-)\right\rvert^2;\mspace{-33mu}d\vec{p}\+d\mspace{-0.75mu}\vec{P}_{\-\sigma}}\\
        &=3\!\integrate[\R^{3(N-1)}]{\!\left\lvert\tfrac{1}{\+\pi^2\!\-}\:\FT{\xi}(\vec{p},\vec{P}_{\!\sigma})\!\integrate[\R^3]{\frac{e^{i\+\vec{r}\+\cdot\left(\frac{\vec{p}}{2}-\+\vec{q}\right)}}{\abs{\vec{p}\--\-\vec{q}}^2+q^2+P_\sigma^2+\lambda};\-d\vec{q}}-\tfrac{\FT{\xi}(\vec{p},\vec{P}_{\!\sigma}\-)}{r}+(\widehat{\Gamma^\lambda_{\!\mathrm{diag}}\+\xi})(\vec{p},\vec{P}_{\!\sigma}\-)\right\rvert^2;\mspace{-33mu}d\vec{p}\+d\mspace{-0.75mu}\vec{P}_{\-\sigma}}\\
        &=3\!\integrate[\R^{3(N-1)}]{\!\left\lvert\tfrac{1}{\+\pi^2\!\-}\:\FT{\xi}(\vec{p},\vec{P}_{\!\sigma})\!\integrate[\R^3]{\frac{e^{i\+\vec{r}\+\cdot\+\vec{q}}}{\frac{1}{2}\+p^2+2q^2+P_\sigma^2+\lambda};\-d\vec{q}}-\tfrac{\FT{\xi}(\vec{p},\vec{P}_{\!\sigma}\-)}{r}+\tfrac{1}{\!\-\sqrt{2\+}\+}\sqrt{\tfrac{1}{2}\+p^2\-+\-P_{\!\sigma}^2\-+\-\lambda\,}\,
        \FT{\xi}(\vec{p},\vec{P}_{\!\sigma})\right\rvert^2;\mspace{-33mu}d\vec{p}\+d\mspace{-0.75mu}\vec{P}_{\-\sigma}}\\
        &=3\!\integrate[\R^{3(N-1)}]{\!\left[\tfrac{e^{-\frac{r}{\sqrt{2}}\sqrt{\frac{1}{2}\+p^2+P_{\!\sigma}^2+\lambda}}-1}{r}+\tfrac{1}{\!\-\sqrt{2\+}\+}\sqrt{\tfrac{1}{2}\+p^2\-+\-P_{\!\sigma}^2\-+\-\lambda\,}\right]^2\!\abs{\FT{\xi}(\vec{p},\vec{P}_{\!\sigma})}^2;\mspace{-33mu}d\vec{p}\+d\mspace{-0.75mu}\vec{P}_{\-\sigma}}.
    \end{align*}
    Since $\xi\in H^1(\R^{3(N-1)})$, once again the dominated convergence theorem can be exploited by pointing out that $\frac{1}{2}(\frac{1}{2}\+p^2+P_{\!\sigma}^2+\lambda)\+\abs{\FT{\xi}(\vec{p},\vec{P}_{\!\sigma})}^2$ is a uniform majorant of the integrand.\newline
    Finally, for the representation of the resolvent~\eqref{resolventH}, given $f\-\in\-\hilbert_N$ we have to find $\psi=\varphi_\lambda+\mathcal{G}^\lambda\xi\-\in\-\dom{\mathcal{H}}$ such that $(\mathcal{H}+\lambda)\psi=f$ or, equivalently, $(\mathcal{H}_0+\lambda)\varphi_\lambda=f$ owing to~\eqref{HpsiHfreew}.
    Then, $\varphi_\lambda=(\mathcal{H}_0+\lambda)^{-1}f$ and $\xi$ must be the solution of equation~\eqref{chargeEq}.
    
\end{proof}

\section{Comparison with the Approach via Dirichlet Forms}\label{dirichletSection}

In this section, we compare our results with those obtained in~\cite{AHK}, following an approach based on the theory of Dirichlet forms.\newline
First, we analyze the example provided in~\cite[section 2 - example 4]{AHK}.
Consider the Dirichlet form $E_\phi$ in $\hilbert_N$ given by~\eqref{AlbeverioDF} and the function $\phi\in\Lp*{2}[\R^{3N}]$ defined by~\eqref{weightFunction}.
It is noteworthy that
\begin{equation}
    \Delta\+\phi\mspace{0.75mu}(\vec{\mathrm{x}})=(-\mathcal{H}_0\+\phi)(\vec{\mathrm{x}})=2m^2\phi\mspace{0.75mu}(\vec{\mathrm{x}}),\qquad \forall \vec{\mathrm{x}}\in\R^{3N}\-\setminus \pi.
\end{equation}
In~\cite[section 2]{AHK}
it is shown that the quantity \begin{equation}\label{dirichletQF}
\mathcal{Q}_D[\psi]\vcentcolon=E_\phi\!\left[\psi/\phi\right]-2m^2\norm{\psi}^2\-,\qquad\psi\in\hilbert_N : \nabla\tfrac{\psi}{\phi}\in \Lp{2}[\R^{3N}\-,\+\phi^2(\vec{\mathrm{x}})\+d\vec{\mathrm{x}}]
\end{equation}
uniquely characterizes a singular perturbation of $\mathcal{H}_0\+,\hilbert_N\cap H^2(\R^{3N})$ supported on $\pi$.
More precisely, for any non-negative value of $m$, the quadratic form $\mathcal{Q}_D$ is associated to a bounded from below operator, denoted by $-\Delta_m$, such that
\begin{gather*}
    -\Delta_m \psi=\mathcal{H}_0\+\psi,\qquad \forall \psi\-\in\!\hilbert_N\cap H_0^2(\R^{3N}\!\smallsetminus \pi),\\
    -\Delta_m\-\geq\- -2\+m^2.
\end{gather*}
Thus, in this example, a class of zero-range Hamiltonians with preassigned lower bound $-2m^2$ is defined, circumventing instability issues.
However, the authors in~\cite{AHK} do not characterize the domain of the Hamiltonian, and then it is not clear which boundary condition is satisfied along the coincidence hyperplanes. 
To address this gap, we reformulate $\mathcal{Q}_D$ within our framework.
In this way, the domain of the Hamiltonian is made explicit, and a comparison with our results is therefore possible.\newline
According to~\cite{P1, P2}, every singular perturbation of a densely defined symmetric operator (which is closed with respect to its graph norm) can be exhaustively characterized.
In particular, $-\Delta_m$ is a singular perturbation of the free Hamiltonian supported on the coincidence hyperplanes, hence we know there exists a (locally Lipschitz) continuous map $\maps{\Gamma_{\!D\+}}{\rho(\mathcal{H}_0);\linear{\hilbert*_{N-1}}}$ fulfilling some suitable conditions (see \eg, \cite[eq.~(2.19)]{FeT2}), such that $-\Delta_m$ coincides with the s.a. extension $-\Delta^T_{\Gamma_{\!D\+}}$ characterized by
\begin{equation}\label{bcHamiltonianCharacterization}
    \begin{dcases}
    \dom{-\Delta^T_{\Gamma_{\!D}}}=\left\{\psi\in\hilbert_N\,\big|\;\psi=\varphi_z\-+\mathcal{G}^{-z}\+\xi,\:\Gamma_{\!D}(z)\+\xi=T\varphi_z\+,\;\varphi_z\in\dom{\mathcal{H}_0},\:\xi\in \dom{\Gamma_{\!D}(z)\-},\: z\in\rho(\mathcal{H}_0) \right\}\-,\\
    -\Delta^T_{\Gamma_{\!D}}\+\psi=\mathcal{H}_0\+\varphi_z+z\,\mathcal{G}^{-z}\+\xi.
    \end{dcases}
\end{equation}
Here $\mathcal{G}^{-z}$ stands for the analytic continuation of the potential defined in~\eqref{potentialDef} and, setting $\vec{x}=\tfrac{\vec{x}_N\++\,\vec{x}_{N\--1}}{2}$ and $\vec{r}\-=\vec{x}_N\!-\vec{x}_{N\--1}$, we have $T\!\in\-\bounded{\dom{\mathcal{H}_0}, H^{\frac{1}{2}}(\R^{3(N\--1)})}$ given by
$$(Tf)(\vec{x},\vec{x}_1,\ldots,\vec{x}_{N-2})=4\pi\+N(N\!-\-1) f(\vec{r},\vec{x},\vec{x}_1,\ldots,\vec{x}_{N-2})\big|_{\vec{r}=0}\+,\qquad f\in \dom{\mathcal{H}_0}.$$
In particular, any $\psi\-\in\-\dom{-\Delta^T_{\Gamma_{\!D\+}}}$ can be decomposed as in equation~\eqref{psiDec} for any fixed $z\-\in\-\C\smallsetminus\R_+$, in terms of a regular component
$\varphi_z\-\in\- H^2(\R^{3N})$ and an element $\xi$ (independent of $z$) in the domain of the closed operator $\Gamma_{\!D\+}(z)$.
Moreover, analogously to~\eqref{regularComponentTraced}, the following boundary condition holds
$$\Gamma_{\!D\+}(z)\,\xi = T \varphi_z.$$
Therefore, in order to understand the boundary conditions along $\pi$ encoded by $-\Delta_m$, we need to clarify what kind of s.a. extension is described by the map $\Gamma_{\!D\+}$.
To this end, let us mention that the energy form $\mathcal{Q}_D$ associated to $-\Delta^T_{\Gamma_{\!D\+}}$ can be written as follows for $\psi\in\dom{-\Delta^T_{\Gamma_{\!D}}}$
\begin{equation}\label{singularPerturbationQF}\begin{split}
    \mathcal{Q}_D[\psi]&=E_\phi[\psi/\phi]-2m^2\norm{\psi}^2=\scalar{\psi}{-\Delta^T_{\Gamma_{\!D}} \psi}\\
    &=z\norm{\psi}^2\!+\scalar{\varphi_z}{(\mathcal{H}_0-z)\+\varphi_z}+\scalar{\xi}{\Gamma_{\!D\+}(z)\+\xi}[\hilbert*_{N-1}],
\end{split}
\end{equation}
with either $z\-<\-0$ or $\varphi_z\-\perp \mathcal{G}^{-z}\+\xi$ (required in order for $\mathcal{Q}_D$ to be real valued).
From~\eqref{singularPerturbationQF}, one deduces that the form domain of $\mspace{-0.75mu}-\Delta^T_{\Gamma_{\!D}}\mspace{-0.75mu}$ is composed of elements $\psi\!\in\mspace{-2.25mu}\hilbert_N$ that can be written as $\psi\mspace{-2.25mu}=\mspace{-2.25mu}w_\lambda\mspace{-2.25mu}+\mathcal{G}^\lambda \mspace{0.75mu}\xi$, for some $w_\lambda\-\in\- H^1(\R^{3N})$ and $\xi$ in a suitable dense subspace of $\hilbert*_{N-1}$ (which is independent of $\lambda\->\-0$).
Therefore, let us consider a vector in this form domain reading as $$\psi=\mathcal{G}^\lambda\xi,\qquad \text{for some }\,\lambda>0.$$
With this choice, we can isolate in~\eqref{singularPerturbationQF} the quadratic form of the charges $\Phi^\lambda_D$ associated to the operator $\Gamma_{\!D}(-\lambda)$, \ie
\begin{equation}\label{dirichletChargeQF}
    \define{\Phi_{\-D\+}^\lambda[\xi];E_\phi[\+\mathcal{G}^\lambda\xi/\phi\+]+(\lambda-2m^2)\lVert\mathcal{G}^\lambda\xi\rVert^2}.
\end{equation}
At this point, it is sufficient to study the relation between $\Phi^\lambda_D$ and our quadratic form of the charges $\Phi^\lambda$ given by~\eqref{phiDefinition}, in order to compare the associated zero-range Hamiltonians $-\Delta_m$ and $\mathcal{H}$.

\vs

While the above example illustrates the approach for a specific weight function $\phi$, the framework of~\cite{AHK} is more general.
As far as the subject of this paper is concerned, we refer to~\cite[theorem 2.4]{AHK}, which allows to construct s.a. extensions of any operator $-\Delta+\frac{\Delta\phi}{\phi},\, C_c^2(\R^{3N}\-\smallsetminus\- \pi)$, where $\frac{\Delta\phi}{\phi}$ is meant in distributional sense, provided the real function $\phi\-\in\- \Lp*{2}[\R^{3N}]\cap\hilbert_N$ such that $\nabla\phi, \frac{\nabla\phi}{\phi}$ and $\frac{\Delta\phi}{\phi}$ are all in $\Lp*{2}[\R^{3N}\-\smallsetminus\- \pi]$.
In particular, these conditions ensure the closability of the Dirichlet form $E_\phi$ introduced in~\eqref{AlbeverioDF}.\newline
Indeed, the isometry $f\longmapsto \phi^{-1} f$ between $\Lp{2}[\R^{3N}]$ and $\Lp{2}[\R^{3N}\!,\, \phi^2(\vec{\mathrm{x}})d\vec{\mathrm{x}}]$ takes the s.a. operator associated with $E_\phi$, which we denote by $H_\phi,$ into $-\Delta + \frac{\Delta\phi}{\phi}$ for functions in $C^2_c(\R^{3N}\-\smallsetminus\- \pi)$.
Namely from one hand one has
$$\integrate[\R^{3N}]{\phi^2(\vec{\mathrm{x}})\frac{\conjugate*{\varphi(\vec{\mathrm{x}})}}{\phi(\vec{\mathrm{x}})}\big(H_\phi\,\phi^{-1}\psi\big)\-(\vec{\mathrm{x}});\mspace{-10mu}d\vec{\mathrm{x}}}\-=\!\-\integrate[\R^{3N}]{\conjugate*{\varphi(\vec{\mathrm{x}})}\big(\phi\+H_\phi\,\phi^{-1}\psi\big)\-(\vec{\mathrm{x}});\mspace{-10mu}d\vec{\mathrm{x}}},\quad\forall \psi\-\in\- C_c^2(\R^{3N}\!\smallsetminus\- \pi),\+ \varphi\-\in\- C_c^1(\R^{3N}),$$
while, on the other hand
\begin{equation*}
    \integrate[\R^{3N}]{\phi^2(\vec{\mathrm{x}})\,\nabla\tfrac{\conjugate{\varphi}}{\phi} \-\cdot\- \nabla\tfrac{\psi}{\phi};\mspace{-10mu} d\vec{\mathrm{x}}}=\!\-\integrate[\R^{3N}]{\conjugate*{\varphi(\vec{\mathrm{x}})}(-\Delta \psi)(\vec{\mathrm{x}})+\conjugate*{\varphi(\vec{\mathrm{x}})}\,\tfrac{(\Delta\phi)(\vec{\mathrm{x}})\-}{\phi(\vec{\mathrm{x}})} \,\psi(\vec{\mathrm{x}});\mspace{-10mu} d\vec{\mathrm{x}}},\quad\forall \psi\in C_c^2(\R^{3N}\mspace{-3.75mu}\smallsetminus\- \pi),\+\varphi\in C_c^1(\R^{3N}).
\end{equation*}
Therefore,
\begin{equation}\phi\+H_\phi\,\phi^{-1}=-\Delta+\tfrac{\Delta\phi}{\phi},\qquad \text{in }\,C_c^2(\R^{3N}\-\smallsetminus\-\pi).
\end{equation}
Let us assume $\frac{\Delta\phi}{\phi}\!\in\!\Lp{\infty}[\R^{3N}]$ so that the associated multiplication operator is an infinitesimal perturbation for both $-\Delta$ and $H_\phi$.
In particular, we obtain that $\phi\big(H_\phi\--\tfrac{\Delta\phi}{\phi}\big)\phi^{-1}$ is s.a. and it is equal to $-\Delta$ in $C_c^{2}(\R^{3N}\-\smallsetminus \pi)$.
Moreover, the boundedness of $\tfrac{\Delta\phi}{\phi}$ also implies
$$\phi\big(H_\phi-\tfrac{\Delta\phi}{\phi}\big)\phi^{-1}\geq -\mathop{\mathrm{ess\,sup}} \tfrac{\Delta\phi}{\phi}.$$
In other words, the quadratic form
$E_\phi[\psi/\phi]\--\-\scalar{\psi}{\!\tfrac{\Delta\phi}{\phi}\mspace{-0.75mu}\psi}$ uniquely defines a s.a. and lower-bounded extension $-\Delta_\phi$ of $\--\Delta, C_c^2(\R^{3N}\!\-\smallsetminus\mspace{-2.25mu} \pi\mspace{-0.75mu})$, \ie~a many-body Hamiltonian with zero-range interaction.
Our aim is to understand the boundary conditions along $\pi$ satisfied by the elements of the domain of such a Hamiltonian by establishing a connection with the class of Hamiltonians constructed in Theorem~\ref{hamiltonianCharacterizationTheo}.\newline
By the same argument exploited in the previous example, we know there exists a quadratic form $\Phi^\lambda_\phi$ such that
\begin{equation*}
E_\phi[\psi/\phi]-\scalar{\psi}{\tfrac{\Delta\phi}{\phi}\psi}=-\lambda\norm{\psi}^2+\norm{\nabla w_\lambda}^2+\lambda\norm{w_\lambda}^2+\Phi_\phi^\lambda[\xi],
\end{equation*}
where we are again using the decomposition $\psi=w_\lambda\!+\mathcal{G}^\lambda\xi$ for some $w_\lambda\-\in\- H^1(\R^{3N})$, $\lambda\->\-0$ and $\xi$ in a proper dense subspace of $\hilbert*_{N-1}$.
In analogy with the previous case, we select $w_\lambda=0$, or equivalently $\psi=\mathcal{G}^\lambda\xi$ for some $\lambda\->\-0$, so that we find
\begin{equation}\label{dirichletGeneralChargeQF}
\Phi^\lambda_\phi[\xi]=\lambda\lVert\mathcal{G}^\lambda\xi\rVert^2+E_\phi[\mathcal{G}^\lambda\xi/\phi]-\scalar{\mathcal{G}^\lambda\xi}{\tfrac{\Delta\phi}{\phi}\mathcal{G}^\lambda\xi}.
\end{equation}
We now have all the ingredients to compare the zero-range Hamiltonians $-\Delta_\phi$, where $\phi$ is given by~\eqref{genericPhiDF}, with the class of operators constructed in Theorem~\ref{hamiltonianCharacterizationTheo}.
\begin{proof}[Proof of Proposition~\ref{lastResult}]
    The result is obtained as soon as we show that, under our assumptions, $\Phi^\lambda\!=\Phi^\lambda_\phi$, with $\Phi^\lambda_\phi$ given by~\eqref{dirichletGeneralChargeQF}.
    To this end, we proceed with the evaluation of $E_\phi[\mathcal{G}^\lambda\xi/\phi]$.\newline
    Let $D^{\+\epsilon}\vcentcolon=\big\{\-(\vec{x}_1,\ldots,\vec{x}_N)\-\in\R^{3N}\,\big|\;\min\limits_{1\leq \+i\,<\,j\+\leq N}\+\abs{\vec{x}_i-\vec{x}_j}\->\epsilon\big\}$ so that
\begin{align*}
    E_\phi[\mathcal{G}^\lambda\xi/\phi]\-&=\!\lim_{\epsilon\to 0^+}\integrate[D^{\+\epsilon}]{\phi^2(\vec{\mathrm{x}})\left\lvert\frac{\nabla \mathcal{G}^\lambda\xi}{\phi}-\mathcal{G}^\lambda\xi\;\frac{\!\nabla\phi}{\,\phi^2\!}\+\right\rvert^2;\-d\vec{\mathrm{x}}}\\
    &=\!\lim_{\epsilon\to 0^+}\integrate[D^{\+\epsilon}]{\!\left[\abs{\nabla \mathcal{G}^\lambda\xi}^2\--2\+\Re\,\mathcal{G}^\lambda\xi\;\nabla\mathcal{G}^\lambda\xi\cdot\frac{\!\nabla\phi}{\phi}+\abs{\mathcal{G}^\lambda\xi}^2\,\frac{\abs{\nabla\phi}^2\!\-}{\,\phi^2}\,\right];\-d\vec{\mathrm{x}}}\\
    &=\!\lim_{\epsilon\to 0^+}\integrate[D^{\+\epsilon}]{\!\left[\abs{\nabla \mathcal{G}^\lambda\xi}^2\--\nabla\abs{\mathcal{G}^\lambda\xi}^2\-\cdot\nabla\ln\phi+\abs{\mathcal{G}^\lambda\xi}^2\,\frac{\abs{\nabla\phi}^2\!\-}{\,\phi^2}\,\right];\-d\vec{\mathrm{x}}}\!.
\end{align*}
Exploiting the first Green's identity, one gets
\begin{align*}
    E_\phi[\mathcal{G}^\lambda\xi/\phi]\-=\!\lim_{\epsilon\to 0^+}\!&\left[\integrate[\partial D^{\+\epsilon}]{\!\left(\mathcal{G}^\lambda\xi\,\partial_{\vec{n}}\mathcal{G}^\lambda\xi-\abs{\mathcal{G}^\lambda\xi}^2\partial_{\vec{n}}\ln\phi\right);\mspace{-10mu}d\vec{\mathrm{s}}}+\right.\\
    &\;\:-\left.\!\!\integrate[D^{\+\epsilon}]{\!\left(\!\mathcal{G}^\lambda\xi\,\Delta\mathcal{G}^\lambda\xi\--\abs{\mathcal{G}^\lambda\xi}^2\+\Delta\ln\phi-\abs{\mathcal{G}^\lambda\xi}^2\,\frac{\abs{\nabla\phi}^2\!\-}{\,\phi^2}\:\right)\!;\-d\vec{\mathrm{x}}}\right]\!,
\end{align*}
where $\partial_{\vec{n}}$ denotes the outer normal derivative along $\partial D^{\+\epsilon}$.
Furthermore, we stress that
\begin{equation*}
    \partial D^{\+\epsilon}=\nquad\bigcup_{\{i,\+j\}\+\in\,\mathcal{P}_N}\nquad\big\{(\vec{x}_1,\ldots,\vec{x}_N)\-\in\R^{3N}\,\big|\;\abs{\vec{x}_i\--\vec{x}_j}\-=\epsilon,\,\min\limits_{\substack{\{k,\,\ell\}\+\in\,\mathcal{P}_N\,:\\ \{k,\,\ell\}\+\neq\+\{i,\,j\}}}\!\abs{\vec{x}_k\--\vec{x}_\ell}>\epsilon\big\}=\vcentcolon \!\!\bigcup_{\sigma\+\in\+\mathcal{P}_N} \!\-\partial D^{\+\epsilon}_{\-\sigma}\+,
\end{equation*}
where $\{\partial D^{\+\epsilon}_{\-\sigma}\}_{\sigma\+\in\,\mathcal{P}_N}$ are pairwise disjoint.
Taking into account the following identities for any $\vec{\mathrm{x}}\-\in\-\R^{3N}\-\setminus \pi$
\begin{gather}
    (\mathcal{H}_0\-+\-\lambda)\+ \mathcal{G}^\lambda\xi(\vec{\mathrm{x}})\-=0,\\
    (\Delta\ln\phi)(\vec{\mathrm{x}})=\frac{(\Delta\+\phi)(\vec{\mathrm{x}})}{\phi\mspace{0.75mu}(\vec{\mathrm{x}})}-\frac{\abs{\nabla\phi\mspace{0.75mu}(\vec{\mathrm{x}})}^2\!\-}{\,\phi^2(\vec{\mathrm{x}})},
\end{gather}
the computation in the previous integral yields
\begin{align}
    E_\phi[\mathcal{G}^\lambda\xi/\phi]\-=\!\-\integrate[\R^{3N}]{\abs{\mathcal{G}^\lambda\xi}^2\!\left(\tfrac{\Delta\phi}{\phi}-\lambda\right);\mspace{-10mu}d\vec{\mathrm{x}}}+\tfrac{1}{2}\-\lim_{\epsilon\to 0^+}\!\integrate[\partial D^{\+\epsilon}]{\phi^2(\vec{\mathrm{x}})\,\partial_{\vec{n}}\tfrac{\abs{\mathcal{G}^\lambda\xi}^2\!\-}{\phi^2};\mspace{-10mu}d\vec{\mathrm{s}}}\nonumber\\
    =\!\-\integrate[\R^{3N}]{\abs{\mathcal{G}^\lambda\xi}^2\!\left(\tfrac{\Delta\phi}{\phi}-\lambda\right);\mspace{-10mu}d\vec{\mathrm{x}}}+\frac{1}{2}\-\sum_{\sigma\+\in\+\mathcal{P}_N}\lim_{\epsilon\to 0^+}\!\integrate[\partial D^{\+\epsilon}_{\-\sigma}]{\phi^2(\vec{\mathrm{x}})\,\partial_{\vec{n}_\sigma}\!\tfrac{\abs{\mathcal{G}^\lambda\xi}^2\!\-}{\phi^2};\mspace{-10mu}d\vec{\mathrm{s}}}.\label{middleStepDirichletParticularCase}
\end{align}
Here it is possible to compute the outer normal derivative along each $\partial D^{\+\epsilon}_{\-\sigma}$, since they all are a sort of cylindrical hypersurfaces cut in a non-trivial manner nearby the singularities represented by the intersection of the coincidence hyperplanes.
More precisely, the tangent plane of $\partial D^{\+\epsilon}_{\-\sigma\+=\{i,\,j\}}$ at a fixed point $(\vec{x}_1,\ldots,\vec{x}_N)\-\in\R^{3N}$ is $\{(\vec{v}_1,\ldots,\vec{v}_N)\-\in\R^{3N}\,\big|\;(\vec{x}_i-\vec{x}_j)\-\cdot\- (\vec{v}_i-\vec{v}_j)=0\}$.
Thus,
\begin{equation}\label{bloodyNormal}
    \partial_{\vec{n}_{\sigma=\{i,\,j\}}}\!=-\frac{1}{\!\-\sqrt{2}\,\abs{\vec{x}_i\!-\-\vec{x}_j}}\, (\vec{x}_i\!-\-\vec{x}_j,\vec{x}_j\!-\-\vec{x}_i,\vec{0})\cdot (\nabla_{\!\vec{x}_i}, \nabla_{\!\vec{x}_j},\nabla_{\!\vec{X}_{\-\sigma}})=-\frac{(\vec{x}_i\!-\-\vec{x}_j)\-\cdot\- (\nabla_{\!\vec{x}_i}\!-\nabla_{\!\vec{x}_j})}{\sqrt{2}\,\abs{\vec{x}_i\!-\-\vec{x}_j}}.
\end{equation}
Hence, because of~\eqref{dirichletGeneralChargeQF} and the symmetry of the integrand in exchanging any couple $\vec{x}_i\-\longleftrightarrow\-\vec{x}_j\+$, we obtain
\begin{equation}\label{startingPointChargeDirichletQF}
    \Phi_\phi^\lambda[\xi]=\tfrac{N(N\--1)\!}{4}\-\lim_{\epsilon\to 0^+}\!\integrate[\partial D^{\+\epsilon}_{\-\sigma}]{\phi^2(\vec{\mathrm{x}})\,\partial_{\vec{n}_\sigma}\!\tfrac{\abs{\mathcal{G}^\lambda\xi}^2\!\-}{\phi^2};\mspace{-10mu}d\vec{\mathrm{s}}},\qquad\forall\sigma\-\in\-\mathcal{P}_N.
\end{equation}
    In order to have a better representation of the previous identity, we adopt the change of variables encoded by the unitary operator $U_\sigma\+$, given by~\eqref{unitaryCoordinate}.
    \comment{We stress that for generic $f,g\-\in\- \hilbert_N \cap H^2(\R^{3N}\setminus\pi)$, one has\footnote{We are replacing each domain $\partial D^{\+\epsilon}_{\{i,j\}}$ with $\{(\vec{x}_1,\ldots,\vec{x}_N)\-\in\-\R^{3N}\:|\;\abs{\vec{x}_i\!-\-\vec{x}_j}\-=\-\epsilon\}$ taking into account subleading orders vanishing as $\epsilon$ goes to zero.}
    \begin{equation*}
        \integrate[D^{\+\epsilon}]{\nabla \conjugate{g}\cdot\-\nabla f;\-d\vec{\mathrm{x}}}=\!\-\integrate[U_\sigma\- D^{\+\epsilon}]{2\+\nabla_{\!\vec{r}}\+U_\sigma\+ \conjugate{g} \cdot \-\nabla_{\!\vec{r}}\+ U_\sigma f\-+\tfrac{1}{2}\+\nabla_{\!\vec{x}}\+U_\sigma\+\conjugate{g}\cdot\-\nabla_{\!\vec{x}}\+U_\sigma f\-+\nabla_{\!\vec{X_\sigma}}U_\sigma\+\conjugate{g}\cdot\-\nabla_{\!\vec{X_\sigma}}U_\sigma f;\nquad d\vec{r}d\vec{x}d\vec{X}_{\-\sigma}},
    \end{equation*}
    simply because $U_\sigma\nabla_{\!\vec{x}_i}\+\adj{U}_\sigma\-=\frac{1}{2}\+\nabla_{\!\vec{x}}\-+\nabla_{\!\vec{r}}$ and $U_\sigma\nabla_{\!\vec{x}_j}\+\adj{U}_\sigma\-=\frac{1}{2}\+\nabla_{\!\vec{x}}\--\nabla_{\!\vec{r}}$. 
    Therefore, the first Green's identity reads
    \begin{equation*}
    \begin{split}
        \integrate[D^{\+\epsilon}]{\nabla \conjugate{g}\cdot\-\nabla f;\-d\vec{\mathrm{x}}}=&\,-\!\-\integrate[U_\sigma\-D^{\+\epsilon}]{(U_\sigma\+ \conjugate{g})\,(2\+\Delta_{\vec{r}} \+U_\sigma f+\tfrac{1}{2}\+\Delta_{\vec{x}}\+U_\sigma f+\Delta_{\-\vec{X_\sigma}}U_\sigma f);\mspace{-16.5mu}d\vec{r}d\vec{x}d\vec{X}_{\-\sigma}}\,+\\
        &-2\!\-\integrate[\R^{3(N\--1)}]{;\mspace{-33mu}d\vec{x}d\vec{X}_{\-\sigma}}\!\!\!\integrate[\abs{\vec{r}}\+=\+\epsilon]{(U_\sigma\+ \conjugate{g})\left(\tfrac{\vec{r}}{\epsilon}\-\cdot\-\nabla_{\!\vec{r}} \+U_\sigma f\right)\!;\mspace{-16mu}d\vec{r}}+\oSmall{1}\-,\quad \text{as }\,\epsilon\longrightarrow 0^+,
    \end{split}
    \end{equation*}
    where $-\frac{\vec{r}\,\cdot\+\nabla_{\!\vec{r}}}{\abs{\vec{r}}}$ is the outer normal derivative in this framework.
    }
    We notice that $\partial D_{\-\sigma}^{\+\epsilon}$ transforms in the following way
    \begin{equation*}
    \partial \+U_\sigma D_{\-\sigma}^{\+\epsilon}=\Big\{\-(\vec{r},\vec{x},\vec{X}_{\-\sigma})\-\in\R^{3N}\,\Big|\;\abs{\vec{r}}=\epsilon, \min_{\substack{\{k,\,\ell\}\+\in\,\mathcal{P}_N\+:\\\{k,\,\ell\}\,\cap\,\sigma\+=\,\emptyset}}\!\abs{\vec{x}_k\!-\-\vec{x}_\ell}\->\mspace{-0.75mu}\epsilon, \min_{\substack{k\+\in\mspace{2.25mu}\{1,\,\ldots\+,\,N\}\+:\\ k\+\notin\+\sigma}}\!\abs{\vec{x}_k\!-\-\vec{x}\-\pm\tfrac{\vec{r}}{2}}\->\mspace{-0.75mu}\epsilon\Big\}\subset \R^{3N-1}.
    \end{equation*}
    Additionally, the outer normal derivative is now given by $-\frac{\vec{r}\,\cdot\+\nabla_{\!\vec{r}}}{\abs{\vec{r}}}$.
    We claim that equation~\eqref{startingPointChargeDirichletQF} can be represented in terms of these coordinates in the following way
    \begin{equation}\label{goodPointChargeDirichletQF}
        \Phi_\phi^\lambda[\xi]=\tfrac{N(N\--1)\!}{2}\-\lim_{\epsilon\to 0^+}\!\integrate[\R^{3(N-1)}]{;\mspace{-33mu}d\vec{x}d\vec{X}_{\-\sigma}}\!\!\!\integrate[\abs{\vec{r}}\+=\+\epsilon]{(U_\sigma\+\phi)^2(\vec{r}\-,\vec{x},\vec{X}_{\-\sigma})\!\left(\--\tfrac{\vec{r}}{\epsilon}\-\cdot\-\nabla_{\!\vec{r}}\+\tfrac{\abs{U_\sigma\+\mathcal{G}^\lambda\xi}^2\!\-}{(U_\sigma\+\phi)^2}\+\right);\mspace{-16mu} d\vec{r}}\!.
    \end{equation}
    Indeed, we shall see in the following that the integrand can be written as $\oBig{\frac{1}{r^2}}\-f(\vec{x},\vec{X}_{\-\sigma})$ for a proper function $f\!\in\!\Lp{1}[\R^{3(N-1)}\!,\,d\vec{x}d\vec{X}_{\-\sigma}]$, and therefore $\partial \+U_\sigma D^{\+\epsilon}_{\-\sigma}$ can be replaced with $\{(\vec{r},\vec{x},\vec{X}_{\-\sigma})\-\in\R^{3N}\,\big|\;\abs{\vec{r}}\-=\-\epsilon\}$, since the rest is a set whose Hausdorff measure (in $\R^{3N-1}$) is vanishing as $\epsilon$ goes to zero.
    Furthermore, for generic $f,g\-\in\- \hilbert_N \cap H^2(\R^{3N}\setminus\pi_\sigma)$, one has
    \begin{equation*}
        \integrate[\abs{\vec{x}_i-\+\vec{x}_j}\+>\+\epsilon]{\nabla \conjugate{g}\cdot\-\nabla f;\mspace{-51mu} d\vec{\mathrm{x}}}=\!\-\integrate[\abs{\vec{r}}\+>\+\epsilon]{2\+\nabla_{\!\vec{r}}\+U_\sigma\+ \conjugate{g} \cdot \-\nabla_{\!\vec{r}}\+ U_\sigma f\-+\tfrac{1}{2}\+\nabla_{\!\vec{x}}\+U_\sigma\+\conjugate{g}\cdot\-\nabla_{\!\vec{x}}\+U_\sigma f\-+\nabla_{\!\vec{X_\sigma}}U_\sigma\+\conjugate{g}\cdot\-\nabla_{\!\vec{X_\sigma}}U_\sigma f;\nquad d\vec{r}d\vec{x}d\vec{X}_{\-\sigma}},
    \end{equation*}
    simply because $U_\sigma\nabla_{\!\vec{x}_i}\+\adj{U}_\sigma\-=\frac{1}{2}\+\nabla_{\!\vec{x}}\-+\nabla_{\!\vec{r}}$ and $U_\sigma\nabla_{\!\vec{x}_j}\+\adj{U}_\sigma\-=\frac{1}{2}\+\nabla_{\!\vec{x}}\--\nabla_{\!\vec{r}}$. 
    Therefore, the first Green's identity reads
    \begin{equation*}
    \begin{split}
        \integrate[\abs{\vec{x}_i-\+\vec{x}_j}\+>\+\epsilon]{\nabla \conjugate{g}\cdot\-\nabla f;\mspace{-51mu} d\vec{\mathrm{x}}}=&\,-\!\-\integrate[\abs{\vec{r}}\+>\+\epsilon]{(U_\sigma\+ \conjugate{g})\,(2\+\Delta_{\vec{r}} \+U_\sigma f+\tfrac{1}{2}\+\Delta_{\vec{x}}\+U_\sigma f+\Delta_{\-\vec{X_\sigma}}U_\sigma f);\mspace{-16.5mu}d\vec{r}d\vec{x}d\vec{X}_{\-\sigma}}\,+\\
        &-2\!\-\integrate[\R^{3(N\--1)}]{;\mspace{-33mu}d\vec{x}d\vec{X}_{\-\sigma}}\!\!\!\integrate[\abs{\vec{r}}\+=\+\epsilon]{(U_\sigma\+ \conjugate{g})\left(\tfrac{\vec{r}}{\epsilon}\-\cdot\-\nabla_{\!\vec{r}} \+U_\sigma f\right)\!;\mspace{-16mu}d\vec{r}}.
    \end{split}
    \end{equation*}
    This means that applying the change of variables encoded by $U_\sigma$ in equation~\eqref{startingPointChargeDirichletQF}, taking account of~\eqref{bloodyNormal}, we find out the proper induced transformation of the hypersurface in $\R^{3N\mspace{-0.75mu}-1}$ which is needed to be applied to the Hausdorff measure $d\vec{\mathrm{s}}\mspace{-0.75mu}\longmapsto\mspace{-2.25mu}\sqrt{2}\, d\vec{x}d\vec{X}_{\-\sigma}\+d\vec{r}|_{\abs{\vec{r}}\+=\+\epsilon}$.
    This proves equation~\eqref{goodPointChargeDirichletQF}.\newline
    We stress that the quadratic form $\Phi^\lambda_\phi$, and therefore the operator $-\Delta_\phi$, remains unchanged if we multiply $\phi$ by any nonzero constant, owing to equation~\eqref{goodPointChargeDirichletQF},
    namely,
    \begin{equation}
        -\Delta_{\+c\,\phi}=-\Delta_\phi,\qquad \forall c\in\R\-\smallsetminus\-\{0\}.
    \end{equation}
    Next, we compute the asymptotic expansion of the integrand for $r$ small.\newline
    First of all, one can check that the assumption $\theta\-\in\! H^2(\Rplus)$ implies $\tfrac{\theta(\abs{\+\vec{\cdot}\+})}{\abs{\+\vec{\cdot}\+}}\-\in\! H^2(\R^3)$, and therefore one has the following asymptotic behavior in the topology induced by $\Lp*{2}[\R^{3(N-1)}]$ for the weight function $\phi$ defined in~\eqref{genericPhiDF}
    \begin{subequations}\label{Dirichlet4Body}
    \begin{equation}\label{genericPhiExpansion}
        (U_\sigma \+\phi)(\vec{r},\vec{x},\vec{X}_{\-\sigma})=\frac{1}{r}+\alpha_0+A_\theta(\vec{x},\vec{X}_{\-\sigma})+\oSmall{1}\-, \qquad r\longrightarrow 0^+,
    \end{equation}
    where we have defined the shortcut
    \begin{equation}
        A_\theta(\vec{x},\vec{X}_{\-\sigma}) = 2\sum^N_{\substack{\ell\,=\+1\,:\\ \ell\,\notin\+ \sigma}} \frac{\theta(\abs{\vec{x}\--\-\vec{x}_\ell})}{\abs{\vec{x}\--\-\vec{x}_\ell}}+\nquad\sum_{\substack{1\leq\+k\,<\, \ell\+\leq N\\ k\,\notin\+\sigma, \:\ell\,\notin\+\sigma}} \mspace{-12mu}\frac{\theta(\abs{\vec{x}_k\!-\-\vec{x}_\ell})}{\abs{\vec{x}_k\!-\-\vec{x}_\ell}}.
    \end{equation}
    \end{subequations}
    In fact, under the assumption $\theta\-\in\- H^2(\Rplus)\!\subset\! C^1(\Rplus)$, one has $\!\lim\limits_{r\to 0^+}\! \frac{\theta(r)\+-\+1}{r}\-=\-\theta'\-(0)$.
    Moreover, a similar expansion for the potential, holding in the weak topology near $\pi_\sigma$, is provided by~\eqref{potentialAsymptotics}
    \begin{equation*}
        (U_\sigma\+\mathcal{G}^\lambda\xi)(\vec{r},\vec{x},\vec{X}_{\-\sigma})=\frac{\xi(\vec{x},\vec{X}_{\-\sigma})}{r}-(\Gamma^{\+\sigma, \+\lambda}\+\xi)(\vec{x},\vec{X}_{\-\sigma})+\oSmall{1}\-,\qquad \text{as }\, r\longrightarrow0^+.
    \end{equation*}
    Elementary calculations yield
        \begin{equation}\label{sayElementaryCalculations}
        \begin{split}
        \frac{\abs{(U_\sigma\+\mathcal{G}^\lambda\xi)(\vec{r},\vec{x},\vec{X}_{\-\sigma})}^2\!\-}{(U_\sigma\+\phi)^2(\vec{r},\vec{x},\vec{X}_{\-\sigma})}=&\:\abs{\xi(\vec{x},\vec{X}_{\-\sigma})}^2\--2r\+\Re\,\conjugate*{\xi(\vec{x},\vec{X}_{\-\sigma})} (\Gamma^{\+\sigma, \+\lambda}\+\xi)(\vec{x},\vec{X}_{\-\sigma})\++\\
        &-2r\+\big[\alpha_0+A_\theta(\vec{x},\vec{X}_{\-\sigma})\big]\abs{\xi(\vec{x},\vec{X}_{\-\sigma})}^2\-+\oSmall{r}\!.
        \end{split}
    \end{equation}
    Thus, the integrand in equation~\eqref{goodPointChargeDirichletQF} for $r\longrightarrow 0^+$ reads
    \begin{equation*}
        (U_\sigma\+\phi)^2\!\left[\--\tfrac{\vec{r}}{\abs{\vec{r}}}\mspace{-2.25mu}\cdot\mspace{-2.25mu}\nabla_{\!\vec{r}}\+\tfrac{\abs{(U_\sigma\+\mathcal{G}^\lambda\xi)}^2\!\!}{(U_\sigma\+\phi)^2}\+\right]\!(\vec{r}\-,\vec{x},\vec{X}_{\-\sigma}\-)\-=\-\tfrac{\-2\-\big[\-A_\theta(\vec{x},\+\vec{X}_{\-\sigma})+\+\alpha_0\-\big]\-\abs{\xi(\vec{x},\+\vec{X}_{\-\sigma})}^2\++2\+\Re\:\conjugate*{\xi(\vec{x},\+\vec{X}_{\-\sigma})} \+(\Gamma^{\+\sigma, \+\lambda}\+\xi)(\vec{x},\+\vec{X}_{\-\sigma})\!}{r^2}+\oSmall{\tfrac{1}{r^2}}\!.
    \end{equation*}
    Hence, we have obtained
    \begin{equation}\label{isThisOurQF?}
        \Phi_\phi^\lambda[\xi]=4\pi \+N(N\!-\-1)\-\!\integrate[\R^{3(N\--1)}]{\big[A_\theta(\vec{x},\vec{X}_{\-\sigma})+\alpha_0\big]\abs{\xi(\vec{x},\vec{X}_{\-\sigma})}^2\-+\Re\,\conjugate*{\xi(\vec{x},\vec{X}_{\-\sigma})}(\Gamma^{\+\sigma, \+\lambda}\+\xi)(\vec{x},\vec{X}_{\-\sigma});\mspace{-33mu}d\vec{x}d\vec{X}_{\-\sigma}}.
    \end{equation}
    In particular, notice that the first term on the right-hand side of~\eqref{isThisOurQF?} coincides with $\Phi_{\mathrm{reg}}[\xi]$, given $\gamma\-=\-2$ and $\theta\in H^2(\Rplus)$ with $\theta(0)=1$.
    Furthermore, by construction, the second term satisfies (see Appendix~\ref{heuristics})
    \begin{equation}\label{claimDirichlet}
    \Re\!\integrate[\R^{3(N-1)}]{\conjugate*{\xi(\vec{x},\vec{X}_{\-\sigma})} (\Gamma^{\+\sigma, \+\lambda}\+\xi)(\vec{x},\vec{X}_{\-\sigma});\mspace{-33mu}d\vec{x}d\vec{X}_{\-\sigma}}=(\Phi^\lambda_{\mathrm{diag}}\!+\Phi^\lambda_{\mathrm{off,\,0}}\-+\Phi^\lambda_{\mathrm{off,\+1}})[\xi].
    \end{equation}
    
\end{proof}
In conclusion, we provide an explicit example represented by
\begin{equation}
    \theta:\,r\:\longmapsto \, \mathop{\mathrm{sech}}\!\Big(\-\sqrt{\tfrac{Ne_0}{2}\+}\+ r\-\Big),\qquad r, e_0\geq 0.
\end{equation}
Despite this repulsive many-body interaction depends on the number of particles, it is settled so that it vanishes pointwise almost everywhere as $N$ grows large.
In this case one has $\theta'\-(0)=0$, so that if we fix $\gamma=2$ our many-body Hamiltonian encoding a contact interaction can be written in terms of a Dirichlet form where the weight function is
$$\phi(\vec{x}_1,\ldots,\vec{x}_N)=\nquad\sum_{1\leq\+k\,<\,\ell\+\leq N}\mspace{-12mu} \frac{\mathop{\mathrm{sech}}\!\Big(\-\sqrt{\frac{Ne_0}{2}\+}\+\abs{\vec{x}_k\!-\-\vec{x}_\ell}\Big)}{\abs{\vec{x}_k\!-\-\vec{x}_\ell}}+\alpha_0.$$
Our claim is that the many-body Hamiltonian in $\hilbert_N$ defined in this way is not only lower-bounded, but also stable of second kind, namely there holds
$\inf\spectrum{\mathcal{H}}\geq -const \,N$ (improving the estimates obtained in~\eqref{unoptimalLowerBound}).
To this end, let us briefly show that for any $\theta\in H^2(\Rplus)$
$$\Delta\phi(\vec{x}_1,\ldots,\vec{x}_N)=2\nquad \sum_{1\leq\+k\,<\,\ell\+\leq N}\mspace{-12mu} \frac{\theta''(\abs{\vec{x}_k\!-\-\vec{x}_\ell})}{\abs{\vec{x}_k\!-\-\vec{x}_\ell}}.$$
Ordering properly the elements of the sum one gets 
\begin{align*}
    \Delta\phi(\vec{x}_1,\ldots,\vec{x}_N)&=\sum_{k\+=\+1}^N\Delta_{\vec{x}_k}\nquad\sum_{1\leq\+i\,<\,j\+\leq N}\mspace{-12mu}\frac{\theta(\abs{\vec{x}_i\--\-\vec{x}_j})}{\abs{\vec{x}_i\--\-\vec{x}_j}}\\
    &=\sum_{k\+=\+1}^N\Delta_{\vec{x}_k}\Bigg[\sum_{\substack{1\leq\+i\,<\,j\+\leq N\,:\\ k\+\notin\+\{i,\,j\}}}\mspace{-12mu}\frac{\theta(\abs{\vec{x}_i\--\-\vec{x}_j})}{\abs{\vec{x}_i\--\-\vec{x}_j}}+\!\!\sum_{j\+=\+k+1}^N \frac{\theta(\abs{\vec{x}_k\--\-\vec{x}_j})}{\abs{\vec{x}_k\--\-\vec{x}_j}}+\sum_{i\+=\+1}^{k-1} \frac{\theta(\abs{\vec{x}_i\--\-\vec{x}_k})}{\abs{\vec{x}_i\--\-\vec{x}_k}}\Bigg]\\
    &=\sum_{k\+=\+1}^N\mathop{\mathrm{div}_{\vec{x}_k}}\sum_{\substack{\ell\+=\+1\,:\\ \ell\+\neq\+ k}}^N\frac{\vec{x}_k\--\-\vec{x}_\ell}{\abs{\vec{x}_k\--\-\vec{x}_\ell}^2}\Bigg[\theta'\-(\abs{\vec{x}_k\--\-\vec{x}_\ell})- \frac{\theta(\abs{\vec{x}_k\--\-\vec{x}_\ell})}{\abs{\vec{x}_k\--\-\vec{x}_\ell}}\Bigg].
\end{align*}
Since one has $\mathop{\mathrm{div}_{\vec{x}_k}} \frac{\vec{x}_k-\,\vec{x}_\ell}{\abs{\vec{x}_k-\,\vec{x}_\ell}^2}=\frac{1}{\abs{\vec{x}_k-\,\vec{x}_\ell}^2}$, one obtains
\begin{align*}
    \Delta\phi(\vec{x}_1,\ldots,\vec{x}_N)= \sum_{k\+=\+1}^N \sum_{\substack{\ell\+=\+1\,:\\\ell\+\neq\+ k}}^N &\,\frac{1}{\abs{\vec{x}_k\--\-\vec{x}_\ell}^2}\Bigg[\theta'\-(\abs{\vec{x}_k\--\-\vec{x}_\ell})- \frac{\theta(\abs{\vec{x}_k\--\-\vec{x}_\ell})}{\abs{\vec{x}_k\--\-\vec{x}_\ell}}\Bigg]+\\
    &+\frac{1}{\abs{\vec{x}_k\--\-\vec{x}_\ell}}\Bigg[\theta''\-(\abs{\vec{x}_k\--\-\vec{x}_\ell})- \frac{\theta'(\abs{\vec{x}_k\--\-\vec{x}_\ell})}{\abs{\vec{x}_k\--\-\vec{x}_\ell}}+\frac{\theta(\abs{\vec{x}_k\--\-\vec{x}_\ell})}{\abs{\vec{x}_k\--\-\vec{x}_\ell}^2}\Bigg]
\end{align*}
which proves the claim.
Next, we stress that
\begin{equation}
    \frac{d^2}{dr^2} \mathop{\mathrm{sech}}(a\+r)=a^2\- \mathop{\mathrm{sech}}(a\+r)\big[1-2\mathop{\mathrm{sech}^2}(a\+r)\big]\leq a^2\-\mathop{\mathrm{sech}}(a\+r), \qquad \forall a\geq 0,
\end{equation}
therefore, assuming $\alpha_0\-\geq\- 0$ one obtains $\frac{\Delta\phi}{\phi}\leq N e_0$, yielding
\begin{equation}\label{secondKindStability}
\inf\spectrum{\mathcal{H}}\geq - N \+e_0.
\end{equation}
Of course the example provided by~\cite{AHK} is quite similar to this one if one sets $m\-=\-\sqrt{\frac{Ne_0}{2}}$, but the crucial difference is that the scattering length $\frac{1}{m}$ in this case must vanish for $N$ large, whereas we are providing a simple modification which allows to keep the repulsive interaction $\theta$ and the scattering length $-\frac{1}{\alpha_0}$ independent from each other (for instance in the so-called unitary limit $\alpha_0=0$ we have the non-trivial lower bound~\eqref{secondKindStability} for the many-body Hamiltonian, while considering~\cite[example 4]{AHK} $m=0$ implies $\mathcal{H}\geq 0$).

\appendix

\section{Heuristic Derivation of the Quadratic Form}\label{heuristics}

Here we provide a heuristic discussion meant to justify the definition of the quadratic form $\mathcal{Q}$ given in~\eqref{QF}.
Considering a vector $\psi\!\in\!\hilbert_N \cap H^2(\R^{3N}\setminus\pi)$ fulfilling boundary condition~\eqref{mfBC}, our aim is to compute the energy form $\scalar{\psi}{\tilde{\mathcal{H}}\psi}$ associated with the formal Hamiltonian $\tilde{\mathcal{H}}$ discussed in the introduction.
Recalling that, by construction, the Hamiltonian acts as $\mathcal{H}_0$ outside $\pi$, given $\epsilon\->\- 0$, let $$\define{D_\epsilon\!;\-\big\{\!(\vec{x}_1,\ldots,\vec{x}_N)\!\in\R^{3N}\,\big | \min\limits_{1\leq\+i\,<\,j\+\leq N}\abs{\vec{x}_i\--\vec{x}_j}\!>\epsilon\big\}}.$$
Then
\begin{equation}\label{freeLimitQF}
    \scalar{\psi}{\tilde{\mathcal{H}}\psi}=\lim_{\epsilon\to 0}\- \integrate[D_\epsilon]{\conjugate*{\psi(\vec{x}_1,\ldots, \vec{x}_N)}\+(\mathcal{H}_0\+\psi)(\vec{x}_1,\ldots, \vec{x}_N); \-d\vec{x}_1\cdots d\vec{x}_N}.
\end{equation}
Introducing the decomposition~\eqref{psiDec}, $\psi=w^\lambda\!+\mathcal{G}^\lambda\xi\+$, equation~\eqref{freeLimitQF} reads
\begin{equation}
    \scalar{\psi}{\tilde{\mathcal{H}}\psi}=\scalar{w^\lambda}{(\mathcal{H}_0+\lambda)w^\lambda}-\lambda\norm{\psi}^2\!+\scalar{\mathcal{G}^\lambda\xi}{(\mathcal{H}_0+\lambda)w^\lambda},
\end{equation}
since $(\mathcal{H}_0+\-\lambda)\+\mathcal{G}^\lambda\xi=0$ in $D_\epsilon\+$.
The last term can be simplified using~\eqref{potentialDistributionalDef}
\begin{align*}
    \scalar{\psi}{\tilde{\mathcal{H}}\psi}&=\scalar{w^\lambda}{(\mathcal{H}_0+\lambda)w^\lambda}-\lambda\norm{\psi}^2\!+8\pi\,\scalar{\xi}{\!{\textstyle\sum\limits_{\sigma\+\in\+\mathcal{P}_N}}\!\- w^\lambda|_{\pi_\sigma}}[\hilbert*_{N-1}]\\[-2.5pt]
    &=\scalar{w^\lambda}{(\mathcal{H}_0+\lambda)w^\lambda}-\lambda\norm{\psi}^2\!+8\pi\,\scalar{\xi}{\!{\textstyle \sum\limits_{\sigma\+\in\+\mathcal{P}_N}\!}(\Gamma^{\+\sigma,\+\lambda}\!+\Gamma^{\+\sigma}_{\!\mathrm{reg}})\+\xi}[\hilbert*_{N-1}]\+.
\end{align*}
In the last step, we took account of~\eqref{regularComponentTraced}.
Moreover, considering the definition of $\Gamma^{\+\sigma,\+\lambda}\-$,
\begin{equation}\label{inProgressQF1}
    \scalar{\psi}{\tilde{\mathcal{H}}\psi}=\scalar{w^\lambda}{(\mathcal{H}_0+\lambda)w^\lambda}-\lambda\norm{\psi}^2\!+8\pi\,\scalar{\xi}{\!{\textstyle \sum\limits_{\sigma\+\in\+\mathcal{P}_N}}\!\big(\Gamma_{\!\mathrm{diag}}^\lambda\-+\!{\textstyle \sum\limits_{\nu\+\neq\+\sigma}}\Gamma_{\!\mathrm{off};\+|\nu\+\cap\+\sigma|}^{\+\nu\+\smallsetminus\sigma,\+\lambda}\-+\Gamma_{\!\mathrm{reg}}^{\+\sigma}\big)\+\xi}[\hilbert*_{N-1}]\+.
\end{equation}
Exploiting the bosonic symmetry and considering definitions~\eqref{GammaDefs}
$$\scalar{\psi}{\tilde{\mathcal{H}}\psi}=\scalar{w^\lambda}{(\mathcal{H}_0+\lambda)w^\lambda}-\lambda\norm{\psi}^2\!+4\pi\+ N(N\!-\-1)\,\scalar{\xi}{\big(\Gamma_{\!\mathrm{diag}}^\lambda\!+\Gamma_{\!\mathrm{off};\+0}^\lambda\-+\Gamma_{\!\mathrm{off};\+1}^\lambda\-+\Gamma_{\!\mathrm{reg}}\big)\+\xi}[\hilbert*_{N-1}]\+.$$
In conclusion, adopting suitable changes of variables, where required, one obtains
\begin{equation}\label{inProgressQF2}
    \scalar{\psi}{\tilde{\mathcal{H}}\psi}=\scalar{w^\lambda}{(\mathcal{H}_0+\lambda)w^\lambda}-\lambda\norm{\psi}^2\!+4\pi\+ N (N\!-\-1)(\Phi^\lambda_{\mathrm{diag}}\!+\Phi^\lambda_{\mathrm{off};\+0}\-+\Phi^\lambda_{\mathrm{off};\+1}\-+\Phi_{\mathrm{reg}})[\xi],
\end{equation}
coming up with the definition of the quadratic form $\mathcal{Q}$ provided by~\eqref{QF}.
We point out that for any given $\sigma\-\in\-\mathcal{P}_N$
\begin{align*}
    &\left|\left\{\nu\-\in\-\mathcal{P}_N\,\big|\;|\nu\cap\sigma|\-=0\right\}\right|=\tfrac{(N-2)(N-3)}{2}, & \left|\left\{\nu\-\in\-\mathcal{P}_N\,\big|\;|\nu\cap\sigma|\-=\-1\right\}\right|=2\+(N\!-\-2)\+.
\end{align*}

\comment{
\section{Instability of the 4-Body singularity}

In this section, our aim is to prove that the four-body contribution in boundary condition~\eqref{mfBC} is crucial for our results whenever $N\!\geq\- 5$.
More precisely, we show that the quadratic form in $\hilbert_N$ associated with~\eqref{formalH} is unbounded from below if the boundary condition only takes account of the three-body repulsion.\newline[5]
Here we therefore focus on the following quadratic form
\begin{align}
    &\dom{\breve{\mathcal{Q}}}\-\vcentcolon= \!\left\{\psi \in \hilbert_N \,\big |\;\psi-\mathcal{G}^\lambda\xi=w^\lambda\-\in\- H^1(\R^{3N}),\:\xi \in \dom{\breve{\Phi}^{\lambda}}, 
    \:\lambda\->\-0\right\}\-,\nonumber\\
    \label{QF4}&\mspace{7.5mu}\breve{\mathcal{Q}}[\psi]\vcentcolon=\|\mathcal{H}_0^{\frac{1}{2}}w^\lambda\|^2+\lambda\|w^\lambda\|^2\--\lambda\- \norm{\psi}^2\-+\breve{\Phi}^\lambda[\xi], 
\end{align}
where $\breve{\Phi}^\lambda$ is now missing the $4$-body regularizing term
\begin{gather}
 \dom{\breve{\Phi}^{\lambda}} =\dom{\Phi^\lambda}, \nonumber\\
  \label{phi4Definition} \breve{\Phi}^\lambda \vcentcolon= \Phi^\lambda-\Phi^{(4)}_{\mathrm{reg}}.
\end{gather}

\n Consider the trial sequence of vectors in $\dom{\breve{\mathcal{Q}}}$ given by
\begin{subequations}
\begin{gather}
    \psi^{\lambda,\+\beta}_n(\vec{x}_1,\ldots,\vec{x}_N)=\frac{(\mathcal{G}^\lambda\xi^\beta_n)(\vec{x}_1,\ldots,\vec{x}_N)}{\lVert\mathcal{G}^\lambda\xi^\beta_n\rVert},\qquad n\in\N,\:\beta,\lambda>0,\\[-5pt]
    \xi^\beta_n(\vec{z},\vec{y}_1,\ldots,\vec{y}_{N-2})\-=\frac{\sqrt{2\+}\+n^2}{\!\-\sqrt{(N\!-\-2)(N\!-\-3)}\,}\!\-\sum_{\substack{\nu\+\in\+\mathcal{P}_{N\--2}\+:\\\nu\+=\+\{k,\,\ell\}}}\!\!\!\! g_\beta\big(n\+(\vec{y}_k\!-\vec{y}_\ell)\-\big)\,f(\vec{z},\tfrac{\vec{y}_k+\+\vec{y}_\ell}{2},\vec{y}_1,\ldots\check{\vec{y}}_\nu\ldots,\vec{y}_{N-2}),\label{trialChargeSequence}
\end{gather}
\end{subequations}
where, denoting by $\Char{S}$ the characteristic function of the (measurable) set $S$
\begin{gather*}
    \FT{g}_\beta\-:\;\vec{q}\:\longmapsto\: \sqrt{\-\frac{\beta}{4\pi}\+}\, q^{\frac{\beta}{4}-2}\:\Char{\!\big(4^{-\frac{1}{\beta}}\!,\: 1\big)}\-(q),\\
    \FT{f}: (\vec{k},\vec{\kappa},\vec{p}_1,\ldots,\vec{p}_{N-4})\longmapsto \tfrac{1}{\pi^{3(N\--2)/4}}\,e^{-\frac{1}{2}\left(k^2+\+\kappa^2+\+p_1^2+\,\ldots\,+\+p_{N-4}^2\right)}.
\end{gather*}
With this choice we are considering a charge $\xi^\beta_n\!\in\!\dom{\breve{\Phi}^\lambda}$ which focuses around the $4$-body singularities as $n$ grows, according to equation~\eqref{off0PhiSplitted}.
Moreover, we stress for later purposes that $g_\beta$ has been chosen in such a way that
\begin{equation}\label{gShape}
  \FT{g}_\beta(\vec{q})=\sqrt{\frac{\beta}{4\pi}}\;\frac{u\+(\beta \+\ln{q})}{q^2},\qquad u:\;t\:\longmapsto e^{\frac{t}{4}}\,\Char{(-\ln{4},\,0)}(t),  
\end{equation}
where, in particular
$$\integrate[\R^3]{q\,\abs{\FT{g}_\beta(\vec{q})}^2;\!d\vec{q}}=1,\qquad\norm{u}[\Lp{2}[\R]]=1.$$
With this setting we have
\begin{equation}
\FT{\xi}^\beta_n(\vec{k},\vec{p}_1,\ldots,\vec{p}_{N-2})=\frac{1}{n}\:\frac{\sqrt{2}}{\!\-\sqrt{(N\!-\-2)(N\!-\-3)}\,}\-\!\sum_{\substack{\nu\+\in\+\mathcal{P}_{N\--2}\+:\\\nu\+=\+\{k,\,\ell\}}}\!\!\! \FT{g}_\beta\!\left(\-\tfrac{\vec{p}_k-\+\vec{p}_\ell}{2\+n}\-\right)\FT{f}(\vec{k},\vec{p}_k\!+\vec{p}_\ell,\vec{p}_1,\ldots\check{\vec{p}}_\nu\ldots,\vec{p}_{N-2}).\label{trialChargeFourier}
\end{equation}

\vspace{-0.15cm}

\n Then, let us ensure that $\psi^{\lambda,\+\beta}_n$ is well defined by showing that $\inf\limits_{n\+\in\+\N}\,\lVert\mathcal{G}^\lambda\xi^\beta_n\rVert>0$ for all $\beta,\lambda\->\-0$.
\vspace{-0.15cm}
\begin{align*}
    \lVert\mathcal{G}^\lambda\xi^\beta_n\rVert^2&=\frac{8}{\pi}\-\integrate[\R^{3N}]{\frac{\left\lvert\sum_{\mathcal{P}_N\+\ni\,\sigma=\{i,\+j\}} \FT{\xi}^\beta_n(\vec{p}_i\!+\vec{p}_j,\vec{p}_1,\ldots\check{\vec{p}}_\sigma\ldots,\vec{p}_N\-)\-\right\rvert^2}{(p_1^2+\ldots+p_N^2+\lambda)^2};\mspace{-10mu}d\vec{p}_1\-\cdots d\vec{p}_N}\\
    &\geq \frac{8}{\pi}\frac{N(N\!-\-1)}{2}\!\-\integrate[\R^{3N}]{\frac{ \abs{\FT{\xi}^\beta_n(\vec{p}_1\!+\vec{p}_2,\vec{p}_3,\ldots,\vec{p}_N\-)}^2\!\-}{(p_1^2+\ldots+p_N^2+\lambda)^2};\mspace{-10mu}d\vec{p}_1\-\cdots d\vec{p}_N},
\end{align*}
where we exploited the non-negativity of $\FT{\xi}^\beta_n$ and the bosonic symmetry.
Hence, we simplify the latter lower bound by changing properly the integration variables
\begin{align*}
    \lVert\mathcal{G}^\lambda\xi^\beta_n\rVert^2 &\geq \frac{4\+N(N\!-\-1)}{\pi}\!\-\integrate[\R^{3(N-1)}]{\abs{\FT{\xi}^\beta_n(\vec{k},\vec{p}_3,\ldots,\vec{p}_N\-)}^2;\mspace{-33mu}d\vec{k}d\vec{p}_3\cdots d\vec{p}_N}\!\!\integrate[\R^3]{\frac{1}{(\frac{1}{2}\+k^2\-+2\+\tau^2\-+p_3^2+\ldots+p_N^2+\-\lambda)^2\!\-};\-d\vec{\tau}}\\
    &=\sqrt{2\+}\+\pi\,N(N\!-\-1)\!\-\integrate[\R^{3(N-1)}]{\frac{\abs{\FT{\xi}^\beta_n(\vec{k},\vec{p}_1,\ldots,\vec{p}_{N-2}\-)}^2}{\!\sqrt{\frac{1}{2}\+k^2\-+p_1^2+\ldots+p_{N-2}^2+\-\lambda\+}\+};\mspace{-33mu}d\vec{k}d\vec{p}_1\-\cdots d\vec{p}_{N-2}}.
\end{align*}
Taking account of the explicit expression for $\FT{\xi}^\beta_n$, given by~\eqref{trialChargeFourier}, we again make use of its symmetry and non-negativity in order to achieve the following lower bound
\begin{align*}
    \lVert\mathcal{G}^\lambda\xi^\beta_n\rVert^2 &\geq\frac{\sqrt{2\+}\+\pi\,N(N\!-\-1)}{n^2}\!\-\integrate[\R^{3(N-1)}]{\frac{\abs{\FT{g}_\beta\big(\tfrac{\vec{q}}{n}\big)\+\FT{f}(\vec{k},\vec{\kappa},\vec{p}_3,\ldots,\vec{p}_{N-2})}^2}{\!\sqrt{\frac{1}{2}\+k^2\-+\frac{1}{2}\+\kappa^2\-+2\+q^2\-+p_3^2+\ldots+p_{N-2}^2+\-\lambda\+}\+};\mspace{-33mu}d\vec{q}d\vec{k}d\vec{\kappa}d\vec{p}_3\cdots d\vec{p}_{N-2}}\\
    &=\sqrt{2\+}\+\pi\,N(N\!-\-1)\!\-\integrate[\R^{3(N-1)}]{\frac{\abs{\FT{g}_\beta(\vec{q})\+\FT{f}(\vec{k},\vec{\kappa},\vec{p}_1,\ldots,\vec{p}_{N-4})}^2}{\!\sqrt{2\+q^2\-+\frac{1}{2n^2}\+(k^2\-+\kappa^2)+\frac{1}{n^2}(p_1^2+\ldots+p_{N-4}^2+\-\lambda)\+}\+};\mspace{-33mu}d\vec{q}d\vec{k}d\vec{\kappa}d\vec{p}_1\-\cdots d\vec{p}_{N-4}}\\
    &\geq\sqrt{2\+}\+\pi\,N(N\!-\-1)\!\-\integrate[\R^{3(N-1)}]{\frac{\abs{\FT{g}_\beta(\vec{q})\+\FT{f}(\vec{k},\vec{\kappa},\vec{p}_1,\ldots,\vec{p}_{N-4})}^2}{\!\sqrt{2\+q^2\-+\frac{1}{2}\+k^2\-+\frac{1}{2}\+\kappa^2\-+p_1^2+\ldots+p_{N-4}^2+\-\lambda\+}\+};\mspace{-33mu}d\vec{q}d\vec{k}d\vec{\kappa}d\vec{p}_1\-\cdots d\vec{p}_{N-4}}.
\end{align*}
Since this estimate from below does not depend on $n$, we proved that $\big\{\lVert\mathcal{G}^\lambda\xi^\beta_n\rVert\big\}_{n\+\in\+\N}$ is away from $0$ for any value of $\beta,\lambda\->\-0$.

\n Our goal is to show that, for some $\beta\->\-0$ one has
$$\lim_{n\to\infty}\breve{\mathcal{Q}}[\psi^{\lambda,\+\beta}_n]= \lim_{n\to\infty} \!\left(\!-1 +\tfrac{1}{\vphantom{\big(\big)}\lVert\mathcal{G}^\lambda\xi^\beta_n\rVert^2\!\-}\:\breve{\Phi}^\lambda[\xi^\beta_n]\-\right)=\mInfty,\qquad\forall\lambda>0.$$
Clearly, to this end it suffices to show the same statement only for $\breve{\Phi}^\lambda[\xi^\beta_n]$, whose leading term, for $n$ large, is provided by the following lemma.

\begin{lemma}
    Let $\breve{\Phi}^\lambda$ be the hermitian quadratic form given by~\eqref{phi4Definition} and $\xi^\beta_n$ be the sequence defined by~\eqref{trialChargeSequence}.
    Then, one has for any fixed $\lambda,\beta \->\- 0$
    \begin{gather*}
    \breve{\Phi}^\lambda[\xi^\beta_n]=n^2\!\left(\Phi^0_{\mathrm{diag}}[\xi]+\Phi^0_{\mathrm{off;}\+0}[\xi]+\Phi^0_{\mathrm{off;}\+1}[\xi]+\breve{\Phi}_0[\xi]\right)\!+\oSmall{n^2}\-,\qquad n\to\infty.
    \end{gather*}
    \begin{proof}
        We start by considering the component $\breve{\Phi}_{\mathrm{reg}}[\xi^\beta_n]$.
        By exploiting~\eqref{positiveBoundedCondition}
        \begin{equation*}
            \big(\Phi^{(2)}_{\mathrm{reg}}+\Phi^{(3)}_{\mathrm{reg}}\big)[\xi^\beta_n]\leq (N\!-\-2)\gamma\-\!\integrate[\R^{3(N-1)}]{\frac{\abs{\xi^\beta_n(\vec{x},\vec{x}'\-,\vec{X})}^2}{\abs{\vec{x}\--\-\vec{x}'}} ;\mspace{-33mu}d\vec{x}d\vec{x}'\-d\vec{X}}+\-\left[\abs{\alpha_0}+\tfrac{(N\--2)\+\gamma}{b}\right]\!\lVert\xi^\beta_n\rVert_{\hilbert*_{N-1}}^2.
        \end{equation*}
        First of all, making use of~\eqref{highPowerTriangular} one has
        \begin{align*}
            \lVert\xi_n^\beta\rVert_{\hilbert*_{N-1}}^2 &\leq n^4\nqquad\sum_{\mathcal{P}_{N\--2}\+\ni\+\nu\+=\+\{k,\,\ell\}}\!\integrate[\R^{3(N-1)}]{\abs{g_\beta(n\+\abs{\vec{x}_k\!-\vec{x}_\ell})\+f(\vec{x},\tfrac{\vec{x}_k+\+\vec{x}_\ell}{2},\vec{X}_{\-\nu})}^2;\mspace{-33mu}d\vec{x}d\vec{x}_k d\vec{x}_\ell\+ d\vec{X}_{\-\nu}}\\
            &= n\,\tfrac{(N\--2)(N\--3)}{2} \!\-\integrate[\R^{3(N-1)}]{\abs{g_\beta(\vec{r})\+f(\vec{x},\vec{R},\vec{X})}^2;\mspace{-33mu}d\vec{x}d\vec{r}d\vec{R}\+d\vec{X}}.
        \end{align*}

        \vspace{-0.1cm}
        
        \n This proves
        \begin{equation}
            \lVert\xi_n^\beta\rVert_{\hilbert*_{N-1}}^2 \!= \oBig{n},\qquad n\longrightarrow\pInfty,\, \forall \beta>0.
        \end{equation}
        Then, considering the following decomposition
        \begin{equation*}
            \begin{split}
                \xi^\beta_n(\vec{x},\vec{x}_1,\ldots,\vec{x}_{N-2})=\tfrac{\sqrt{2\+}\+n^2}{\!\-\sqrt{(N\--2)(N\--3)\+}\+}\-\Bigg[&\sum_{\ell\+=\+1}^{N\--3} g_\beta\big(n(\vec{x}_\ell\!-\vec{x}_{N-2})\big)\+f\big(\vec{x},\tfrac{\vec{x}_\ell+\+\vec{x}_{N-2}}{2},\vec{x}_1,\ldots\check{\vec{x}}_\ell\ldots,\vec{x}_{N-3})\big)\++\\[-7.5pt]
                &+\nqquad\sum_{\mathcal{P}_{N-3}\+\ni\+\nu\+=\+\{k,\,\ell\}}\nqquad g_\beta\big(n(\vec{x}_k\!-\vec{x}_\ell)\big)\+f\big(\vec{x},\tfrac{\vec{x}_k+\+\vec{x}_\ell}{2},\vec{x}_1,\ldots\check{\vec{x}}_\nu\ldots,\vec{x}_{N-2}\big)\-\Bigg]\-,
            \end{split}
        \end{equation*}

        \vspace{-0.25cm}
        
        \n one can obtain
        \begin{equation*}
            \begin{split}
            \integrate[\R^{3(N-1)}]{\frac{\abs{\xi^\beta_n(\vec{x},\vec{x}'\-,\vec{X})}^2\!\-}{\abs{\vec{x}\--\-\vec{x}'}} ;\mspace{-33mu}d\vec{x}d\vec{x}'\-d\vec{X}}\leq &\,\tfrac{4\,n^4}{N\--2}\!\integrate[\R^{3(N-1)}]{\frac{\abs{g_\beta\big(n(\vec{x}'\!\--\vec{x}'')\big)\+f(\vec{x},\tfrac{\vec{x}'\-+\+\vec{x}''\!}{2},\vec{X})}^2\!\-}{\abs{\vec{x}\--\vec{x}'}};\mspace{-33mu}d\vec{x}d\vec{x}'\-d\vec{x}''\!d\vec{X}}\,+\\
            &+\tfrac{2\+(N\--4)\+n^4\!\-}{N\--2}\!\-\integrate[\R^{3(N-1)}]{\frac{\abs{g_\beta\big(n(\vec{x}''\!\--\vec{x}''')\big)\+f\big(\vec{x},\tfrac{\vec{x}''\!+\+\vec{x}'''\!\-}{2},\vec{x}'\-,\vec{X}\big)}^2\!\-}{\abs{\vec{x}\--\vec{x}'}};\mspace{-33mu}d\vec{x}d\vec{x}'\-d\vec{x}''\!d\vec{x}'''\!d\vec{X}}.
            \end{split}
        \end{equation*}
        By the means of proper changes of variables the previous inequality reads
        \begin{equation*}
            \begin{split}
            \integrate[\R^{3(N-1)}]{\frac{\abs{\xi^\beta_n(\vec{x},\vec{x}'\-,\vec{X})}^2\!\-}{\abs{\vec{x}\--\-\vec{x}'}} ;\mspace{-33mu}d\vec{x}d\vec{x}'\-d\vec{X}}\leq &\,\tfrac{4\,n^4}{N\--2}\!\integrate[\R^{3(N-1)}]{\frac{\abs{g_\beta\big(n\+\vec{y}''\big)\+f(\vec{y}+\vec{y}'\-+\tfrac{\vec{y}''\!}{2},\vec{y}'\-,\vec{Y})}^2\!\-}{y};\mspace{-33mu}d\vec{y}d\vec{y}'\-d\vec{y}''\!d\vec{Y}}\,+\\
            &+\tfrac{2\+(N\--4)\+n^4\!\-}{N\--2}\!\-\integrate[\R^{3(N-1)}]{\frac{\abs{g_\beta\big(n\+\vec{y}''\big)\+f\big(\!-\!\vec{y}+\tfrac{\vec{y}'\-}{2},\vec{y}'''\!,-\vec{y}-\tfrac{\vec{y}'\-}{2},\vec{Y}\big)}^2\!\-}{y'};\mspace{-33mu}d\vec{y}d\vec{y}'\-d\vec{y}''\!d\vec{y}'''\!d\vec{Y}}.
            \end{split}
        \end{equation*}
        Now we can exploit the Hardy-Rellich inequality~\eqref{Hardy-Rellich} to get
        \begin{align*}
            \integrate[\R^{3(N-1)}]{\frac{\abs{\xi^\beta_n(\vec{x},\vec{x}'\-,\vec{X})}^2\!\-}{\abs{\vec{x}\--\-\vec{x}'}} ;\mspace{-33mu}d\vec{x}d\vec{x}'\-d\vec{X}}\leq &\,\tfrac{1}{n^2}\,\tfrac{2\+\pi}{N\--2}\!\integrate[\R^{3(N-1)}]{q\:\abs{\FT{g}_\beta\big(\tfrac{\vec{q}''\-}{n}-\tfrac{\vec{q}}{2\+n}\big)\+\FT{f}(\vec{q},\vec{q}'\!\--\-\vec{q},\vec{Q})}^2;\mspace{-33mu}d\vec{q}d\vec{q}'\-d\vec{q}''\!d\vec{Q}}+\\
            &+\tfrac{1}{n^2}\,\tfrac{(N\--4)\+\pi}{N\--2}\!\-\integrate[\R^{3(N-1)}]{q'\-\:\abs{\FT{g}_\beta\big(\tfrac{\vec{q}''\!}{n}\big)\+\FT{f}\big(\vec{q}'\-\!-\tfrac{\vec{q}}{2},\vec{q}'''\!,-\vec{q}'\-\!-\tfrac{\vec{q}}{2},\vec{Q}\big)}^2;\mspace{-33mu}d\vec{q}d\vec{q}'\-d\vec{q}''\!d\vec{q}'''\!d\vec{Q}}\\
            =&\,n\:\tfrac{2\+\pi}{N\--2}\!\integrate[\R^{3(N-1)}]{p\:\abs{\FT{g}_\beta\big(\vec{p}''\big)\+\FT{f}(\vec{p},\vec{p}',\vec{P})}^2;\mspace{-33mu}d\vec{p}\+d\vec{p}'\-d\vec{p}''\!d\mspace{-0.75mu}\vec{P}}+\\
            &+n\:\tfrac{(N\--4)\+\pi}{N\--2}\!\-\integrate[\R^{3(N-1)}]{\tfrac{\abs{\vec{p}\+-\+\vec{p}'''\-}}{2}\,\abs{\FT{g}_\beta\big(\vec{p}''\big)\+\FT{f}\big(\vec{p},\vec{p}'\-,\vec{p}'''\!,\vec{P}\big)}^2;\mspace{-33mu}d\vec{p}\+d\vec{p}'\-d\vec{p}''\!d\vec{p}'''\!d\mspace{-0.75mu}\vec{P}}.
        \end{align*}
        This yields
        \begin{equation}
            \breve{\Phi}_{\mathrm{reg}}[\xi^\beta_n]=\oBig{n}\-,\qquad n\longrightarrow\pInfty,\,\forall\beta>0.
        \end{equation}
        Next, we take account of $\Phi^\lambda_{\mathrm{diag}}[\xi^\beta_n]$.
        By assumption, one has
        \begin{align*}
            &\Phi^\lambda_{\mathrm{diag}}[\xi^\beta_n]\-=\tfrac{1}{\!\sqrt{2\,}\,}\!\-\integrate[\R^{3(N-1)}]{\sqrt{\tfrac{1}{2}\+p^2\-+\-P^2\-+\lambda\,}\,\abs{\FT{\xi}^\beta_n(\vec{p},\vec{P})}^2;\mspace{-33mu}d\vec{p}\+d\mspace{-0.75mu}\vec{P}}\\[-5pt]
            &\pushright{=\!\tfrac{1}{\!\sqrt{2\+}\+n^2\!}\,\tfrac{2}{\!(N\--2)(N\--3)\!}\-\integrate[\R^{3(N-1)}]{\!\sqrt{\-\tfrac{1}{2}\+p^2\!+\-p_1^2\!+\-\ldots\-+\-p_{N-2}^2\!+\-\lambda\mspace{1mu}}\mspace{1mu}\Bigg[\mspace{-0.75mu}\sum_{\substack{\nu\+\in\+\mathcal{P}_{N\--2}\+:\\ \nu\+=\+\{k,\,\ell\}}}\mspace{-10.5mu}\abs{\FT{g}_\beta\mspace{-0.75mu}\big(\tfrac{\vec{p}_k\mspace{-0.75mu}-\+\vec{p}_\ell}{2\+n}\-\big)\FT{f}(\vec{p},\vec{p}_k\!\-+\-\vec{p}_\ell,\-\vec{P}_{\!\nu}\-)}^2;\mspace{-36mu}d\vec{p}\+d\vec{p}_1\!\cdots d\vec{p}_{N-2}}+}\\[-7.5pt]
            &\pushright{+\nqquad\sum_{\substack{\mu,\,\nu\+\in\+\mathcal{P}_{N\--2}\+: \\\mu\+=\+\{i,\,j\},\:\nu\+=\+\{k,\,\ell\},\\\mu\+\neq\+\nu}}\nqquad\FT{g}_\beta\mspace{-0.75mu}\big(\tfrac{\vec{p}_k\mspace{-0.75mu}-\+\vec{p}_\ell}{2\+n}\-\big)\FT{f}(\vec{p},\vec{p}_k\!\-+\-\vec{p}_\ell,\-\vec{P}_{\!\nu}\-)\,\FT{g}_\beta\mspace{-0.75mu}\big(\tfrac{\vec{p}_i-\+\vec{p}_j}{2\+n}\-\big)\FT{f}(\vec{p},\vec{p}_i\!+\-\vec{p}_j,\-\vec{P}_{\!\mu})\Bigg]}.
        \end{align*}


        \n Then, exploiting the symmetry of the charge
        \begin{align*}
            \Phi^\lambda_{\mathrm{diag}}&[\xi^\beta_n]\-=\\
            =&\,\tfrac{1}{\!\sqrt{2\,}\,n^2}\!\-\integrate[\R^{3(N-1)}]{\sqrt{\tfrac{1}{2}\+p^2\!+\-p_1^2\-+\ldots+\-p_{N-2}^2\-+\-\lambda\,}\,\abs{\FT{g}_\beta(\tfrac{\vec{p}_1-\+\vec{p}_2}{2\+n})\+ \FT{f}(\vec{p}, \vec{p}_1\!+\vec{p}_2,\vec{p}_3,\ldots,\vec{p}_{N-2})}^2;\mspace{-33mu}d\vec{p}\+d\vec{p}_1\-\cdots d\vec{p}_{N-2}}+\\
            &+\tfrac{1}{\!\sqrt{2\,}\,n^2}\!\-\integrate[\R^{3(N-1)}]{\sqrt{\tfrac{1}{2}\+p^2\!+\-p_1^2\-+\ldots+\-p_{N-2}^2\-+\-\lambda\,}\,\FT{g}_\beta(\tfrac{\vec{p}_1-\+\vec{p}_2}{2\+n})\+\FT{f}(\vec{p},\vec{p}_1\!+\vec{p}_2,\vec{p}_3,\ldots,\vec{p}_{N-2}); \mspace{-33mu}d\vec{p}\+d\vec{p}_1\-\cdots d\vec{p}_{N-2}}\times\\
            &\mspace{213mu}\times \!\left[2\+(N\!-\-4)\,\FT{g}_\beta(\tfrac{\vec{p}_1-\+\vec{p}_3}{2\+n})\+\FT{f}(\vec{p},\vec{p}_1\!+\-\vec{p}_3,\vec{p}_2,\vec{p}_4,\ldots,\vec{p}_{N-2})\,+\right.\\
            &\mspace{234mu}+\!\left.\tfrac{(N\--4)(N\--5)}{2}\FT{g}_\beta(\tfrac{\vec{p}_3-\+\vec{p}_4}{2\+n})\+\FT{f}(\vec{p},\vec{p}_3\mspace{-2.25mu}+\-\vec{p}_4,\vec{p}_1,\vec{p}_2,\vec{p}_5,\ldots,\vec{p}_{N-2})\right]\!.
        \end{align*}
        We define the quantities $\Phi^{\lambda,\+\beta}_{\mathrm{diag};\+\iota}(n)$ as the contributions of the previous identity where $\iota\in\{0,1,2\}$ represents the number of elements in common between the couples $\mu,\nu$ involved in the summation coming from the squaring of the trial charge.
        By definition, there holds $$\Phi^\lambda_{\mathrm{diag}}[\xi^\beta_n]=\Phi^{\lambda,\+\beta}_{\mathrm{diag};\+0}(n)+\Phi^{\lambda,\+\beta}_{\mathrm{diag};\+1}(n)+\Phi^{\lambda,\+\beta}_{\mathrm{diag};\+2}(n).$$
        We claim that for all positive $\beta$ and $\lambda$
        \begin{equation}\label{diagLeadingOrder}
            \Phi^\lambda_{\mathrm{diag}}[\xi^\beta_n]=n^2\!\-\integrate[\R^{3(N-1)}]{q\,\abs{\FT{g}_\beta(\vec{q})\+\FT{f}(\vec{p},\vec{P})}^2;\mspace{-33mu}d\vec{q}d\vec{p}\+d\mspace{-0.75mu}\vec{P}}+\oSmall{n^2}\-,\qquad n\longrightarrow\pInfty.
        \end{equation}
        In order to show this, we first focus on studying $\Phi^{\lambda,\+\beta}_{\mathrm{diag};\+2}(n)$.
        By a simple substitution, one has
        \begin{align*}
        \Phi^{\lambda,\+\beta}_{\mathrm{diag};\+2}(n)=&\,\tfrac{1}{\!\sqrt{2\,}\,n^2}\!\-\integrate[\R^{3(N-1)}]{\sqrt{\tfrac{1}{2}\+p^2\!+\-2\+q_1^2\-+\tfrac{1}{2}\+q_2^2\-+\-P^2\-+\-\lambda\,}\,\abs{\FT{g}_\beta(\tfrac{\vec{q}_1}{n})\+ \FT{f}(\vec{p}, \vec{q}_2,\vec{P})}^2;\mspace{-33mu}d\vec{p}\+d\vec{q}_1d\vec{q}_2d\mspace{-0.75mu}\vec{P}}\\
        =&\,n^2\!\-\integrate[\R^{3(N-1)}]{\sqrt{q_1^2\-+\tfrac{1}{4\+n^2}\+(p^2\!+q_2^2)\-+\tfrac{1}{2\+n^2}(\-P^2\-+\-\lambda)\+}\,\abs{\FT{g}_\beta(\vec{q}_1)\+ \FT{f}(\vec{p}, \vec{q}_2,\vec{P})}^2;\mspace{-33mu}d\vec{p}\+d\vec{q}_1d\vec{q}_2d\mspace{-0.75mu}\vec{P}}\\
        =&\,n^2\!\-\integrate[\R^{3(N-1)}]{q\,\abs{\FT{g}_\beta(\vec{q})\+\FT{f}(\vec{p},\vec{p}_1,\ldots,\vec{p}_{N-3})}^2;\mspace{-33mu}d\vec{q}d\vec{p}\+d\vec{p}_1\-\cdots\vec{p}_{N-3}}+\\
        &+n^2\!\-\integrate[\R^{3(N-1)}]{\!\left(\sqrt{q^2\!+\tfrac{p^2+\+p{'\+}^2+\+2P^2+\+2\+\lambda}{4n^2}}-q\right)\!\abs{\FT{g}_\beta(\vec{q})\+\FT{f}(\vec{p},\vec{p}'\-,\vec{P})}^2;\mspace{-33mu}d\vec{q}d\vec{p}\+d\vec{p}'\-d\mspace{-0.75mu}\vec{P}}.
        \end{align*}
        Because of the elementary inequality $\sqrt{a^2+b^2}-\abs{a}\leq\abs{b}$, dominated convergence theorem implies
        $$\lim_{n\to\infty} n\!\-\integrate[\R^{3(N-1)}]{\!\left(\sqrt{q^2\!+\tfrac{p^2+\+p{'\+}^2+\+2P^2+\+2\+\lambda}{4n^2}}-q\right)\!\abs{\FT{g}_\beta(\vec{q})\+\FT{f}(\vec{p},\vec{p}'\-,\vec{P})}^2;\mspace{-33mu}d\vec{q}d\vec{p}\+d\vec{p}'\-d\mspace{-0.75mu}\vec{P}}=0.$$
        Hence, we have just proved
        \begin{subequations}
        \begin{equation}\label{diagLeadingOrder2}
            \Phi^{\lambda,\+\beta}_{\mathrm{diag};\+2}(n)= n^2\!\-\integrate[\R^{3(N-1)}]{q\,\abs{\FT{g}_\beta(\vec{q})\+\FT{f}(\vec{p},\vec{P})}^2;\mspace{-33mu}d\vec{q}d\vec{p}\+d\mspace{-0.75mu}\vec{P}}+\oSmall{n}\-,\qquad n\longrightarrow \pInfty.
        \end{equation}
        Then, taking into account $\Phi^{\lambda,\+\beta}_{\mathrm{diag};1}(n)$, we perform the following change of variables
        $$\left[\!\!\begin{array}{c}
            \vec{p}_1\\
            \vec{p}_2\\
            \vec{p}_3
        \end{array}\!\!\right]\!=\!\left[\begin{array}{l c r}
            \;\:\tfrac{3}{5} & \;\:\tfrac{3}{5} & \;\tfrac{1}{4}\\
            \!\--\tfrac{7}{5} & \;\:\tfrac{3}{5} & \;\tfrac{1}{4}\\
            \;\:\tfrac{3}{5} & \!\--\tfrac{7}{5} & \;\tfrac{1}{4}
        \end{array}\right]\!\left[\!\!\begin{array}{c}
            \vec{q}_1\\
            \vec{q}_2\\
            \vec{q}_3
        \end{array}\!\!\right]\qquad\iff\qquad\left[\!\!\begin{array}{c}
            \vec{q}_1\\
            \vec{q}_2\\
            \vec{q}_3
        \end{array}\!\!\right]\!=\!\left[\begin{array}{l c r}
            \tfrac{1}{2} & \--\tfrac{1}{2} & 0\\
            \tfrac{1}{2} & \;\;\:0 & -\tfrac{1}{2}\\
            \tfrac{8}{5} & \;\;\:\tfrac{6}{5} & \tfrac{6}{5}
        \end{array}\right]\!\left[\!\!\begin{array}{c}
            \vec{p}_1\\
            \vec{p}_2\\
            \vec{p}_3
        \end{array}\!\!\right]$$
        so that, with straightforward manipulations we get 
        \begin{align*}
        \Phi^{\lambda,\+\beta}_{\-\mathrm{diag};\+1}\-(\mspace{-0.75mu}n\mspace{-0.75mu})\mspace{-2.25mu}&=\mspace{-2.25mu}\tfrac{\sqrt{2\mspace{0.75mu}}\mspace{0.75mu}(N\--4)\!}{n^2}\!\!\integrate[\R^{3(N-1)}]{\!\sqrt{\-\tfrac{1}{2}\+p^2\!+\tfrac{67}{25}(q_1^2\-+\-q_2^2)\-+\tfrac{18}{25}\+\vec{q}_1\mspace{-3.75mu}\cdot\mspace{-0.75mu}\vec{q}_2\-+\-\tfrac{3}{16}\+q_3^2\--\tfrac{\vec{q}_3}{10}\mspace{-2.25mu}\cdot\-(\vec{\vec{q}_1\mspace{-3.75mu}+\mspace{-2.25mu}\vec{q}_2})\-+\!P^2\!+\-\lambda\mspace{0.75mu}};\mspace{-36mu}d\vec{p}\+d\vec{q}_1d\vec{q}_2d\vec{q}_3d\mspace{-0.75mu}\vec{P}}\mspace{0.75mu}\times\\
        &\pushright{\times\+ \tfrac{\beta \, n^{4-\beta/2}\!}{4\pi^{3N/2-2}}\,q_1^{\frac{\beta}{4}-2}q_2^{\frac{\beta}{4}-2}\,\Char{\big(\tfrac{n}{4^{1/\beta}},\,n\big)}\!(q_1)\,\Char{\big(\tfrac{n}{4^{1/\beta}},\,n\big)}\!(q_2)\:e^{-p^2-\+\frac{5}{16}\+q_3^2\+-\+\frac{11}{5}(q_1^2+\,q_2^2)\++\frac{18}{5}\+\vec{q}_1\cdot\,\vec{q}_2-P^2}}\\
        &\leq \mspace{-2.25mu}\tfrac{\sqrt{2\mspace{0.75mu}}\mspace{0.75mu}\beta\+(N\--4)\+n^{2-\beta/2}\!}{4\pi^{3N/2-2}}\!\!\integrate[\R^{3(N-1)}]{\!\sqrt{\-\tfrac{1}{2}\+p^2\!+\tfrac{309}{100}(q_1^2\-+\-q_2^2)\-+\-\tfrac{23}{80}\+q_3^2+\!P^2\!+\-\lambda\+};\mspace{-36mu}d\vec{p}\+d\vec{q}_1d\vec{q}_2d\vec{q}_3d\mspace{-0.75mu}\vec{P}}\+\times\\
        &\pushright{q_1^{\frac{\beta}{4}-2}q_2^{\frac{\beta}{4}-2}\,e^{-p^2-\+\frac{2}{5}(q_1^2+\,q_2^2)-\+\frac{5}{16}\+q_3^2\+-P^2}}.
        \end{align*}
        Thus,
        \begin{equation}\label{diagLeadingOrder1}
            \Phi^{\lambda,\+\beta}_{\mathrm{diag};\+1}(n)=\oBig{n^{2-\frac{\beta}{2}}}\-,\qquad n\longrightarrow\pInfty.
        \end{equation}
        Similarly, taking account of $\Phi^{\lambda,\+\beta}_{\mathrm{diag};\+0}(n)$, we consider the following change of variables
        $$\left[\!\!\begin{array}{c}
            \vec{p}_1\\
            \vec{p}_2\\
            \vec{p}_3\\
            \vec{p}_4
        \end{array}\!\!\right]\!=\!\left[\begin{array}{l c c r}
            \;\:1 & \;\:0 & \;\tfrac{1}{2} & \;0\\
            \!\--1 & \;\:0 & \;\tfrac{1}{2} & 0\\
            \;\:0 & \;\:1 & \;0 & \tfrac{1}{2} \\
            \;\:0 & \!\--1 & \;0 & \tfrac{1}{2}
        \end{array}\right]\!\left[\!\!\begin{array}{c}
            \vec{q}_1\\
            \vec{q}_2\\
            \vec{q}_3\\
            \vec{q}_4
        \end{array}\!\!\right]\qquad\iff\qquad\left[\!\!\begin{array}{c}
            \vec{q}_1\\
            \vec{q}_2\\
            \vec{q}_3\\
            \vec{q}_4
        \end{array}\!\!\right]\!=\!\left[\begin{array}{l c c r}
            \tfrac{1}{2} & \!\--\tfrac{1}{2} & \;0 & \!0\\
            0 & \;\:0 & \;\tfrac{1}{2} & \!-\tfrac{1}{2}\\
            1 & \;\:1 & \;0 & \!0\\
            0 & \;\:0 & \;1 & \!1
        \end{array}\right]\!\left[\!\!\begin{array}{c}
            \vec{p}_1\\
            \vec{p}_2\\
            \vec{p}_3\\
            \vec{p}_4
        \end{array}\!\!\right]$$
        obtaining
        \begin{align*}
            \Phi^{\lambda,\+\beta}_{\-\mathrm{diag};\+0}\-(\mspace{-0.75mu}n\mspace{-0.75mu})\mspace{-2.25mu}&=\mspace{-2.25mu}\tfrac{(N\--4)(N\--5)\!}{2\+\sqrt{2\+}\+n^2}\!\!\integrate[\R^{3(N-1)}]{\!\sqrt{\-\tfrac{1}{2}\+p^2\!+2\+(q_1^2\-+\-q_2^2)\-+\tfrac{1}{2}(q_3^2\-+\-q_4^2)\-+\!P^2\!+\-\lambda\+};\mspace{-36mu}d\vec{p}\+d\vec{q}_1d\vec{q}_2d\vec{q}_3d\vec{q}_4 d\mspace{-0.75mu}\vec{P}}\+\times\\
            &\pushright{\times\+ \tfrac{\beta \, n^{4-\beta/2}\!}{4\pi^{3N/2-2}}\:q_1^{\frac{\beta}{4}-2}q_2^{\frac{\beta}{4}-2}\,\Char{\big(\tfrac{n}{4^{1/\beta}},\,n\big)}\!(q_1)\,\Char{\big(\tfrac{n}{4^{1/\beta}},\,n\big)}\!(q_2)\:e^{-p^2-\+q_1^2\+-\+q_2^2\+-\frac{3}{4}\+(q_3^2\++\+q_4^2)\+-P^2}.}
        \end{align*}
        Therefore, it is plain to see that also in this case we have
        \begin{equation}\label{diagLeadingOrder0}
            \Phi^{\lambda,\+\beta}_{\mathrm{diag};\+0}(n)=\oBig{n^{2-\frac{\beta}{2}}}\-,\qquad n\longrightarrow\pInfty.
        \end{equation}        \end{subequations}
        Considering altogether~(\ref{diagLeadingOrder2}-c), equation~\eqref{diagLeadingOrder} is proven.\newline
        Then, concerning $\Phi^\lambda_{\mathrm{off};\+1}[\xi^\beta_n]$, we can use~\eqref{off1PhiSplitted} to have the upper bound
        \begin{align*}
            \Phi^\lambda_{\mathrm{off};\+1}[\xi^\beta_n]\leq&\, -2(N\!-\!2)\!\-\integrate[\R^{3(N-1)}]{\frac{\:e^{-\sqrt{\-\frac{\lambda}{2}}\,\abs{\vec{x}-\+\vec{x}'}}\!\-}{\abs{\vec{x}-\vec{x}'}}\,\abs{\xi^\beta_n(\vec{x},\vec{x}'\-,\vec{X})}^2;\mspace{-33mu}d\vec{x}d\vec{x}'\-d\vec{X}}\++\\[-5pt]
            &+32\pi\+(N\!-\-2)\!\-\integrate[\R^{3(N-1)}]{\abs{\xi_n^\beta(\vec{x},\vec{x}'\-,\vec{X})}^2;\mspace{-33mu}d\vec{x}d\vec{x}'\-d\vec{X}}\!\!\integrate[\R^{3(N-1)}]{;\mspace{-33mu}d\vec{y}d\vec{y}'\-d\vec{Y}}\+G^\lambda\bigg(\!\!
            \begin{array}{r c c l}
                \vec{x}, & \!\!\-\vec{x}'\-, &\!\!\- \vec{x}, & \!\!\! \vec{X}\\[-2.5pt]
                \vec{y}, & \!\!\-\vec{y}, &\!\!\-\vec{y}'\-, & \!\!\!\vec{Y}
            \end{array}\!\!\-\bigg)\\[-5pt]
            =&\,2(N\!-\!2)\!\-\integrate[\R^{3(N-1)}]{\frac{\:e^{-\sqrt{\-\frac{\lambda}{2}}\,\abs{\vec{x}-\+\vec{x}'}}\!\-}{\abs{\vec{x}-\vec{x}'}}\,\abs{\xi_n^\beta(\vec{x},\vec{x}'\-,\vec{X})}^2;\mspace{-33mu}d\vec{x}d\vec{x}'\-d\vec{X}}\!,
        \end{align*}
        according to equation~\eqref{greenIntegrated}.
        Hence, one has
        $$\abs{\Phi^\lambda_{\mathrm{off};\+1}[\xi^\beta_n]}\leq 2\+(N\!-\!2)\!\-\integrate[\R^{3(N-1)}]{\frac{\abs{\xi_n^\beta(\vec{x},\vec{x}'\-,\vec{X})}^2\!\-}{\abs{\vec{x}-\vec{x}'}};\mspace{-33mu}d\vec{x}d\vec{x}'\-d\vec{X}}.$$
        This means that we can adopt the same argument used for the regularizing term in order to state
        $$\abs{\Phi^\lambda_{\mathrm{off};\+1}[\xi^\beta_n]}=\oBig{n}\-,\qquad n\longrightarrow \pInfty.$$

        Lastly, ............. $\Phi^\lambda_{\mathrm{off};\+0}[\xi^\beta_n]$
        
    \end{proof}
\end{lemma}
}

\comment{
\section{Self-Adjoint Extensions of Restrictions}\label{krein}

In this section, we summarize some known results obtained in~\cite{P1} and~\cite{P2} (see~\cite{BHS} for further information on this abstract setting).

\subsection*{A general perspective}

\n Let $\hilbert*$ be a complex Hilbert space and $\maps{A}{\dom{A}\-\subseteq\-\hilbert*;\hilbert*}$ a s.a. operator.
Endowing $\dom{A}$ with the graph norm of $A$, we treat $\dom{A}$ as a Hilbert subspace.
Then, we consider a densely defined, closed and symmetric operator $\maps{S}{\dom{S}\-\subset\-\dom{A};\hilbert*}$ such that $A \!\upharpoonright\- \dom{S}\!=\-S$.
Under these prescriptions, we shall provide the characterization of the whole set of s.a. extensions of $S$.

\n To this end, let $\maps{\tau}{\dom{A};\X*}$ be a linear and continuous operator in an auxiliary complex Hilbert space $\X*$, such that $\overline{\ran{\tau}}=\X*$ and $\ker{\tau}=\dom{S}$.

\n Then, for each $z\-\in\-\rho(A)$, we define the bounded operator $\define{\mathscr{G}(z);\adj{(\tau\resolvent{A}[\conjugate{z}])}}\!\!\in\-\bounded{\X*,\hilbert*}$.
\begin{note}\label{uniqueDecomposition}
    The operator $\mathscr{G}(z)$ fulfills the following properties.
    \begin{enumerate}[label=\roman*)]
        \item $\mathscr{G}(z)$ is injective, since
        $\ker{\mathscr{G}(z)}=\ran{\tau}^\perp\-=\-\{0\}$.
        \item Given a compact $K\-\subset\-\rho(A)$, there exists $C\->\-0$ s.t. $\norm{\mathscr{G}(z)}[\mathscr{L}(\X*,\,\hilbert*)]\!\leq C$ uniformly in $z\-\in\! K$, since $\tau$ is continuous and in general $\;\norm{\resolvent{A}[w]}[\mathscr{L}(\hilbert*,\,\dom{A})]\!=\mathrm{dist}(w,\,\spectrum{A})^{-1}\+,\quad\forall w\in\rho(A)$.
        \item For any $z\-\in\-\rho(A)$, $\,\ran{\mathscr{G}(z)\-}\cap\+\dom{A}\-=\-\{0\}\+$ owing to the density of $\ker{\tau}$ in $\dom{A}$.
        
        \n To prove this fact, consider $\psi\in\ker{\tau}$, so that one has
        \begin{equation*}
            0=\scalar{\phi}{\tau\psi}[\X*]=\scalar{\phi}{\tau\resolvent{A}[\conjugate{z}](A-\conjugate{z})\psi}[\X*]=\scalar{\mathscr{G}(z)\+\phi}{(A-\conjugate{z})\psi}[\hilbert*]\+,\qquad\forall\phi\in\X*.
        \end{equation*}
        By way of contradiction, assume that there exists $\phi\neq 0$ such that $\mathscr{G}(z)\+\phi\-\in\-\dom{A}$.
        Then,
        \begin{equation*}
            0=\scalar{(A-z)\+\mathscr{G}(z)\+\phi}{\psi}[\hilbert*]\+,\qquad\forall\psi\in\ker{\tau}.
        \end{equation*}
        This means that $(A-z)\+\mathscr{G}(z)\+\phi\-\in\-\ker{\tau}^{\perp}\!=\-\{0\}$, and therefore the only solution is $\mathscr{G}(z)\+\phi=0$, since $z\-\in\-\rho(A)$, but the injectivity of $\mathscr{G}(z)$ provides the contradiction.\newline
        Furthermore, this result can be strengthen, for instance, by considering $A>\mu\+$ for some $\mu\-\in\-\R$. In this case one can prove that
        $\,\ran{\mathscr{G}(\lambda)\-}\cap\+\mathfrak{Q}(A)\-=\-\{0\}\+$ for any $\lambda\-<\-\mu$, where $\mathfrak{Q}(A)\!\supset\-\dom{A}$ denotes the form domain of $A$.\newline
        This can be proven again by way of contradiction, by assuming the existence of $\phi\-\neq\- 0$ such that $\mathscr{G}(\lambda)\+\phi\-\in\-\mathfrak{Q}(A)$.
        Indeed, in this case one has
        $$0=\scalar{\sqrt{A-\lambda\+}\,\mathscr{G}(\lambda)\+\phi}{\sqrt{A-\lambda\+}\,\psi}[\hilbert*],\qquad\forall\psi\in\ker{\tau}.$$
        We stress that the quantity $\scalar{\varphi_1}{\varphi_2}[\mathfrak{Q}(A)]\-\vcentcolon=\-\scalar{\sqrt{A-\lambda\+}\,\varphi_1}{\sqrt{A-\lambda\+}\,\varphi_2}[\hilbert*]$ defines a scalar product in $\mathfrak{Q}(A)$, since $\lambda\-\in\-\rho(A)$ (hence $\norm{\varphi}[\mathfrak{Q}(A)]\!=0$ implies $\varphi=0$).
        Next, we observe that the density of $\ker{\tau}$ in $\dom{A}$ (and the fact that $A$ is densely defined) directly implies the density of $\ker{\tau}$ in $\mathfrak{Q}(A)$.
        Thus
        $$\scalar{\+\mathscr{G}(\lambda)\+\phi}{\psi}[\mathfrak{Q}(A)]\-=0,\qquad \forall \psi\in\ker{\tau} \quad\implies\quad \mathscr{G}(\lambda)\+\phi=0.$$
        Once again we have a contradiction because of the injectivity of $\,\mathscr{G}(\lambda)$.
    \end{enumerate}
\end{note}

\n Another important property of the operator $\mathscr{G}(z)$ is proved in the following proposition.

\begin{prop}\label{potentialProp}
Let $z,w\in\rho(A)$ and $\phi_1,\phi_2\in\X*$. Then, if $\ran{\tau}$ is dense in $\X*$
\begin{equation*}
\mathscr{G}(z)\+\phi_1\--\mathscr{G}(w)\+\phi_2\in\dom{A} \iff \phi_1\-=\phi_2\+.
\end{equation*}
\begin{proof}
    \n $\Leftarrow)\quad$ Applying the operator $\tau$ in the first resolvent identity for the operator $A$, one gets
    \begin{subequations}\label{potentialResolventIdentity}
    \begin{gather}
        (\conjugate{z}-\conjugate{w})\+\adj{\mathscr{G}(w)}\resolvent{A}[\conjugate{z}]=\adj{\mathscr{G}(z)}\!-\adj{\mathscr{G}(w)}\!,\label{potentialResolventIdentityA}\\
        (z-w)\+\resolvent{A}[z]\+\mathscr{G}(w)=\mathscr{G}(z)\--\mathscr{G}(w),\label{potentialResolventIdentityB}
    \end{gather}
    \end{subequations}
    and therefore $\ran{\mathscr{G}(z)\--\mathscr{G}(w)}\in\dom{A}$.
    
    \n $\Rightarrow)\quad$ Rewriting $\mathscr{G}(z)$ according to equation~\eqref{potentialResolventIdentityB}, one obtains
    \begin{equation*}
        \mathscr{G}(z)\+\phi_1\--\mathscr{G}(w)\+\phi_2=\mathscr{G}(w)(\phi_1\--\phi_2)+(z-w)\+\resolvent{A}[z]\+\mathscr{G}(w)\+\phi_1\+.
    \end{equation*}
    By hypothesis, the left-hand side belongs to $\dom{A}$ and the same is true for the last term in the right-hand side, whereas $\ran{\mathscr{G}(w)}\cap\dom{A}=\{0\}$ owing to point $iii)$ of Remark~\ref{uniqueDecomposition}.\newline
    These facts imply $\phi_1\--\phi_2\in\ker{\mathscr{G}(w)}=\{0\}$.
    
\end{proof}
\end{prop}
\n Lastly, we consider a continuous map $\maps{\Gamma}{\rho(A);\linear{\X*}}$ satisfying
\begin{subequations}\label{conditionsGammaTheory}
\begin{align}
    &\Gamma:z\longmapsto\maps{\Gamma(z)}{D\subseteq\X*;\X*} \;\text{is a densely defined operator},\label{baseGammaProperty}\\
    &\Gamma:z\longmapsto\Gamma(z)\-=\adj{\Gamma(\conjugate{z})},\label{adjointGammaTheory}\\
    &\Gamma(z)\--\Gamma(w)=(w\--\-z)\+\adj{\mathscr{G}(\bar{w})}\mathscr{G}(z)\-=\-(w\--\-z)\+\adj{\mathscr{G}(\bar{z})}\mathscr{G}(w),\quad\forall w,z\in\rho(A),\label{firstResolventGammaTheory}\\
    &\exists z\in\rho(A) \,:\: 0\in\rho\+(\Gamma(z)).\label{invertibilityGammaTheory}
\end{align}
\end{subequations}
\begin{note}\label{closedGamma}
    Observe that condition~\eqref{adjointGammaTheory} makes the operator $\Gamma(z)$ closed for any $z\-\in\-\rho(A)$.
    Moreover, by condition~\eqref{firstResolventGammaTheory}, one has that the domain $D$ is independent of $z\in\rho(A)$ and the map $\Gamma$ is actually locally Lipschitz continuous, \ie~\emph{(see the second point of Remark~\ref{uniqueDecomposition})}
    \begin{equation*}
        \norm{\Gamma(w)-\Gamma(z)}[\mathscr{L}(\X*)]\leq C \+|w-z|,\qquad \forall w,z\in K\-\subset\-\rho(A)
    \end{equation*}
    with some $K$ compact and $C\->\-0$.
    In particular, notice that for any $w,z \-\in\-\rho(A)$, one has that the difference $\Gamma(w)\--\Gamma(z)\-\in\-\bounded{\X*}$, so that any unbounded contribution of $\Gamma(z)$ cannot depend on $z\+$.
    Indeed, there holds $\frac{\mathrm{d}}{\mathrm{d}z}\Gamma(z)=-\+\adj{\mathscr{G}(\conjugate{z})}\+\mathscr{G}(z)\in\bounded{\X*},\quad\forall z\in\rho(A).$

    \vs
    
    Suppose $A>\mu$ and define $c(\lambda)\vcentcolon=\!\inf\limits_{\substack{\xi\+\in\+\X*\+:\\ \norm{\xi}[\X*]\+=\+1}}\!\scalar{\xi}{\Gamma(\lambda)\+\xi}[\X*]\,$ for any $\lambda<\mu$.
    Then,
    $$\frac{\scalar{\xi}{\big(\Gamma(\lambda_1)-\Gamma(\lambda_2)\big)\xi}[\X*]}{\norm{\xi}[\X*]^2}\leq C\+\abs{\lambda_1-\lambda_2},\qquad\forall\lambda_1,\lambda_2\in K\subset (-\infty, \mu)\,\text{ compact.}$$
    Hence, in particular
    $$c(\lambda_1)-\frac{\scalar{\xi}{\Gamma(\lambda_2)\+\xi}[\X*]}{\norm{\xi}[\X*]^2}\leq C\+\abs{\lambda_1-\lambda_2},\qquad\forall\lambda_1,\lambda_2\in K\subset (-\infty, \mu)\,\text{ compact.}$$
    Since the right-hand side does not depend on $\xi\in\X*$ one can optimize the inequality by taking the supremum
    obtaining that $c(\lambda)$ is locally Lipschitz continuous.
    One always has $\Gamma(\lambda)\geq c(\lambda).$
\end{note}
\n In the next proposition, we establish a sufficient condition to check the validity of properties~\eqref{conditionsGammaTheory}.
\begin{prop}\label{sufficientConditionsGamma}
    Assume $A-\mu>0$ and let $\lambda\-<\mu$ and $\maps{\Gamma_{\!\lambda}}{D\subset\X*;\X*}$ a s.a. operator satisfying
    \begin{enumerate}[label=\roman*)]
        \item $\Gamma_{\!\lambda_1}\!-\Gamma_{\!\lambda_2}=\tau(\mathscr{G}(\lambda_2)-\mathscr{G}(\lambda_1)),\qquad\forall\lambda_1,\lambda_2\-<\mu,$
        \item $\exists \maps{c}{(-\infty,\mu);\R}$ locally Lipschitz continuous s.t. $\Gamma_{\!\lambda}\geq c(\lambda), \lim\limits_{\lambda\to\mInfty}c(\lambda)=\pInfty\,$ and $\,\exists \lambda_0<\mu \,:\quad \forall\lambda<\lambda_0\quad c(\lambda)$ is positive and $\+c(\lambda_0)=0$.
    \end{enumerate}
    Then, the map
    \begin{equation*}
        z\longmapsto\define{\Gamma(z);\Gamma_{\!\lambda_0}\-+(\lambda_0-z)\+\adj{\mathscr{G}(\lambda_0)}\mathscr{G}(z)}
    \end{equation*}
    fulfills conditions~\eqref{conditionsGammaTheory}.
\begin{proof}
    Condition~\eqref{baseGammaProperty} is automatically satisfied owing to the self-adjointness of $\Gamma_{\!\lambda}\+$.
    Applying the operator $\tau$ to~\eqref{potentialResolventIdentityB}, one gets for all $z,w\in\rho(A)$
    \begin{equation}
        \tau\+[\mathscr{G}(z)-\mathscr{G}(w)]=(z-w)\+\adj{\mathscr{G}(\conjugate{z})}\mathscr{G}(w)=(z-w)\+\adj{\mathscr{G}(\conjugate{w})}\mathscr{G}(z).
    \end{equation}
    In particular, this means that \begin{gather}
        \Gamma_{\!\lambda_1}\!-\Gamma_{\!\lambda_2}\-=(\lambda_2-\lambda_1)\+\adj{\mathscr{G}(\lambda_2)}\mathscr{G}(\lambda_1)=(\lambda_2-\lambda_1)\+\adj{\mathscr{G}(\lambda_1)}\mathscr{G}(\lambda_2),\qquad\forall\lambda_1,\lambda_2<\mu,\label{evenRealDomainDoesNotDependOnLambda}\\
        \intertext{and, moreover}
        \Gamma(z)\-=\Gamma_{\!\lambda_0}\-+(\lambda_0-z)\+\adj{\mathscr{G}(\conjugate{z})}\mathscr{G}(\lambda_0).\label{alternativeDefGamma}
    \end{gather}
    Thus, \eqref{adjointGammaTheory} can be proven, since
    \begin{equation*}
        \adj{\Gamma(\conjugate{z})}\!-\Gamma_{\!\lambda_0}=(\lambda_0-z)\+\adj{\mathscr{G}(\conjugate{z})}\mathscr{G}(\lambda_0)=(\lambda_0-z)\+\adj{\mathscr{G}(\lambda_0)}\mathscr{G}(z)=\Gamma(z)-\Gamma_{\!\lambda_0}\+.
    \end{equation*}
    Next, concerning  condition~\eqref{firstResolventGammaTheory}, we use equations~\eqref{potentialResolventIdentity}, the definition of $\Gamma(z)$ and identity~\eqref{alternativeDefGamma}
    \begin{align*}
        \Gamma(z)\--\Gamma(w)=&\:\Gamma_{\!\lambda}+(\lambda-z)\+\adj{\mathscr{G}(\conjugate{z})}\mathscr{G}(\lambda)-\Gamma_{\!\lambda}+(w-\lambda)\+\adj{\mathscr{G}(\lambda)}\mathscr{G}(w)\\
        =&\,(\lambda-z)\+\adj{\mathscr{G}(\conjugate{z})}\-\big[\mathscr{G}(w)+(\lambda-w)\+\resolvent{A}[\lambda]\+\mathscr{G}(w)\big]+\\
        &+(w-\lambda)\-\big[\adj{\mathscr{G}(\conjugate{z})}+(\lambda-z)\+\adj{\mathscr{G}(\conjugate{z})}\resolvent{A}[\lambda]\big]\mathscr{G}(w)\\
        =&\,(w-z)\+\adj{\mathscr{G}(\conjugate{z})}\mathscr{G}(w).
    \end{align*}
    Finally, we show that $\Gamma(z)$ is invertible with bounded inverse for all $z\-\in\-\C\smallsetminus[\lambda_0,\pInfty)$.
    According to~\eqref{firstResolventGammaTheory}
    \begin{equation*}
        \Gamma(\conjugate{z})\--\Gamma(z)=(z-\conjugate{z})\+\adj{\mathscr{G}(z)}\mathscr{G}(z),
    \end{equation*}
    while condition~\eqref{adjointGammaTheory} implies
    \begin{equation*}
        \scalar{\xi}{[\Gamma(\conjugate{z})\--\Gamma(z)]\+\xi}[\X*]\-=-2i\,\Im{\+\scalar{\xi}{\Gamma(z)\+\xi}[\X*]}\+,\qquad\forall\xi\in D.
    \end{equation*}
    Hence, one obtains
    \begin{equation}\label{actuallyHermitianQF}
        \Im{\+\scalar{\xi}{\Gamma(z)\+\xi}[\X*]}=-\Im(z)\norm{\mathscr{G}(z)\+\xi}[\hilbert*]^2,\qquad\forall\xi\in D.
    \end{equation}
    Therefore, since $\Gamma(z)+\Gamma(\conjugate{z})$ is s.a., one gets for any $\xi\-\in\- D$
    \begin{align*}
        \norm{\xi}[\X*]^2\norm{\Gamma(z)\+\xi}[\X*]^2&\geq\abs{\scalar{\xi}{\Gamma(z)\+\xi}[\X*]}^2=|\scalar{\xi}{\tfrac{\Gamma(z)\++\+\Gamma(\conjugate{z})}{2}\+\xi}[\X*]|^2\-+|\scalar{\xi}{\tfrac{\Gamma(z)-\Gamma(\conjugate{z})}{2}\+\xi}[\X*]|^2\\
        &=|\scalar{\xi}{\tfrac{\Gamma(z)\++\+\Gamma(\conjugate{z})}{2}\+\xi}[\X*]|^2\-+(\Im\+z)^2\norm{\mathscr{G}(z)\+\xi}[\hilbert*]^4\\
        &=\-\left\lvert\scalar{\xi}{\Gamma_{\!\lambda}\+\xi}+\abs{\lambda\--\-z}^2\scalar{\mathscr{G}(z)\+\xi}{\resolvent{A}[\lambda]\+\mathscr{G}(z)\+\xi}+(\lambda\--\-z)\norm{\mathscr{G}(z)\+\xi}[\hilbert*]^2\right\rvert^2\!,\qquad \forall \lambda<\mu.
    \end{align*}
    In the last step we exploited identities~\eqref{potentialResolventIdentity} and condition~\eqref{firstResolventGammaTheory}.
    From the previous expression one can obtain the following estimate from below
    \begin{equation}\label{lowerEstimateClosedRange}
        \abs{\scalar{\xi}{\Gamma(z)\+\xi}}^2\-\geq \scalar{\xi}{\Gamma_{\!\lambda}\+\xi}^2\-+(\Im{\+z})^2\norm{\mathscr{G}(z)\+\xi}[\hilbert*]^4\-,\qquad \forall \Re{\+z}\leq\lambda.
    \end{equation}
    In other words, $\Gamma(z)$ is injective for all $z\-\in\C\smallsetminus(\lambda_0,\pInfty)$.
    In particular $\adj{\Gamma(z)}\!=\Gamma(\conjugate{z})$ is injective for all $z\-\in\C\smallsetminus(\lambda_0,\pInfty)$ as well, hence $\Gamma(z)$ has dense range in $\X*$, since $\{0\}\-=\ker{\adj{\Gamma(z)}}=\ran{\Gamma(z)}^\perp$.
    We need to prove that $\ran{\Gamma(z)}$ is also closed.
    Due to the closedness of the operator $\Gamma(z)$, according to~\cite[Chapter IV, theorem 5.2]{Kato} one has
    $$\inf_{\xi\+\in\+\X*}\left\{\abs{\scalar{\xi}{\Gamma(z)\+\xi}[\X*]} \:\Big|\; \norm{\xi}[\X*]=1\right\}\->0\qquad\iff\qquad \ran{\Gamma(z)} \,\text{ is closed}.$$
    In case the operator $\tau$ is surjective, then
    $$\inf_{\xi\+\in\+\X*}\left\{\norm{\mathscr{G}(z)\+\xi}[\hilbert*]^2 \:\Big|\; \norm{\xi}[\X*]=1\right\}=\inf_{\xi\+\in\+\X*}\left\{\scalar{\xi}{\adj{\mathscr{G}(z)}\+\mathscr{G}(z)\+\xi} \:\Big|\; \norm{\xi}[\X*]=1\right\}\->0,$$
    hence, we immediately conclude that $\Gamma(z)$ is boundedly invertible for all $z\-\in\C\smallsetminus[\lambda_0,\pInfty)$.
    On the other hand, if $\ran{\tau}$ is only dense in $\X*$, at this stage we can just infer that the operator $\Gamma(z)$ is invertible, with bounded inverse, for instance, for any $\Re{\+z}\-<\-\lambda_0\+$.
    Moreover,~\eqref{lowerEstimateClosedRange} also implies that $$c(\lambda_1)\leq\!\inf\limits_{\substack{\xi\+\in\+\X*\+:\\ \norm{\xi}[\X*]\+=\+1}}\!\scalar{\xi}{\Gamma(\lambda_1)\+\xi}[\X*]\leq \!\inf\limits_{\substack{\xi\+\in\+\X*\+:\\ \norm{\xi}[\X*]\+=\+1}}\!\scalar{\xi}{\Gamma(\lambda_2)\+\xi}[\X*],\qquad \forall\lambda_2\leq\lambda_1<\mu.$$
    In other words, the function $c$ is always bounded from above by a non-increasing function.
    However, we are going to prove that the same result holds under the assumptions of this proposition.\newline
    It is known that, given $z\-\in\-\rho(A)$ such that $\Gamma(z)$ is boundedly invertible, if there exists some $w\-\in\-\rho(A)$ so that
    $$\norm{\Gamma(w)\--\Gamma(z)}[\linear{\X*}]\norm{\Gamma(z)^{-1}}[\linear{\X*}]\!<\-1$$
    then also $\Gamma(w)$ is invertible with bounded inverse and
    $$\norm{\Gamma(w)^{-1}}[\linear{\X*}]\leq\left[\:\frac{1}{\norm{\Gamma(z)^{-1}}[\linear{\X*}]\!\!\!\!}\,-\norm{\Gamma(w)\--\Gamma(z)}[\linear{\X*}]\-\right]^{-1}\!.$$
    Hence, for any $z\-\in\-\left\{w\in\C\:\big|\;\Re{\+w}\-<\-\lambda_0, \,\Im{\+w}\neq 0\right\}$ consider for a given $\lambda\-<\-\lambda_0$
    \begin{align*}
        \norm{\Gamma(2\lambda\--\-z)\--\Gamma(z)}[\linear{\X*}]\norm{\Gamma(z)^{-1}}[\linear{\X*}]\-&= 2\+\abs{\lambda\--\-z}\norm{\adj{\mathscr{G}(2\lambda\--\-\conjugate{z})}\+\mathscr{G}(z)}[\linear{\X*}]\norm{\Gamma(z)^{-1}}[\linear{\X*}]\\
        &\leq \frac{2\norm{\tau}[\linear{\dom{A},\,\X*}]^2 \abs{\lambda\--\-z}}{(\Im{\+z})^2}\+\norm{\Gamma(z)^{-1}}[\linear{\X*}].
    \end{align*}
    According to~\eqref{lowerEstimateClosedRange} we have $\norm{\Gamma(z)^{-1}}[\linear{\X*}]\leq \frac{1}{c(\Re{\+z})}$, hence
    \begin{equation*}
        \norm{\Gamma(2\lambda\--\-z)\--\Gamma(z)}[\linear{\X*}]\norm{\Gamma(z)^{-1}}[\linear{\X*}]\-\leq \frac{2\sqrt{(\lambda\--\-\Re{\+z})^2+(\Im{\+z})^2}}{(\Im{\+z})^2\,c(\Re{\+z})}.
    \end{equation*}
    It is clear that picking $\abs{\Im{\+z}}$ large enough the right-hand side of the previous inequality is smaller than $1$.
    Thus, we found $0\in\-\rho(\Gamma(w))$ for all $w\-\in\-\rho(A)$ such that $\Re{\+w}\->2\lambda\--\-\lambda_0$ and $\abs{\Im{\+w}}$ sufficiently large.
    More precisely, the new information is that one can find for any fixed $\Re{\+w}\-\geq\-\lambda_0$ the boundedly invertibility by requiring
    $$\abs{\Im{\+w}}>\frac{\norm{\tau}[\linear{\dom{A},\,\X*}]^2}{c(2\lambda\--\Re{\+w})}\,\sqrt{2+2\sqrt{1+\tfrac{c(2\lambda-\Re{\+w})^2\!}{\norm{\tau}[\linear{\dom{A},\,\X*}]^4}\,(\Re{\+w}\--\-\lambda)^2\, }\,}=\vcentcolon L_{\lambda}(\Re{\+w})$$
    and there holds for any $\lambda\-<\-\lambda_0$
    $$\norm{\Gamma(w)^{-1}}[\linear{\X*}]\leq\frac{1}{c(2\lambda\--\Re{\+w})\--\frac{2\norm{\tau}[\linear{\dom{A},\,\X*}]^2\!}{(\Im{\+w})^2}\+\sqrt{(\Re{\+w}\--\-\lambda)^2+(\Im{\+w})^2\+}\,}=\vcentcolon \frac{1}{g_{\lambda}(\Re{\+w},\Im{\+w})}.$$
    One could argue that the function $g_{\lambda}(\Re{\+w},\+\cdot)$ is increasing and positive in the interval $(L_{\lambda}(\Re{\+w}),\pInfty)$ with \begin{align*}
        g_{\lambda}(\Re{\+w},L_{\lambda}(\Re{\+w}))=0, & & \lim\limits_{t\to\pInfty} g_{\lambda}(\Re{\+w},t)=c(2\lambda-\Re{\+w}).
    \end{align*}
    Similarly to what has been done here, one can pick $w\-\in\-\rho(A)$ such that $\Im{\+w}\->\-L_{\lambda}(\Re{\+w})$ and $\Re{\+w}\-\geq\-\lambda_0\+$.
    Then, given $0\-<\-s\-<\-\Im{\+w}$ we consider the quantity
    $$\norm{\Gamma(w\--i\+s)\--\Gamma(w)}[\linear{\X*}]\norm{\Gamma(w)^{-1}}[\linear{\X*}]\leq \frac{s\norm{\tau}[\linear{\dom{A},\,\X*}]^2}{\Im{\+w}\,(\Im{\+w}-s)\,g_{\lambda}(\Re{\+w},\Im{\+w})}$$
    which turns out to be less than $1$ whenever
    $$s<\frac{(\Im{\+w})^2\+g_{\lambda}(\Re{\+w},\Im{\+w})}{\norm{\tau}[\linear{\dom{A},\,\X*}]^2+\Im{\+w}\,g_{\lambda}(\Re{\+w},\Im{\+w})}.$$
    Hence, we proved that $\Gamma(\zeta)$ is boundedly invertible for $\zeta\-\in\-\rho(A)$ such that $\Re{\+\zeta}\-\geq\-\lambda_0$ and
    $$\Im{\+\zeta}>\frac{\Im{\+w}\norm{\tau}[\linear{\dom{A},\,\X*}]^2}{\norm{\tau}[\linear{\dom{A},\,\X*}]^2+\Im{\+w}\,g_{\lambda}(\Re{\+\zeta},\Im{\+w})}.$$
    In particular, by choosing $\Im{\+w}$ larger and larger we can notice that the lower bound for $\Im{\+\zeta}$ is smaller than $L_{\lambda}(\Re{\+\zeta})$.
    Indeed, taking account of the symmetry for conjugation, one gets the boundedly invertibility of $\Gamma(\zeta)$ for all $\Re{\+\zeta}\-\geq\-\lambda_0$ such that $\abs{\Im{\+\zeta}}\->\frac{\norm{\tau}[\linear{\dom{A},\,\X*}]^2}{c(2\lambda-\Re{\+\zeta})}$.
    In the end, by exploiting the degree of freedom of $\lambda\-<\-\lambda_0$ which can be chosen smaller and smaller, the operator $\Gamma(\zeta)$ has a bounded inverse for any $\Re{\+\zeta}\-\geq\-\lambda_0$ and $\abs{\Im{\+\zeta}}\->\-0$.
    
    \comment{However, let $w\-\in\-\C$ be such that $\Re{\+w}\-<\-\lambda_0$ and $\abs{\Im{\+w}}\->\-\epsilon\->\-0$.
    
    Notice that because of~\eqref{lowerEstimateClosedRange}
    $$\epsilon'<\,\frac{1}{\!\!\!\sup\limits_{\substack{\Re{\+w}\+<\+\lambda_0,\\ \abs{\Im{\+w}}\+>\+\epsilon}}\!\!\norm{\Gamma(w)^{-1}}[\linear{\X*}]\nquad\!\!\!}$$
    Fix $\epsilon'>0$ such that $$\epsilon'<\,\frac{1}{\!\!\!\sup\limits_{\substack{\Re{\+w}\+<\+\lambda_0,\\ \abs{\Im{\+w}}\+>\+\epsilon}}\!\!\norm{\Gamma(w)^{-1}}[\linear{\X*}]\nquad\!\!\!}$$
    
    Hence
    \begin{align*}
        \sup_{\substack{\Re{\+w}\+<\+\lambda_0,\\ \abs{\Im{\+w}}\+>\+\epsilon}} \!\!\norm{\Gamma(w\-+\-\delta)\--\Gamma(w)}[\linear{\X*}]\norm{\Gamma(w)^{-1}}[\linear{\X*}]\leq\! \sup_{\substack{\Re{\+w}\+<\+\lambda_0,\\ \abs{\Im{\+w}}\+>\+\epsilon}} \!\!\norm{\Gamma(w\-+\-\delta)\--\Gamma(w)}[\linear{\X*}]\sup_{\substack{\Re{\+w}\+<\+\lambda_0,\\ \abs{\Im{\+w}}\+>\+\epsilon}}\!\!\norm{\Gamma(w)^{-1}}[\linear{\X*}]\\[-2pt]
        =\delta\!\!\sup_{\substack{\Re{\+w}\+<\+\lambda_0,\\ \abs{\Im{\+w}}\+>\+\epsilon}}\!\!\norm{\adj{\mathscr{G}(\conjugate{w}\-+\-\delta)}\+\mathscr{G}(w)}[\linear{\X*}]\sup_{\substack{\Re{\+w}\+<\+\lambda_0,\\ \abs{\Im{\+w}}\+>\+\epsilon}}\!\!\norm{\Gamma(w)^{-1}}[\linear{\X*}]\\
        \leq\delta\norm{\tau}[\linear{\dom{A},\,\X*}]^2\!\sup_{\substack{\Re{\+w}\+<\+\lambda_0,\\ \abs{\Im{\+w}}\+>\+\epsilon}}\!\!\frac{1}{\mathrm{dist}(w,\spectrum{A}\-)\,\mathrm{dist}(\conjugate{w}\-+\-\delta,\spectrum{A}\-)}\sup_{\substack{\Re{\+w}\+<\+\lambda_0,\\ \abs{\Im{\+w}}\+>\+\epsilon}}\!\!\norm{\Gamma(w)^{-1}}[\linear{\X*}]\\
        \leq\frac{\delta\norm{\tau}[\linear{\dom{A},\,\X*}]^2}{(\epsilon')^2}\sup_{\substack{\Re{\+w}\+<\+\lambda_0,\\ \abs{\Im{\+w}}\+>\+\epsilon}}\!\!\norm{\Gamma(w)^{-1}}[\linear{\X*}]
    \end{align*}
    }
    
\end{proof}
\end{prop}
\comment{\n 
In light of Proposition~\ref{sufficientConditionsGamma} one finds out that $\Gamma(z)$ can be always parametrized as follows
$$\Gamma(z)=\Theta+\tau \-\left(\tfrac{\mathscr{G}(z_0)\++\+\mathscr{G}(\conjugate{z}_0)}{2}-\mathscr{G}(z)\right)\!,$$
for some $z_0\in\rho(A)$ and $\Theta\in\linear{\X*}$ a s.a. (and typically unbounded) operator.
}
\n One of the main results of~\cite{P1} and~\cite{P2} is that the choice of $\tau$ and $\Gamma$ identifies in an exhaustive way every s.a. extension of $S$, denoted by $A^\tau_\Gamma\+$.
Indeed, for any $z\-\in\-\rho(A)$ satisfying~\eqref{invertibilityGammaTheory}, namely such that $\Gamma(z)^{-1}\-\in\bounded{\X*}$, the quantity
\begin{equation}
    \define{R^\tau_{\Gamma}(z);\resolvent{A}[z]+\mathscr{G}(z)\Gamma(z)^{-1}\adj{\mathscr{G}(\conjugate{z})}}
\end{equation}
defines the resolvent of a s.a. operator $\define{A_\Gamma^\tau;R^\tau_\Gamma(z)^{-1}\!+\-z}$ that does not depend on $z$ and coincides with $A$ on $\ker{\tau}$.
More precisely, it is characterized by
\begin{equation}\label{saExtensionA}
    \begin{dcases}
    \dom{A^\tau_\Gamma}\vcentcolon=\left\{\phi\in\hilbert*\,\big|\;\phi=\varphi_z\-+\mathscr{G}(z)\Gamma(z)^{-1}\tau\varphi_z\+,\:\varphi_z\in\dom{A},\:z\in\rho(A) \text{ satisfying~\eqref{invertibilityGammaTheory}} \right\}\-,\\
    (A^\tau_\Gamma-z)\+\phi=(A-z)\+\varphi_z\+.
    \end{dcases}
\end{equation}
Notice that, for any fixed $z\-\in\-\rho(A)$ satisfying~\eqref{invertibilityGammaTheory}, the decomposition of an element in $\dom{A^\tau_\Gamma}$ is unique because of point $iii)$ of Remark~\ref{uniqueDecomposition}.

\n Equivalently, the operator $A^\tau_\Gamma$ can be represented in terms of a proper boundary condition.
Indeed, if we introduce the quantity $\xi\-\vcentcolon=\-\Gamma(z)^{-1}\tau\varphi_z\-\in\-\X*$, it is straightforward to see that 
\begin{equation}\label{bcHamiltonianCharacterization}
    \begin{dcases}
    \dom{A^\tau_\Gamma}=\left\{\phi\in\hilbert*\,\big|\;\phi=\varphi_z\-+\mathscr{G}(z)\+\xi,\:\Gamma(z)\+\xi=\tau\varphi_z\+,\;\varphi_z\in\dom{A},\:\xi\in D,\: z\in\rho(A) \right\}\-,\\
    A^\tau_\Gamma\+\phi=A\varphi_z+z\+\mathscr{G}(z)\+\xi.
    \end{dcases}
\end{equation}
We stress that the definition of $\xi$ does not depend on $z$.
Indeed, one can exploit two distinct decompositions of an element in $\dom{A^\tau_\Gamma}$ to obtain
$$\varphi_z\--\phi_w=\mathscr{G}(w)\+\Gamma(w)^{-1}\tau\phi_w\!-\mathscr{G}(z)\+\Gamma(z)^{-1}\tau\varphi_z\+,\qquad\forall w,z\in\rho(A) \text{ satisfying~\eqref{invertibilityGammaTheory}}.$$ 
Since the left-hand side belongs to $\dom{A}$, according to Proposition~\ref{potentialProp} one has
$$\Gamma(z)^{-1}\+\tau\varphi_z=\Gamma(w)^{-1}\+\tau\phi_w\+.$$

\n It is noteworthy evaluating the quadratic form associated to the s.a. operator $A^\tau_\Gamma\+$.
In particular, one has for any $\phi\in\dom{A^\tau_\Gamma}\subset\mathfrak{Q}(A^\tau_\Gamma)$
\begin{align}
    \scalar{\phi}{A^\tau_\Gamma\+ \phi}[\hilbert*]&=\scalar{\phi}{A\+ \varphi_z\-+z\,\mathscr{G}(z)\+\xi}[\hilbert*]\-=z\norm{\phi}[\hilbert*]^2\-+\scalar{\phi}{(A-z)\+\varphi_z}[\hilbert*]\nonumber\\
    &=z\norm{\phi}[\hilbert*]^2\-+\scalar{\varphi_z}{(A-z)\+\varphi_z}[\hilbert*]+\scalar{\mathscr{G}(z)\+\xi}{(A-z)\+\varphi_z}[\hilbert*]\nonumber\\
    &=z\norm{\phi}[\hilbert*]^2\-+\scalar{\varphi_z}{(A-z)\+\varphi_z}[\hilbert*]+\scalar{\xi}{\tau \varphi_z}[\X*].\nonumber\\
    \intertext{Therefore, the boundary condition implies}
    \scalar{\phi}{A^\tau_\Gamma\+ \phi}[\hilbert*]&=z\norm{\phi}[\hilbert*]^2\-+\scalar{\varphi_z}{(A-z)\+\varphi_z}[\hilbert*]+\scalar{\xi}{\Gamma(z)\+\xi}[\X*]\+.\label{operatorDomainQF}
\end{align}
Remarkably, since the left-hand side of the previous equation is real, one has
\begin{equation*}
    \Im(z)\norm{\varphi_z\-+\mathscr{G}(z)\+\xi}[\hilbert*]^2-\Im(z)\norm{\varphi_z}[\hilbert*]^2\-+\Im\scalar{\xi}{\Gamma(z)\+\xi}=0.
\end{equation*}
Thus, due to~\eqref{actuallyHermitianQF} one has that the decomposition of a vector $\phi=\varphi_z\-+\mathscr{G}(z)\+\xi\-\in\-\dom{A^\tau_\Gamma}$ requires
\begin{equation}
    \Im(z)=0\quad\lor\quad \varphi_z\perp \mathscr{G}(z)\+\xi.
\end{equation}
\n For the sake of completeness, we mention that
$$\mu\in\spectrum[\mathrm{p}]{A^\tau_\Gamma}\cap\rho(A)\iff 0\in\spectrum[\mathrm{p}]{\Gamma(\mu)},$$
where $\spectrum[\mathrm{p}]{\cdot}$ denotes the point spectrum and the operator $\maps{\mathscr{G}(\mu)}{\ker{\Gamma(\mu)};\ker{A^\tau_\Gamma-\mu}}$ is a bijection.

\n Lastly, we provide a revised version of~\cite[corollary 2.1]{P1}.
\begin{prop}\label{boundednessFromBelow}
    Let $A$ be a s.a. operator in $\hilbert*$ and $\maps{\Gamma}{\rho(A);\linear{\X*}}$ a continuous map satisfying~\eqref{conditionsGammaTheory}.
    Suppose that there exist $\mu\-\in\-\R$ and $\lambda_0\leq \mu$ such that
    \begin{enumerate}[label=\roman*)]
    \item $A-\mu>0\+$,
    \item    $\Gamma(\lambda)>0,\qquad\forall \lambda<\lambda_0\+$.
    \end{enumerate}
    Then, the operator $A^\tau_\Gamma$ defined by~\eqref{saExtensionA} satisfies $A^\tau_\Gamma\geq\lambda_0\+$ for any $\maps{\tau}{\dom{A};\X*}$ linear bounded operator s.t. $\overline{\ran{\tau}}=\X*$ and $\overline{\ker{\tau}}=\dom{A}$.
\begin{proof}
    Taking into account identity~\eqref{operatorDomainQF}, one has for all $\phi\in\dom{A^\tau_\Gamma}$ and $\lambda<\lambda_0$
    \begin{align*}
        \scalar{\phi}{A^\tau_\Gamma\+\phi}[\hilbert*]&=\lambda\norm{\phi}[\hilbert*]^2+\scalar{\phi_\mu}{(A-\lambda)\+\phi_\mu}[\hilbert*]+\scalar{\xi}{\Gamma(\lambda)\+\xi}[\X*]\\
        &>\lambda\norm{\phi}[\hilbert*]^2+\scalar{\phi_\mu}{(A-\mu)\+\phi_\mu}[\hilbert*]>\lambda\norm{\phi}[\hilbert*]^2.
    \end{align*}
\end{proof}
\end{prop}

\subsection*{Application to our Setting}

\n In the following, we establish the connection of the abstract theory with our problem.
The Hilbert space $\hilbert*$ is given by $\hilbert_N\-=\-\LpS{2}[\R^{3N}]$, while the s.a. operator $A$ plays the role of the free Hamiltonian $\mathcal{H}_0$ and $S$ corresponds to $\dot{\mathcal{H}}_0\+$.
Next, define the isomorphic spaces
\begin{subequations}\label{auxiliaryHilberts}
\begin{gather}
    \define{\X_\sigma;\Lp{2}[\R^3,d\vec{x}]\-\otimes\-\LpS{2}[\R^{3(N-2)}\-,\+d\vec{x}_1\cdots d\check{\vec{x}}_\sigma\cdots d\vec{x}_N]},\qquad \sigma\in\mathcal{P}_N\label{auxiliaryHilbertSigma}\\
    \define{\X;\Lp{2}[\R^3]\-\otimes\-\LpS{2}[\R^{3(N-2)}]}\label{auxiliaryHilbertGeneric}
\end{gather}
\end{subequations}
and let $\mathcal{C}_\sigma$ be the unitary transformation sending $\X$ to $\X_\sigma$.
We stress that the space $\X$ shall be the analogue of the auxiliary Hilbert space $\X*$.
Therefore, in order to provide the construction of the operator $\tau$, we need to introduce $\maps{T_\sigma}{\dom{\mathcal{H}_0};\X_\sigma}$ the linear bounded operators whose action in $\mathcal{S}(\R^{3N})$ is explicitly given by
\begin{equation}\label{singleTraceOperator}
    T_{\sigma\+=\+\{i,\,j\}}: f(\vec{x}_1,\ldots,\vec{x}_N)\longmapsto 8\pi f(\vec{x}_1,\ldots,\vec{x}_N)|_{\vec{x}_i=\,\vec{x}_j=\,\vec{x}}.
\end{equation}
In our framework, $\tau$ shall be therefore represented by the trace operator $T\-\in\bounded{\dom{\mathcal{H}_0},\X}$ as follows
\begin{equation}\label{traceOperator}
    T:\psi\longmapsto \sum_{\sigma\+\in\+\mathcal{P}_N}\! \adj{\mathcal{C}}_\sigma\+T_\sigma\+\psi=\tfrac{N(N-1)}{2}\,\adj{\mathcal{C}}_{\nu}\+T_{\nu}\+\psi,\qquad\forall\nu\in\mathcal{P}_N.
\end{equation}
Indeed, notice that $T$ satisfies the required properties: it is bounded since the mappings $f\longmapsto f|_{\pi_\sigma}$ are continuous between $H^{\frac{3}{2}+\+s}(\R^{3N})$ and $H^s(\R^{3(N-1)})$ for any $s\!>\!0$, then $\ran{T}$ is dense in $\X$ and lastly  $\ker{T}\!=\!\dom{\dot{\mathcal{H}}_0}$, by construction.\newline
We stress that (see Proposition~\ref{weakLimitTraceOperatorPro}) if there exists $s\->\-0$ such that $f\!\in\!\hilbert_N\cap H^{\frac{3}{2}+s}(\R^{3N})$, then the operator $T_\sigma$ given by~\eqref{singleTraceOperator} satisfies
\begin{equation}\label{weakTraceId}
    (T_\sigma f)(\vec{x},\vec{X}_{\-\sigma})= 8\pi \wlim{r\to 0^+} (U_\sigma f)(\vec{r},\vec{x},\vec{X}_{\-\sigma}),
\end{equation}
where $U_\sigma, \sigma=\{i,j\}\-\in\-\mathcal{P}_N$ (with $i<j$ without loss of generality) is the following unitary operator
\begin{equation}
\label{unitaryCoordinateAppendix}\begin{split}
    &\maps{U_\sigma}{\Lp{2}[\R^{3N}];\Lp{2}[\R^6,d\vec{r}d\vec{x}]\-\otimes\-\Lp{2}[\R^{3(N-2)},d\vec{X}_{\-\sigma}]},\\[-2pt]
    &\qquad(U_\sigma\+f)(\vec{r},\vec{x},\vec{X}_{\-\sigma})=f\Big(\vec{x}_1,\ldots,\underset{i\text{-th}}{\vec{x}+\tfrac{\vec{r}}{2}}\+,\ldots,\underset{j\text{-th}}{\+\vec{x}-\tfrac{\vec{r}}{2}}\+,\ldots,\vec{x}_N\!\Big).
    \end{split}
\end{equation}
The reason behind the constant $8\pi$ in the definition of $T_\sigma$ shall be clarified by Remark~\ref{tuningConstant}.
\begin{note}
    One can argue that the operator defined by equation~\eqref{potentialDef} can be equivalently rewritten as
    \begin{equation}\label{abstractPotential}
        \mathcal{G}^\lambda\!=\-\adj{(T\+\resolvent{\mathcal{H}_0}(-\lambda))}\-=\mathscr{G}(-\lambda),\qquad\lambda>0.
    \end{equation}
    Indeed, by definition, one has
    \begin{gather*}
        \scalar{\psi}{\mathscr{G}(-\lambda)\+\xi}[\hilbert_N]=\sum_{\sigma\+\in\+\mathcal{P}_N}\-\scalar{T_\sigma\resolvent{\mathcal{H}_0}(-\lambda)\+\psi}{\mathcal{C}_\sigma\+\xi}[\X_\sigma]\\[-5pt]
        =8\pi\nquad\sum_{\mathcal{P}_N\+\ni\,\sigma\+=\+\{i,\,j\}}\-\integrate[\R^{3(N-1)}]{\xi(\vec{y},\vec{Y}_{\!\!\sigma});\mspace{-33mu}d\vec{y}d\vec{Y}_{\!\!\sigma}}\!\-\integrate[\R^{3N}]{\conjugate*{\psi(\vec{x}_i,\vec{x}_j,\vec{X}_{\-\sigma})};\mspace{-10mu} d\vec{x}_i\+d\vec{x}_j\+ d\vec{X}_{\-\sigma}}\,G^\lambda\bigg(\!\!
            \begin{array}{r c l}
                \vec{x}_i\+,&\!\!\vec{x}_j\+, &\!\!\vec{X}_{\-\sigma}\\[-2.5pt]
                \vec{y},&\!\!\vec{y},&\!\!\vec{Y}_{\!\!\sigma}
            \end{array}\!\!\!\bigg).
    \end{gather*}
    Therefore, one obtains the expression given in~\eqref{potentialDef}.
\end{note}

\n Taking account of this abstract setting, it is possible to provide an alternative (yet equivalent) construction of the lower semi-bounded Hamiltonian $\mathcal{H}$.
More precisely, we just need the definition of the missing proper continuous map $\maps{\tilde{\Gamma}}{\C\smallsetminus\Rplus;\linear{\X}}$ that fulfills properties~\eqref{conditionsGammaTheory} and encodes the singular boundary conditions at the coincidence hyperplanes described by~\eqref{mfBC}.\newline
Such Hamiltonian $\mathcal{H}$, being a singular perturbation of $\mathcal{H}_0$ supported on $\pi$, coincides by construction with $\mathcal{H}_0$ on $\hilbert_N\cap H_0^2(\R^{3N}\setminus \pi)$ and, according to~\eqref{bcHamiltonianCharacterization}, any element of its domain $\psi$ can be decomposed as $\psi=\varphi_\lambda\-+\mathcal{G}^\lambda\xi$ with $\varphi_\lambda\!\in\-\dom{\mathcal{H}_0}$ and $\xi$ in the domain of $\tilde{\Gamma}(-\lambda)$ satisfying the boundary condition $\tilde{\Gamma}(-\lambda)\+\xi=T\varphi_\lambda\+$.
In order to individuate the specific map $\tilde{\Gamma}$ describing the desired behavior introduced in Section~\ref{intro}, we compute the action of the quadratic form associated to $\tilde{\Gamma}(-\lambda)$ by imposing on $\psi$ the boundary condition~\eqref{mfBC} in the weak topology, \ie
\begin{equation}
    \wlim{r\to 0^+}\!\left[(U_\nu\+\psi)(\vec{r},\vec{x},\vec{X}_\nu)-\tfrac{\xi(\vec{x},\vec{X}_\nu)}{r}\right]\!=(\Gamma^{\+\nu}_{\!\mathrm{reg}}\+\xi)(\vec{x},\vec{X}_\nu),\qquad\forall\nu\in\mathcal{P}_N,
\end{equation}
which implies
\begin{equation*}
        \scalar{\mathcal{C}_\nu\+\xi}{\Gamma^{\+\nu}_{\!\mathrm{reg}}\+\xi}[\X_\nu]+\!\lim_{r\to 0^+}\scalar{\mathcal{C}_\nu\+\xi}{\!\left[\tfrac{\xi(\vec{x},\+\vec{X}_\nu)}{r}-U_\nu(\mathcal{G}^\lambda\xi+\varphi_\lambda)(\vec{r},\vec{x},\vec{X}_\nu)\-\right]}[\X_\nu]\!=0,\qquad\forall \nu\in\mathcal{P}_N.
\end{equation*}
Therefore, since $\varphi_\lambda\!\in\- H^2(\R^{3N})$, taking into account~\eqref{weakTraceId} one gets for any $\nu\in\mathcal{P}_N$
\begin{equation*}
    \scalar{\mathcal{C}_\nu\+\xi}{\Gamma^{\+\nu}_{\!\mathrm{reg}}\+\xi}[\X_\nu]\-+\!\lim_{r\to 0^+}\scalar{\mathcal{C}_\nu\+\xi}{\!\left[\tfrac{\xi(\vec{x},\+\vec{X}_\nu)}{r}-(U_\nu\+\mathcal{G}^\lambda\xi)(\vec{r},\vec{x},\vec{X}_\nu)\-\right]}[\X_\nu]\!
    -\tfrac{1}{8\pi}\scalar{\mathcal{C}_\nu\+\xi}{T_\nu\+\varphi_\lambda}[\X_\nu]\!=0,
\end{equation*}
namely, according to~\eqref{traceOperator}
\begin{equation*}
    \scalar{\mathcal{C}_\nu\+\xi}{\Gamma^{\+\nu}_{\!\mathrm{reg}}\+\xi}[\X_\nu]\-+\!\lim_{r\to 0^+}\scalar{\mathcal{C}_\nu\+\xi}{\!\left[\tfrac{\xi(\vec{x},\+\vec{X}_\nu)}{r}-(U_\nu\+\mathcal{G}^\lambda\xi)(\vec{r},\vec{x},\vec{X}_\nu)\-\right]}[\X_\nu]\!
    =\tfrac{1}{4\pi\+N(N\--1)\!}\+\scalar{\xi}{T\+\varphi_\lambda}[\X].
\end{equation*}
Thus, the boundary condition $\tilde{\Gamma}(-\lambda)\+\xi=T\varphi_\lambda$ provides the expression of the quadratic form associated to the s.a. operator $\tilde{\Gamma}(-\lambda)$, denoted by $\tilde{\Phi}^\lambda$
\begin{equation}\label{phiTildeDefinition}
    \tilde{\Phi}^\lambda[\xi]\vcentcolon=4\pi N(N\!-\-1)\-\left\{\!\scalar{\mathcal{C}_\nu\+\xi}{\Gamma^{\+\nu}_{\!\mathrm{reg}}\+\xi}[\X_\nu]\!+\!
    \lim_{r\to 0^+}\scalar{\mathcal{C}_\nu\+\xi}{\!\left[\tfrac{\xi(\vec{x},\+\vec{X}_\nu)}{r}-(U_\nu\+\mathcal{G}^\lambda\xi)(\vec{r},\vec{x},\vec{X}_\nu)\-\right]}[\X_\nu]\!\right\}\-.
\end{equation}
\comment{
whose expression coincides with~\eqref{phiDefinition}.\newline
Therefore, let $\Phi^\lambda$ be the densely defined quadratic form in $\X$ given by~\eqref{phiDefinition}.
Our goal is to show that $\Phi^\lambda$ is associated to a s.a. operator $\Gamma^\lambda$ satisfying the hypotheses of Propositions~\ref{sufficientConditionsGamma} and~\ref{boundednessFromBelow} so that a s.a. and bounded from below extension of $\dot{\mathcal{H}}_0$ can be defined by~\eqref{bcHamiltonianCharacterization}.
This s.a. extension shall be the Hamiltonian $\mathcal{H}$ described in the introduction, since, by construction, it will coincide with $\mathcal{H}_0$ on $\hilbert_N\cap H^2_0(\R^{3N}\setminus \pi)$ and any element $\psi\-\in\-\dom{\mathcal{H}}$ satisfies boundary condition~\eqref{mfBC} at least in the weak topology.
}
\begin{note}\label{resolventIdPhi}
Notice that definition~\eqref{phiTildeDefinition} satisfies for any $\lambda_1,\lambda_2>0$
\begin{align*}
    \tilde{\Phi}^{\lambda_1}[\xi]\-&-\tilde{\Phi}^{\lambda_2}[\xi]=\+8\pi\!\sum_{\sigma\+\in\+\mathcal{P}_N} \lim_{r\to 0^+}\scalar{\mathcal{C}_\sigma\+\xi}{(U_\sigma\+\mathcal{G}^{\lambda_2}\xi-U_\sigma\+\mathcal{G}^{\lambda_1}\xi)(\vec{r},\vec{x},\vec{X}_{\-\sigma})}[\X_\sigma]\\[-2.5pt]
    &=\+8\pi\!\sum_{\sigma\+\in\+\mathcal{P}_N}\-\scalar{\mathcal{C}_\sigma\+\xi}{\wlim{r\to 0^+}(U_\sigma\+\mathcal{G}^{\lambda_2}\xi-U_\sigma\+\mathcal{G}^{\lambda_1}\xi)(\vec{r},\vec{x},\vec{X}_{\-\sigma})}[\X_\sigma]\!=\-\sum_{\sigma\+\in\+\mathcal{P}_N}\- \scalar{\mathcal{C}_\sigma\+\xi}{T_\sigma(\mathcal{G}^{\lambda_2}\!-\mathcal{G}^{\lambda_1})\+\xi}[\X_\sigma]
\end{align*}
thanks to~\eqref{weakTraceId}, since $\ran{\mathcal{G}^{\lambda_2}\!-\mathcal{G}^{\lambda_1}}\-\subset\-\dom{\mathcal{H}_0}$.
Hence, due to~\eqref{traceOperator}, we have obtained
\begin{equation}\label{differenceGamma}\tilde{\Phi}^{\lambda_1}[\xi]\--\tilde{\Phi}^{\lambda_2}[\xi]=\scalar{\xi}{T\!\left(\mathcal{G}^{\lambda_2}\!-\mathcal{G}^{\lambda_1}\!\right)\mspace{-2.25mu}\xi}[\X]\+.
\end{equation}
\end{note}
We point out that expression~\eqref{phiTildeDefinition} has been obtained only by exploiting the known structure~\eqref{bcHamiltonianCharacterization} of the domain of any singular perturbation of $\mathcal{H}_0\+$ that is supported on $\pi$.
None of the properties of the specific $\tilde{\Gamma}$ has been used, beside the desired boundary condition~\eqref{mfBC}.
Indeed, one could consider the previous argument as a heuristic motivation to start the rigorous analysis by studying the quadratic form defined by~\eqref{phiTildeDefinition}.
If one proves that $\tilde{\Phi}^\lambda$ is closed for any $\lambda\->\-0$ and that it is coercive for $\lambda$ large enough, then there exists an actual s.a. operator associated with $\tilde{\Phi}^\lambda$ which is satisfying Propositions~\ref{sufficientConditionsGamma} and~\ref{boundednessFromBelow} (see Remark~\ref{resolventIdPhi}).\newline
This means that the coercivity of $\tilde{\Phi}^\lambda$ for $\lambda$ large is basically the only crucial ingredient needed in order to define the s.a. and lower-bounded zero-range Hamiltonian $\mathcal{H}$ satisfying boundary condition~\eqref{mfBC}.

\n In conclusion, let us decompose $\tilde{\Phi}^\lambda$ as follows
\begin{equation}\label{phiTildeDecomposition}
    \tilde{\Phi}^\lambda\!=4\pi N(N\!-\-1)\!\left(\tilde{\Phi}^{\lambda}_{\mathrm{diag}}\!+\tilde{\Phi}^{\lambda}_{\mathrm{off;\+0}}\-+\tilde{\Phi}^{\lambda}_{\mathrm{off;\+1}}\-+\tilde{\Phi}_{\mathrm{reg}}\right)\!,
\end{equation}
where, considering $\sigma\-\in\-\mathcal{P}_N$, we have defined
\begin{subequations}\label{phiTildeComponents}
\begin{gather}
    \tilde{\Phi}^{\lambda}_{\mathrm{diag}}[\xi]\vcentcolon=\lim_{r\to 0^+}\scalar{\mathcal{C}_\sigma\+\xi}{\!\left[\tfrac{\xi(\vec{x},\vec{X}_{\-\sigma})}{r}-(U_\sigma\+\mathcal{G}^\lambda_\sigma\+\xi)(\vec{r},\vec{x},\vec{X}_{\-\sigma})\right]\-}[\X_\sigma],\label{diagPhiDef}\\
    \tilde{\Phi}^{\lambda}_{\mathrm{off;}\+i}[\xi]\vcentcolon=-\mspace{-12mu}\sum_{\substack{\nu\,\in\+\mathcal{P}_N\+:\\|\+\sigma\+\cap\,\nu\+|\,=\,i}}\!\-\lim_{r\to 0^+}\scalar{\mathcal{C}_\sigma\+\xi}{(U_\sigma\+\mathcal{G}^\lambda_\nu\+\xi)(\vec{r},\vec{x},\vec{X}_{\-\sigma})}[\X_\sigma],\quad i\-\in\-\{0,1\},\label{offsPhiDef}\\[-10pt]
    \tilde{\Phi}_{\mathrm{reg}}[\xi]\vcentcolon=\scalar{\mathcal{C}_\sigma\+\xi}{\Gamma^{\+\sigma}_{\!\mathrm{reg}}\+\xi}[\X_\sigma].\label{regPhiDef}
\end{gather}
\end{subequations}
Definitions~\eqref{phiTildeComponents} are shown to coincide with the ones given by identities~\eqref{phiComponents} owing to Proposition~\ref{diagonalSingularityCancellation}.
\begin{note}\label{tuningConstant}
    We stress that the cancellation of the singular leading order obtained in Proposition~\ref{diagonalSingularityCancellation} is possible thanks to the properly tuned constant $8\pi$ in definition~\eqref{singleTraceOperator}.
\end{note}
\begin{note}\label{GammaRemark}
    According to our construction there holds
    \begin{equation}
        \tilde{\Gamma}(-\lambda)=\Gamma^\lambda=8\pi\!\sum_{\sigma\,\in\,\mathcal{P}_N}\!\adj{\mathcal{C}}_\sigma(\Gamma^{\+\sigma,\+\lambda}\-+\Gamma^{\+\sigma}_{\!\mathrm{reg}}),
    \end{equation}
    with $\Gamma^{\+\sigma,\+\lambda}$ and $\Gamma^{\+\sigma}_{\!\mathrm{reg}}$ given by equations~\eqref{GammaTMSHyp} and~\eqref{regGammaHyp}, respectively.
\end{note}
}

\section{Properties of the Potential}\label{appB}

\comment{Let us use this space to state the following technical facts.

\begin{prop}\label{weakLimitTraceOperatorPro}
    Given $\maps{f}{\R^m\-\times\R^n;\C}$, suppose $f\in H^{\frac{m}{2}+s}(\R^{n+m})$ for some $s>0$.
    Then
    \begin{equation*}
        f(\vec{x},\vec{y})|_{\vec{x}=\vec{0}}=\wlim{\vec{x}\to\vec{0}} f(\vec{x},\vec{y}),\qquad\text{for almost every }\,\vec{y}\in\R^n.
    \end{equation*}
\begin{proof}
    First, for any $s>0$ denote by $\maps{\tau_0}{H^{\frac{m}{2}+s}(\R^{n+m});H^s(\R^n)}$ the trace operator given by
    \begin{equation}
        (\widehat{\tau_0 f})(\vec{p})=\frac{1}{(2\pi)^{\frac{m}{2}}\!\-}\integrate[\R^m]{\FT{f}(\vec{q},\vec{p});\!\-d\vec{q}}.
    \end{equation}
    We want to show that
    \begin{subequations}\begin{equation}
        \lim_{\vec{x}\to\vec{0}}\integrate[\R^n]{\conjugate*{g(\vec{y})}f(\vec{x},\vec{y});\!d\vec{y}}=\!\integrate[\R^n]{\conjugate*{g(\vec{y})}\,(\tau_0 f)(\vec{y});\!d\vec{y}},\qquad\forall g\in\Lp{2}[\R^n],
    \end{equation}
    or, equivalently
    \begin{equation}
        \lim_{\vec{x}\to\vec{0}}\integrate[\R^n]{\conjugate*{\FT{g}(\vec{p})};\!d\vec{p}}\!\integrate[\R^n]{\frac{e^{-i\vec{p}\+\cdot\+\vec{y}}\!\!}{(2\pi)^{\frac{n}{2}}\!\-}\;f(\vec{x},\vec{y});\!d\vec{y}}=\!\integrate[\R^n]{\conjugate*{\FT{g}(\vec{p})}\,(\widehat{\tau_0 f})(\vec{p});\!d\vec{p}},\qquad\forall g\in\Lp{2}[\R^n].\label{weakLimitWTS}
    \end{equation}\end{subequations}
    Clearly, one has
    \begin{equation}\label{semiFT/semiAntiFT}
        \integrate[\R^n]{\frac{e^{-i\vec{p}\+\cdot\+\vec{y}}\!\!}{(2\pi)^{\frac{n}{2}}\!\-}\;f(\vec{x},\vec{y});\!d\vec{y}}=\!\integrate[\R^m]{\frac{e^{i\vec{x}\+\cdot\+\vec{q}}\!\!}{(2\pi)^{\frac{m}{2}}\!\-}\;\FT{f}(\vec{q},\vec{p});\!\-d\vec{q}}.
        \end{equation}
    We claim that
    \begin{equation}
        \left\lvert\integrate[\R^{n+m}]{\conjugate*{\FT{g}(\vec{p})}\FT{f}(\vec{q},\vec{p});\mspace{-20mu}d\vec{p}\+d\vec{q}}\right\rvert\leq\,\sqrt{\frac{\pi^{\frac{m}{2}}\Gamma(s)}{\Gamma\!\left(s+\frac{m}{2}\right)\!}} \,\norm{g}[2]\norm{f}[H^{m/2\, +\+s}(\R^{n+m})]\!.
    \end{equation}
    Indeed,
    \begin{align*}
        \left(\integrate[\R^{n+m}]{\abs{\FT{g}(\vec{p})}\abs{\FT{f}(\vec{q},\vec{p})};\mspace{-20mu}d\vec{p}\+d\vec{q}}\right)^{\!2}\-\!\!&\leq\! \integrate[\R^{n+m}]{\frac{\abs{\FT{g}(\vec{p})}^2}{(p^2+q^2+1)^{\frac{m}{2}+s}};\mspace{-20mu}d\vec{p}\+d\vec{q}}\!\integrate[\R^{n+m}]{(p^2+q^2+1)^{\frac{m}{2}+s}\,\abs{\FT{f}(\vec{q},\vec{p})}^2;\mspace{-20mu}d\vec{p}\+d\vec{q}}\\
        &\leq \frac{2\pi^{\frac{m}{2}}}{\Gamma\!\left(\frac{m}{2}\right)\!}\,\norm{f}[H^{m/2\,+\+s}(\R^{n+m})]^2\!\-\integrate[\R^n]{\abs{\FT{g}(\vec{p})}^2;\!d\vec{p}}\!\-\integrate[0;\pInfty]{\frac{q^{m-1}}{(q^2+1)^{\frac{m}{2}+s}};\nquad dq}.
    \end{align*}
    Therefore, exploiting~\eqref{semiFT/semiAntiFT} and the dominated convergence theorem, one has
    \begin{equation*}
        \lim_{\vec{x}\to\vec{0}}\integrate[\R^n]{\conjugate*{\FT{g}(\vec{p})};\!d\vec{p}}\!\integrate[\R^n]{\frac{e^{-i\vec{p}\+\cdot\+\vec{y}}\!\!}{(2\pi)^{\frac{n}{2}}\!\-}\;f(\vec{x},\vec{y});\!d\vec{y}}=\!\integrate[\R^n]{\conjugate*{\FT{g}(\vec{p})};\!d\vec{p}}\;\frac{1}{(2\pi)^{\frac{m}{2}}\!\-}\integrate[\R^m]{\FT{f}(\vec{q},\vec{p});\!\-d\vec{q}}.
    \end{equation*}
    and equation~\eqref{weakLimitWTS} is proven.
    
\end{proof}
\end{prop}
}

\n Here we state a useful property of $G^\lambda$ defined by~\eqref{greenLaplace}.

\begin{prop}\label{greenTransformed}
    Let $G^\lambda$ be defined by~\eqref{greenLaplace} and $\sigma=\{i,j\}$.
    Then,
    \begin{equation*}
        \integrate[\R^{3(N-1)}]{e^{-i\vec{q}\+\cdot\vec{y}-\+i\+\vec{Q}_\sigma\cdot\+\vec{Y}_{\!\!\sigma}};\mspace{-33mu}d\vec{y}d\vec{Y}_{\!\!\sigma}}\+G^\lambda\bigg(\!\!
        \begin{array}{r c l}
            \vec{x}_i\+,&\!\!\vec{x}_j\+, &\!\!\vec{X}_{\-\sigma}\\[-2.5pt]
            \vec{y},&\!\!\vec{y},&\!\!\vec{Y}_{\!\!\sigma}
        \end{array}\!\!\!\bigg)\!=\frac{\:e^{-\sqrt{\-\frac{q^2\!}{2}\++\,Q^2_\sigma\++\+\lambda\+}\:\frac{\abs{\vec{x}_i-\+\vec{x}_j}}{\sqrt{2}}}\nquad}{8\pi\,\abs{\vec{x}_i\--\vec{x}_j}}\;e^{-i\vec{q}\,\cdot\+\frac{\vec{x}_i+\+\vec{x}_j}{2}-\+i\+\vec{Q}_\sigma\cdot\+\vec{X}_{\-\sigma}},\quad\vec{x}_i\neq\vec{x}_j\+.
    \end{equation*}
\begin{proof}
    Adopting the change of coordinates $\vec{z}=\vec{y}-\vec{x}_j\+$, $\vec{Z}_{\-\sigma}=\vec{Y}_{\!\!\sigma}\--\-\vec{X}_{\-\sigma}\+$, one has
    \begin{align*}
        \integrate[\R^{3(N-1)}]{e^{-i\vec{q}\+\cdot\vec{y}-\+i\+\vec{Q}_\sigma\cdot\+\vec{Y}_{\!\!\sigma}};\mspace{-33mu}d\vec{y}d\vec{Y}_{\!\!\sigma}}\+G^\lambda\bigg(\!\!
        \begin{array}{r c l}
            \vec{x}_i\+,&\!\!\vec{x}_j\+, &\!\!\vec{X}_{\-\sigma}\\[-2.5pt]
            \vec{y},&\!\!\vec{y},&\!\!\vec{Y}_{\!\!\sigma}
        \end{array}\!\!\!\bigg)\!=e^{-i\vec{q}\,\cdot\+\vec{x}_j-\+i\+\vec{Q}_\sigma\cdot\+\vec{X}_{\-\sigma}}\+\times\\
        \times\!\-\integrate[\R^{3(N-1)}]{e^{-i\vec{q}\+\cdot\vec{z}-\+i\+\vec{Q}_\sigma\cdot\+\vec{Z}_{\-\sigma}};\mspace{-33mu}d\vec{z}d\vec{Z}_{\-\sigma}}G^\lambda\bigg(\!\!\-
        \begin{array}{c c c}
            \vec{x}_i\--\vec{x}_j\+,&\!\!\vec{0}\+, &\!\!\vec{0}\\[-5pt]
            \vec{z},&\!\!\vec{z},&\!\!\vec{Z}_{\-\sigma}
        \end{array}\!\!\-\bigg).
    \end{align*}
    Next, take into account the following identities
    \begin{align*}
        G^\lambda\bigg(\!\!\-
        \begin{array}{c c c}
            \vec{x}_i\--\vec{x}_j\+,&\!\!\vec{0}\+, &\!\!\vec{0}\\[-5pt]
            \vec{z},&\!\!\vec{z},&\!\!\vec{Z}_{\-\sigma}
        \end{array}\!\!\-\bigg)\!
        &=\frac{\!\!\-1}{\,(2\pi)^{3N}\!\-}\-\integrate[\R^{3N}]{\frac{e^{i\+\vec{p}_i\+\cdot\,(\vec{z}-\vec{x}_i+\vec{x}_j)+\+i\+\vec{p}_j\+\cdot\+\vec{z}\,+\,i\+\vec{P}_{\!\sigma}\cdot\+\vec{Z}_{\-\sigma}}}{p_i^2+p_j^2+P^2_{\!\sigma}+\lambda};\mspace{-10mu}d\vec{p}_i \+d\vec{p}_j d\mspace{-0.75mu}\vec{P}_{\!\sigma}}\\
        &=\frac{\!\!\-1}{\,(2\pi)^{3N}\!\-}\-\integrate[\R^{3(N-1)}]{;\mspace{-33mu}d\vec{p}\+d\mspace{-0.75mu}\vec{P}_{\!\sigma}}\!\!\integrate[\R^3]{\frac{e^{i\+\vec{p}\,\cdot\left(\-\sqrt{2}\+\vec{z}\+-\+\frac{\vec{x}_i-\+\vec{x}_j}{\sqrt{2}}\-\right)+\+i\+\vec{q}\mspace{2.25mu}\cdot\left(\frac{\vec{x}_i-\+\vec{x}_j}{\sqrt{2}}\-\right)+\,i\+\vec{P}_{\!\sigma}\+\cdot\+\vec{Z}_{\-\sigma}}}{p^2+q^2+P_{\!\sigma}^2+\lambda};\-d\vec{q}}\\
        &=\frac{2\sqrt{2}\+\pi^2}{(2\pi)^{3N}\abs{\vec{x}_i\--\vec{x}_j}}\-\integrate[\R^{3(N-1)}]{e^{i\+\vec{p}\,\cdot\left(\-\sqrt{2}\+\vec{z}\+-\+\frac{\vec{x}_i-\+\vec{x}_j}{\sqrt{2}}\-\right)+\,i\+\vec{P}_{\!\sigma}\+\cdot\+\vec{Z}_{\-\sigma}}e^{-\sqrt{p^2+P_{\!\sigma}^2\,+\+\lambda\+}\,\frac{\-\abs{\vec{x}_i-\+\vec{x}_j}}{\sqrt{2}}};\mspace{-33mu}d\vec{p}\+d\mspace{-0.75mu}\vec{P}_{\!\sigma}}\!\-,
    \end{align*}
    where we have set $(\vec{p}_i\+,\,\vec{p}_j)=\left(\frac{\vec{p}\,-\,\vec{q}}{\sqrt{2}},\frac{\vec{p}\,+\,\vec{q}}{\sqrt{2}}\right)$ and exploited
    \begin{equation}\label{fourierTransformYukawa}
        \integrate[\R^3]{\frac{e^{i\,\vec{k}\,\cdot\, \vec{x}}}{k^2+a^2};\-d\vec{k} }=\frac{\,2\pi^2\!\-}{\abs{\vec{x}}}\,e^{-a\,\abs{\vec{x}}}\:,\quad \forall \vec{x}\neq \vec{0}\+,\,a>0.
    \end{equation}
    Thus, if we define the function
    \begin{equation}\label{hiddenBesselTransform}
        \begin{split}
            &\maps{\FT{f}^\lambda_x}{\R^3\!\otimes\-\R^{3(N-2)};\Rplus},\qquad x,\lambda>0,\\
            &\qquad(\vec{p},\vec{P})\longmapsto \frac{\:e^{-\sqrt{p^2+P^2+\lambda\+}\,x}\!\-}{x},
        \end{split}
    \end{equation}
    we can write    \begin{equation}\label{almostManifestBesselTransform}
        G^\lambda\bigg(\!\!
        \begin{array}{c c c}
            \vec{x}_i\--\vec{x}_j\+,&\!\!\vec{0}\+, &\!\!\vec{0}\\[-5pt]
            \vec{z},&\!\!\vec{z},&\!\!\vec{Z}_{\-\sigma}
        \end{array}\!\!\-\bigg)=\+\frac{2\pi^2}{\,(2\pi)^{\-\frac{3}{2}(N+1)}\!\-}\; f^\lambda_{\frac{\abs{\vec{x}_i-\+\vec{x}_j}}{\sqrt{2}}}\!\left(\-\sqrt{2}\+\vec{z}-\tfrac{\vec{x}_i-\+\vec{x}_j}{\sqrt{2}},\vec{Z}_{\sigma}\-\right)\!.
    \end{equation}
    Therefore,
    \begin{align*}
        \integrate[\R^{3(N-1)}]{e^{-i\vec{q}\+\cdot\vec{z}-\+i\+\vec{Q}_\sigma\cdot\+\vec{Z}_{\-\sigma}};\mspace{-33mu}d\vec{z}d\vec{Z}_{\-\sigma}}\,G^\lambda\bigg(\!\!
        \begin{array}{c c c}
            \vec{x}_i\--\vec{x}_j\+,&\!\!\vec{0}\+, &\!\!\vec{0}\\[-5pt]
            \vec{z},&\!\!\vec{z},&\!\!\vec{Z}_{\-\sigma}
        \end{array}\!\!\-\bigg)\!=\;&\frac{\pi^2\+e^{-i\vec{q}\,\cdot\+\frac{\vec{x}_i-\+\vec{x}_j}{2}}}{\!\-\sqrt{2}\+(2\pi)^{\-\frac{3}{2}(N+1)}\!\-}\,\times\\ &\times\!\-\integrate[\R^{3(N-1)}]{e^{-i\vec{q}\+\cdot\frac{\vec{s}}{\sqrt{2}}-\+i\+\vec{Q}_\sigma\cdot\+\vec{Z}_{\-\sigma}}f^\lambda_{\frac{\abs{\vec{x}_i-\+\vec{x}_j}}{\sqrt{2}}}\!\left(\vec{s},\vec{Z}_{\-\sigma}\-\right);\mspace{-33mu}d\vec{s}d\vec{Z}_{\-\sigma}}\\
        =\,\frac{\pi^2\+e^{-i\vec{q}\,\cdot\+\frac{\vec{x}_i-\+\vec{x}_j}{2}}}{\!\-\sqrt{2}\+(2\pi)^3\!\-}\,\FT{f}^\lambda_{\frac{\abs{\vec{x}_i-\+\vec{x}_j}}{\sqrt{2}}}\!\left(\tfrac{\vec{q}}{\sqrt{2}},\vec{Q}_\sigma\!\right)\!&=\frac{\:e^{-\sqrt{\-\frac{q^2\!}{2}\++\,Q_\sigma^2\++\+\lambda\+}\:\frac{\abs{\vec{x}_i-\+\vec{x}_j}}{\sqrt{2}}}\!\!\!}{8\pi\,\abs{\vec{x}_i\--\vec{x}_j}}\:e^{-i\vec{q}\,\cdot\+\frac{\vec{x}_i-\+\vec{x}_j}{2}}.
    \end{align*}
    
\end{proof}
\end{prop}
\n In particular, notice that Proposition~\ref{greenTransformed} implies
\begin{equation}\label{greenIntegrated}
        \integrate[\R^{3(N-1)}]{;\mspace{-33mu}d\vec{y}d\vec{Y}_{\!\!\sigma}}\+G^\lambda\bigg(\!\!
        \begin{array}{r c l}
            \vec{x}_i\+,&\!\!\vec{x}_j\+, &\!\!\vec{X}_{\-\sigma}\\[-2.5pt]
            \vec{y},&\!\!\vec{y},&\!\!\vec{Y}_{\!\!\sigma}
        \end{array}\!\!\!\bigg)=\frac{\:e^{-\sqrt{\-\frac{\lambda}{2}\+}\:\abs{\vec{x}_i-\+\vec{x}_j}}\!\!\!}{8\pi\,\abs{\vec{x}_i\--\vec{x}_j}},\qquad\vec{x}_i\neq\vec{x}_j\+.
\end{equation}
In conclusion, we provide one last property concerning the potential $\mathcal{G}^\lambda$.

\begin{prop}\label{diagonalSingularityCancellation}
Given $U_\sigma$ defined by~\eqref{unitaryCoordinate}, let $\xi\!\in\! \dom{\Phi^\lambda}$ and $\Phi^\lambda$ the hermitian quadratic form in $\hilbert*_{N-1}$ defined by~(\ref{phiDefinition},~\ref{phiComponents}).
Then, one has for all $\lambda\->\-0$
\begin{subequations}
    \begin{gather}
        \Phi^\lambda_{\mathrm{diag}}[\xi]=\lim_{\vec{r}\to \vec{0}} \integrate[\R^{3(N-1)}]{\!\left[\tfrac{\abs{\xi(\vec{x},\,\vec{X}_{\-\sigma})}^2\!\!}{r}-\+\conjugate*{\xi(\vec{x},\vec{X}_{\-\sigma})\-}\,(U_\sigma\+\mathcal{G}^\lambda_\sigma\+\xi)(\vec{r},\vec{x},\vec{X}_{\-\sigma}\-)\right]\!;\mspace{-33mu}d\vec{x}d\vec{X}_{\-\sigma}}\-,\label{itIsDiag}\\[7.5pt]
        \Phi^\lambda_{\mathrm{off};\+1}[\xi]=-\!\lim_{\vec{r}\to\vec{0}}\integrate[\R^{3(N-1)}]{\conjugate*{\xi(\vec{x},\vec{X}_{\-\sigma})\-}\,\big(U_\sigma\nquad\;{\textstyle \sum\limits_{\substack{\nu\+\in\+\mathcal{P}_N\+:\\|\nu\+\cap\+\sigma|\+=\+1}}}\!\!\!\!\mathcal{G}^\lambda_\nu\+\xi\big)(\vec{r},\vec{x},\vec{X}_{\-\sigma});\mspace{-33mu}d\vec{x}d\vec{X}_{\-\sigma}},\label{itIsOff1}\\
        \Phi^\lambda_{\mathrm{off};\+0}[\xi]=-\!\lim_{\vec{r}\to\vec{0}}\integrate[\R^{3(N-1)}]{\conjugate*{\xi(\vec{x},\vec{X}_{\-\sigma})\-}\,\big(U_\sigma\nquad\;{\textstyle \sum\limits_{\substack{\nu\+\in\+\mathcal{P}_N\+:\\|\nu\+\cap\+\sigma|\+=\+0}}}\!\!\!\!\mathcal{G}^\lambda_\nu\+\xi\big)(\vec{r},\vec{x},\vec{X}_{\-\sigma});\mspace{-33mu}d\vec{x}d\vec{X}_{\-\sigma}}.\label{itIsOff0}
    \end{gather}
\end{subequations}
\begin{proof}
    We denote by $L^\lambda_{\xi,\+i}\+, \,i\-\in\-\{0,1,2\}$ the limits we want to compute.\newline
    Regarding~\eqref{itIsDiag}, in light of equation~\eqref{greenIntegrated}, we have
    \begin{equation}\begin{split}
        L^\lambda_{\xi,\+2}\-=\-\lim_{\vec{r}\to \vec{0}}&\-\integrate[\R^{3(N-1)}]{\conjugate*{\xi(\vec{x},\vec{X}_{\-\sigma})}\+\bigg[\frac{1\--e^{-\sqrt{\-\frac{\lambda}{2}}\,r}\!\-}{r}\:\xi(\vec{x},\vec{X}_{\-\sigma});\mspace{-33mu}d\vec{x}d\vec{X}_{\-\sigma}}\++\\
        &-8\pi\!\-\integrate[\R^{3(N-1)}]{[\xi(\vec{y},\vec{Y}_{\!\!\sigma})\--\xi(\vec{x},\vec{X}_{\-\sigma})];\mspace{-33mu}d\vec{y}d\vec{Y}_{\!\!\sigma}}\,G^\lambda\bigg(\!\!
        \begin{array}{c c l}
            \vec{x}+\tfrac{\vec{r}}{2}\+, & \!\!\vec{x}-\tfrac{\vec{r}}{2}\+, &\!\!\! \vec{X}_{\-\sigma} \\[-2.5pt]
            \vec{y}, & \!\!\vec{y}, &\!\!\! \vec{Y}_{\!\!\sigma}
        \end{array}\!\!\!\bigg)\!\bigg].
        \end{split}
    \end{equation}
    The first term has the integrable majorant $\sqrt{\frac{\lambda}{2}\+}\+\abs{\xi(\vec{x},\vec{X}_{\-\sigma})}^2$, therefore
    \begin{equation*}
       \begin{split} L^\lambda_{\xi,\+2}\-=\-\sqrt{\tfrac{\lambda}{2}}\norm{\xi}[\hilbert*_{N-1}]^2\!-8\pi\!\lim_{r\to 0^+}\-\integrate[\R^{3(N-1)}]{\conjugate*{\xi(\vec{x},\vec{X}_{\-\sigma})};\mspace{-33mu}d\vec{x}d\vec{X}_{\-\sigma}}\!\-\integrate[\R^{3(N-1)}]{[\xi(\vec{y},\vec{Y}_{\!\!\sigma})\--\xi(\vec{x},\vec{X}_{\-\sigma})];\mspace{-33mu}d\vec{y}d\vec{Y}_{\!\!\sigma}}\+\times\\[-5pt]
        \times\+ G^\lambda\bigg(\!\!\-
        \begin{array}{c c l}
            \vec{x}+\tfrac{\vec{r}}{2}\+, & \!\!\vec{x}-\tfrac{\vec{r}}{2}\+, &\!\!\! \vec{X}_{\-\sigma} \\[-2.5pt]
            \vec{y}, & \!\!\vec{y}, &\!\!\! \vec{Y}_{\!\!\sigma}
        \end{array}\!\!\!\bigg) .
        \end{split}
    \end{equation*}
    Notice that the quantity $\abs{\vec{x}+\tfrac{r}{2}-\vec{y}}^2\-+\abs{\vec{x}-\frac{\vec{r}}{2}-\vec{y}}^2\-+\abs{\vec{X}_{\-\sigma}-\vec{Y}_{\!\!\sigma}}^2$ is symmetric under the exchange $(\vec{x},\vec{X}_{\-\sigma})\-\longleftrightarrow\-(\vec{y},\vec{Y}_{\!\!\sigma})$.
    Hence,
    \begin{equation*}
       \begin{split} L^\lambda_{\xi,\+2}\-=\-\sqrt{\tfrac{\lambda}{2}}\norm{\xi}[\hilbert*_{N-1}]^2\!+4\pi\!\lim_{r\to 0^+}\-\integrate[\R^{3(N-1)}]{;\mspace{-33mu}d\vec{x}d\vec{X}_{\-\sigma}}\!\!\!\integrate[\R^{3(N-1)}]{\abs{\xi(\vec{y},\vec{Y}_{\!\!\sigma})\--\xi(\vec{x},\vec{X}_{\-\sigma})}^2;\mspace{-33mu}d\vec{y}d\vec{Y}_{\!\!\sigma}}\, G^\lambda\bigg(\!\!\-
        \begin{array}{c c l}
            \vec{x}+\tfrac{\vec{r}}{2}\+, & \!\!\vec{x}-\tfrac{\vec{r}}{2}\+, &\!\!\! \vec{X}_{\-\sigma} \\[-2.5pt]
            \vec{y}, & \!\!\vec{y}, &\!\!\! \vec{Y}_{\!\!\sigma}
        \end{array}\!\!\!\bigg).
        \end{split}
    \end{equation*}
    Then, exploiting the decrease of $G^\lambda$ and according to the elementary inequality $$\abs{\vec{x}+\tfrac{\vec{r}}{2}-\vec{y}}^2\-+\abs{\vec{x}-\tfrac{\vec{r}}{2}-\vec{y}}^2\-\geq\tfrac{1}{2}\+\abs{\vec{x}+\tfrac{\vec{r}}{2}-\vec{y}+\vec{x}-\tfrac{\vec{r}}{2}-\vec{y}}^2\-=2\+\abs{\vec{x}-\vec{y}}^2,$$
    one gets
    \begin{equation*}
        G^\lambda\bigg(\!\!\-
        \begin{array}{c c l}
            \vec{x}+\tfrac{\vec{r}}{2}\+, & \!\!\vec{x}-\tfrac{\vec{r}}{2}\+, &\!\!\! \vec{X}_{\-\sigma} \\[-2.5pt]
            \vec{y}, & \!\!\vec{y}, &\!\!\! \vec{Y}_{\!\!\sigma}
        \end{array}\!\!\!\bigg)
        \-\leq \+G^\lambda\bigg(\!\!
            \begin{array}{r c l}
           \vec{x}, & \!\!\vec{x}, &\!\!\! \vec{X}_{\-\sigma} \\[-2.5pt]
           \vec{y}, & \!\!\vec{y}, &\!\!\! \vec{Y}_{\!\!\sigma}
        \end{array}\!\!\!\bigg),\qquad\forall\vec{r}\-\in\-\R^3.
    \end{equation*}
    This means that the limit can be computed inside the integral, and the result is obtained because of identity~\eqref{diagPhiPos}.\newline
    Next, we focus on proving~\eqref{itIsOff1}.
    Due to the unitarity of the Fourier transform, one has
    \begin{equation*}
        L^\lambda_{\xi,\+1}=-\-\lim_{\vec{r}\to\vec{0}}\scalar{\FT{\xi}}{\widehat{U_\sigma\nquad\sum\limits_{\substack{\nu\+\in\+\mathcal{P}_N\+:\\ |\nu\+\cap\+\sigma|\+=\+1}} \!\!\!\!\mathcal{G}^\lambda_\nu\+\xi}}[\hilbert*_{N-1}],
    \end{equation*}
    hence, due to~\eqref{unitaryCoordinateTransformed} and~\eqref{potentialFourier},
    \begin{align*}
        L^\lambda_{\xi,\+1}&=-\tfrac{2\+(N\--2)}{\pi^2}\!\lim_{\vec{r}\to\vec{0}}\integrate[\R^{3(N-1)}]{\conjugate*{\FT{\xi}(\vec{p},\vec{p}'\-,\vec{P})};\mspace{-33mu}d\vec{p}\+d\vec{p}'\-d\mspace{-0.75mu}\vec{P}}\!\integrate[\R^3]{\frac{\FT{\xi}(\vec{p}\--\-\vec{q}+\vec{p}'\-,\vec{q},\vec{P})\+\cos\!\left[\vec{r}\!\cdot\!\left(\frac{\vec{p}}{2}\--\vec{q}\right)\-\right]\!}{\abs{\vec{p}\--\-\vec{q}}^2+q^2\-+p{\+'\+}^2\-+\-P^2\-+\lambda};\-d\vec{q}}\\
        &=-\tfrac{2\+(N\!-\-2)}{\pi^2}\!\lim_{\vec{r}\to\vec{0}}\integrate[\R^{3N}]{\cos\!\left(\-\vec{r}\!\cdot\-\tfrac{\vec{p}_1-\+\vec{p}_2}{2}\!\right)\frac{\conjugate*{\FT{\xi}(\vec{p}_1\!+\vec{p}_2,\vec{p}_3,\vec{P})\-}\:\FT{\xi}(\vec{p}_1\!+\vec{p}_3,\vec{p}_2,\vec{P})}{p_1^2+p_2^2+p_3^2+\-P^2\-+\lambda};\mspace{-10mu} d\vec{p}_1d\vec{p}_2\+d\vec{p}_3 d\mspace{-0.75mu}\vec{P}}.
    \end{align*}
    Now it is plain to see that we can use the dominated convergence theorem in order to recover~\eqref{off1Phi}.\newline
    Similarly, concerning~\eqref{itIsOff0} one has
    \begin{align*}
        L^\lambda_{\xi,\+0}&=-\-\lim_{\vec{r}\to\vec{0}}\scalar{\FT{\xi}}{\widehat{U_\sigma\nquad\sum\limits_{\substack{\nu\+\in\+\mathcal{P}_N\+:\\|\nu\+\cap\+\sigma|\+=\+0}}\!\!\!\! \mathcal{G}^\lambda_\nu\+\xi}}[\hilbert*_{N-1}]\\[-5pt]
        &=-\tfrac{(N\--2)(N\--3)}{2\pi^2}\!\lim_{\vec{r}\to\vec{0}}\integrate[\R^{3(N-1)}]{\conjugate*{\FT{\xi}(\vec{p},\vec{p}'\-,\vec{p}''\!,\vec{P})};\mspace{-33mu}d\vec{p}\+d\vec{p}'\-d\vec{p}''\!d\mspace{-0.75mu}\vec{P}}\!\integrate[\R^3]{\frac{\FT{\xi}(\vec{p}'\-+\vec{p}''\!,\vec{p}\--\-\vec{q},\vec{q},\vec{P})\+\cos\!\left[\vec{r}\!\cdot\!\left(\frac{\vec{p}}{2}\--\vec{q}\right)\-\right]\!}{\abs{\vec{p}\--\-\vec{q}}^2\-+q^2\-+p{\+'\+}^2\-+p{\+''\+}^2\-+\-P^2+\lambda};\-d\vec{q}}\\
        &=-\tfrac{(N\--2)(N\--3)}{2\pi^2}\!\lim_{\vec{r}\to\vec{0}}\integrate[\R^{3N}]{\cos\!\left(\-\vec{r}\!\cdot\-\tfrac{\vec{p}_3-\+\vec{p}_4}{2}\!\right)\frac{\conjugate*{\FT{\xi}(\vec{p}_1\!+\vec{p}_2,\vec{p}_3,\vec{p}_4,\vec{P})\-}\:\FT{\xi}(\vec{p}_3\-+\vec{p}_4,\vec{p}_1,\vec{p}_2,\vec{P})}{p_1^2+p_2^2+p_3^2+p_4^2+\-P^2\-+\lambda};\mspace{-10mu}d\vec{p}_1d\vec{p}_2\+d\vec{p}_3 d\vec{p}_4 d\mspace{-0.75mu}\vec{P}},
    \end{align*}
    and therefore one can recognize in the last expression identity~\eqref{off0Phi}.
    
    \end{proof}
\end{prop}

\end{document}